\documentclass[reprint,amsfonts, amssymb, amsmath,  showkeys,prx, superscriptaddress, twocolumn,longbibliography,nofootinbib,aps]{revtex4-2}

\usepackage{silence}
\usepackage{graphicx} % Required for inserting images
\usepackage{tabularx}
\usepackage{amsmath}
\usepackage{amsthm}
\usepackage{amssymb}
\usepackage{bm}
\usepackage{braket}
\usepackage[left=2cm, right=2cm]{geometry}
\usepackage[colorlinks=true,citecolor=blue,linkcolor=magenta]{hyperref}
\usepackage{xcolor}
\usepackage[at]{easylist}
\usepackage{ulem}
\usepackage[nolist]{acronym}
\usepackage{xparse}
\usepackage{booktabs}
\usepackage{multirow}
\usepackage{mleftright}
\usepackage{csquotes}
\usepackage{enumitem}
\usepackage{pifont}
\usepackage{makecell}
\usepackage{varwidth}
\usepackage{placeins}  % to use \FloatBarrier
\usepackage{etoolbox}
\usepackage{layouts}
\usepackage{ragged2e}
\usepackage{mathtools}

\AtBeginEnvironment{easylist}{\ListProperties(Start1=1)}

\newcommand{\E}{\operatorname{\mathbb{E}}}
\newcommand{\tr}{{\rm Tr}}
\newcommand{\G}{\operatorname{\mathbb{G}}}
\newcommand{\SO}{\operatorname{\mathbb{SO}}}
\newcommand{\SPIN}{\operatorname{\mathbb{SPIN}}}
\newcommand{\SP}{\operatorname{\mathbb{SP}}}
\newcommand{\Cov}{\operatorname{\mathrm{Cov}}}
\renewcommand{\O}{O}  % omit hat
\newcommand{\cmark}{\ding{51}}
\newcommand{\xmark}{\ding{55}}

\newcommand{\omark}{\textbf{\texttt{o}}}

\newcommand{\Var}{{\rm Var}}
\newcommand{\Tr}{{\rm Tr}}

\renewcommand{\vec}[1]{\boldsymbol{#1}}

\newcommand{\id}{\openone}
\newcommand{\ad}{^{\dagger}}

\newcommand{\shotsobs}{n_{\mathrm{s}}^{\mathrm{obs}}}
\newcommand{\shotsker}{n_{\mathrm{s}}^{\kappa}}

\newcommand{\CapRho}{\mathrm{P}} \newcommand{\CapRhoBold}{\bm{P}} \newcommand{\CapRhoBB}{\mathbb{P}}
\usepackage{customrhos}

\newcommand{\idxset}{\mathcal{V}}
\newcommand{\tobs}{\bm{O}}
\newcommand{\tobsmat}{O}
\newcommand{\tobsvec}{\vec{o}}
\newcommand{\tobselem}{o}
\newcommand{\tobsbb}{\mathbb{O}}
\newcommand{\tstate}{\CapRhoBold}
\newcommand{\tstatemat}{\CapRho}
\newcommand{\tstatevec}{\vec{\rho}}
\newcommand{\tstateelem}{\rho}
\newcommand{\tstatebb}{\CapRhoBB}
\newcommand{\Sperf}{S_{\mathrm{pair}}}
\newcommand{\Simp}{S_{\neg\mathrm{pair}}}
\newcommand{\Smix}{S_{\mathrm{mix}}}
\newcommand{\poly}{\operatorname{poly}}
\newcommand{\ketbra}[2]{|#1\rangle\!\langle #2|}

\newcommand{\BC}{\mathcal{B}}
\newcommand{\CC}{\mathcal{C}}
\newcommand{\DC}{\mathcal{D}}

\newcommand{\HC}{\mathcal{H}}

\newcommand{\MC}{\mathcal{M}}
\newcommand{\NC}{\mathcal{N}}
\newcommand{\OC}{\mathcal{O}}
\newcommand{\PC}{\mathcal{P}}

\newcommand{\SC}{\mathcal{S}}
\newcommand{\TC}{\mathcal{T}}
\newcommand{\UC}{\mathcal{U}}
\newcommand{\VC}{\mathcal{V}}

\newcommand{\llangle}{\langle\!\langle}
\newcommand{\rrangle}{\rangle\!\rangle}
\newcommand{\kket}[1]{| #1 \rrangle}

\newcommand{\bbrakket}[2]{\llangle #1 | #2 \rrangle}

\newtheorem{theorem}{Theorem}
\newtheorem{corollary}{Corollary}[section]
\newtheorem{lemma}{Lemma}[section]
\newtheorem{proposition}{Proposition}[section]
\newtheorem{definition}{Definition}[section]

\newtheorem{condition}{Condition}

\renewcommand{\emph}[1]{%
  \ifdim\fontdimen1\font>0pt
    \textup{#1}% if italic font: emphasize with upright
    \else
  \textit{#1}% if upright font: emphasize with italic
  \fi
}
\renewcommand{\em}{%
  \ifdim\fontdimen1\font>0pt
    \upshape % if already italic, switch to upright
  \else
    \itshape % otherwise, switch to italic
  \fi
}

\usepackage[caption=false]{subfig}

\makeatletter
\DeclareDocumentCommand\todo{g}{%
\def\@message{\IfNoValueTF{#1}{TODO}{TODO: #1}}
\textbf{\textcolor[HTML]{FF8811}{\@message}}
\@latex@warning{\@message}{}{}}
\makeatother

\makeatletter 
    
\renewcommand\onecolumngrid{% <<<<<<
\do@columngrid{one}{\@ne}%
\def\set@footnotewidth{\onecolumngrid}% <<<<<<<<<<<<<<<<
\def\footnoterule{\kern-6pt\hrule width 1.5in\kern6pt}%
}

\renewcommand\twocolumngrid{% <<<<<<
        \def\footnoterule{% restore rule
        \dimen@\skip\footins\divide\dimen@\thr@@
        \kern-\dimen@\hrule width.5in\kern\dimen@}
        \do@columngrid{mlt}{\tw@}
}%

\makeatother    
\begin{document}

\addtocontents{toc}{\protect\setcounter{tocdepth}{0}}  % Turn off table of contents for main text
\mleftright

\title{Provable and scalable quantum Gaussian processes for  quantum learning}
\author{Jonas Jäger}
\affiliation{Theoretical Division, Los Alamos National Laboratory, Los Alamos, NM 87545, USA}
\affiliation{Department of Computer Science and Institute of Applied Mathematics, University of British Columbia, Vancouver, V6T 1Z4 B.C., Canada}
\affiliation{Stewart Blusson Quantum Matter Institute, Vancouver, V6T 1Z4 B.C., Canada}

\author{Paolo Braccia}
\affiliation{Theoretical Division, Los Alamos National Laboratory, Los Alamos, NM 87545, USA}

\author{Pablo Bermejo}
\affiliation{Information Sciences, Los Alamos National Laboratory, Los Alamos, NM 87545, USA}
\affiliation{Donostia International Physics Center, Paseo Manuel de Lardizabal 4, E-20018 San Sebasti\'an, Spain}
\affiliation{Department of Applied Physics, Gipuzkoa School of Engineering, University of the Basque
Country (UPV/EHU), Plaza Europa 1, 20018 San Sebastián, Spain}

\author{Manuel~G.~Algaba}
\affiliation{Theoretical Division, Los Alamos National Laboratory, Los Alamos, NM 87545, USA}
\affiliation{IQM Quantum Computers, Georg-Brauchle-Ring 23-25, 80992 Munich, Germany}
\affiliation{PhD Programme in Condensed Matter Physics, Nanoscience and Biophysics,
Doctoral School, Universidad Autónoma de Madrid}

\author{Diego Garc\'ia-Mart\'in}
\affiliation{Department for Quantum Information and Computation at Kepler (QUICK),\\ Johannes Kepler University, Linz, Austria }
\affiliation{Information Sciences, Los Alamos National Laboratory, Los Alamos, NM 87545, USA}

\author{M. Cerezo}
\thanks{cerezo@lanl.gov}
\affiliation{Information Sciences, Los Alamos National Laboratory, Los Alamos, NM 87545, USA}
\affiliation{Quantum Science Center, Oak Ridge, TN 37931, USA}

\begin{abstract}

Despite rapid recent advances in quantum machine learning, the field is in many ways stuck. Existing approaches can exhibit serious limitations, and we still lack learning frameworks that are simple, interpretable, scalable, and naturally suited to quantum data. To address this, here we introduce quantum Gaussian processes, a Bayesian framework for learning from quantum systems through priors over unknown quantum transformations. We show that, under suitable conditions, unitary quantum stochastic processes define Gaussian processes, thereby enabling regression, classification, and Bayesian optimization directly on quantum data. The key ingredient in this framework is sufficient knowledge of a quantum process's structure and symmetries to define an informative prior through its corresponding quantum kernel, effectively injecting a strong, physics-informed inductive bias into the learning model. We then prove that matchgate, or free-fermionic, evolutions give rise to provable and scalable quantum Gaussian processes, providing the first family in our framework where the unknown unitary acts non-trivially on all qubits. Finally, we demonstrate accurate long-range extrapolation, phase-diagram learning in many-body systems, and sample-efficient Bayesian optimization in a quantum sensing task. Our results identify quantum Gaussian processes as a promising route toward simpler and more structured forms of quantum learning.

\end{abstract}

\maketitle

    \begin{acronym}
        \acro{dla}[DLA]{dynamical Lie algebra}
        \acro{gp}[GP]{Gaussian process}
            \acroplural{gp}[GPs]{Gaussian processes}
        \acro{hs}[HS]{Hilbert-Schmidt}
        \acro{ml}[ML]{machine learning}
        \acro{qgp}[QGP]{quantum Gaussian process}
            \acroplural{qgp}[QGPs]{quantum Gaussian processes}
        \acro{qml}[QML]{quantum machine learning}
        \acro{rkhs}[RKHS]{reproducing kernel Hilbert space}
    \end{acronym}

\section{Introduction}

Quantum machine learning~\cite{dunjko2018machine,gujju2024quantum,wang2024comprehensive,chang2025primer} has grown rapidly over the past decade, but its conceptual landscape is still mainly shaped by two broad approaches. One builds learning algorithms from quantum linear-algebraic primitives~\cite{dunjko2018machine,ivashkov2026qkan}, such as state preparation, linear-system solvers, and related data-access models~\cite{kerenidis2016quantum}. The other takes inspiration from modern machine learning and seeks to learn directly with parametrized quantum circuits~\cite{cerezo2020variationalreview,bharti2021noisy}. Both routes have led to important insights, but both also come with serious limitations. In the first case, strong promises about data access models can make the resulting advantages fragile, and in some cases vulnerable to dequantization~\cite{tang2019quantum, tang2021quantum, tang2022dequantizing, cotler2021revisiting_refs}. In the second, the flexibility of variational models comes at the price of difficult optimization~\cite{bittel2021training,you2021exponentially,fontana2022nontrivial,anschuetz2022beyond,anschuetz2021critical,larocca2021diagnosing}, limited interpretability~\cite{steinmuller2022explainable,pira2024interpretability,heese2025explaining,gil2024opportunities}, curses of dimensionality~\cite{larocca2024review,mcclean2018barren,cerezo2020cost,ragone2023lie,fontana2023theadjoint}, and a poor understanding of when these methods should be expected to scale. Taken together, these issues suggest that quantum machine learning may benefit from simpler, more structured frameworks that leverage physics-informed inductive bias, especially in settings where the data are themselves quantum.

\begin{figure*}
\centering
\includegraphics[width=\linewidth]{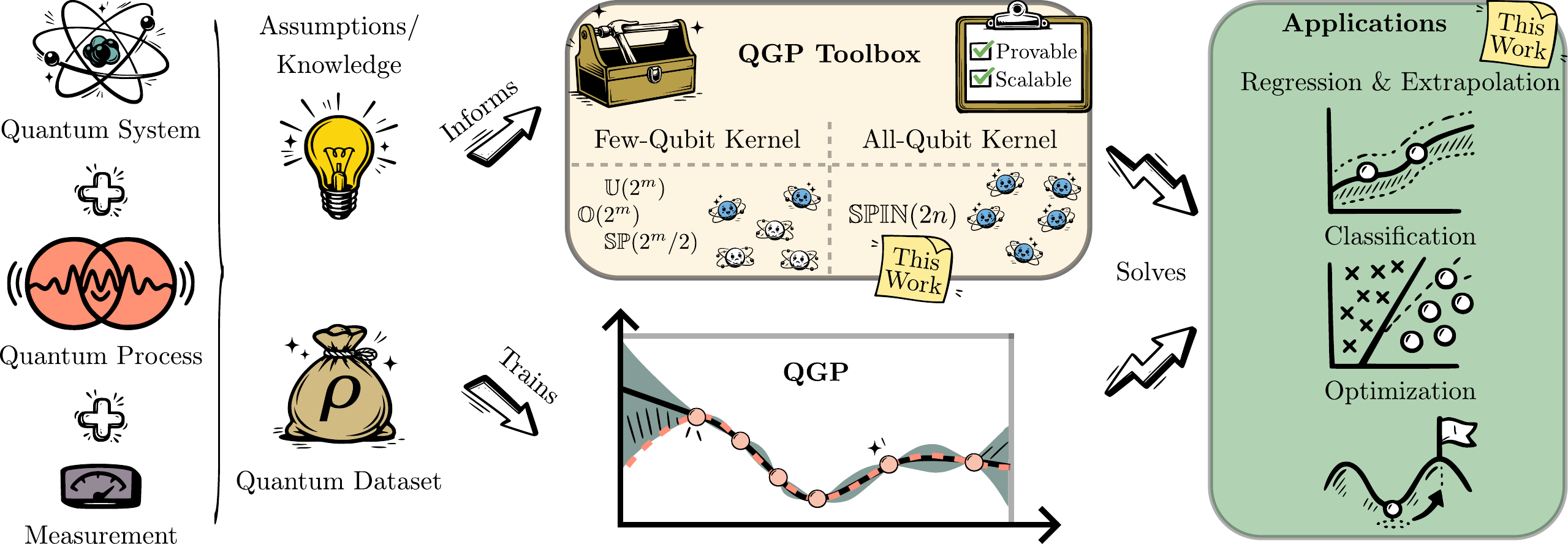}
\caption{\textbf{Schematic representation of our quantum Gaussian process (QGP) framework.} Information about an unknown quantum process and measurement operator, together with access to a quantum dataset, is used to select a kernel from a QGP toolbox. Previously known examples correspond to few-qubit unitary, orthogonal, and symplectic kernels, while this work adds an all-qubit matchgate, or free-fermionic, kernel. The resulting quantum Gaussian process then serves as a Bayesian surrogate that can be used for regression and extrapolation, classification, and optimization.}
\label{fig:schematic}
\end{figure*}

In this work we take a step in that direction by turning to Gaussian processes, a non-parametric Bayesian framework used for regression, classification, and optimization~\cite{rasmussen_GaussianProcessesMachine_2006,shahriari2015taking,bartok2010gaussian,kamath2018neural,deSpiegeleer2018machine,aigrain2023gaussian,deisenroth2011pilco,morales2017remote,macleod2020selfdriving,romero2013navigating}. But our aim is not to simply swap GPs into quantum learning as convenient surrogates. 
While recent studies have either focused on the asymptotic emergence of GPs in quantum models without leveraging them for scalable applications~\cite{garcia2023deep,garcia2024architectures,rad2023deep,girardi2024trained,melchor2025quantitative}, or applied GP regression with heuristically chosen quantum kernels~\cite{otten2020quantum,dai2022quantum,rapp2024quantum,guo2024benchmarking}, we take a holistic approach.
In Bayesian inference, the holy grail is to use the right prior for the right problem, and this is exactly the principle that guides our approach. We ask when a quantum stochastic process is itself naturally described by a Gaussian process, so that the prior, kernel, and learning task all arise from the structure of the unknown dynamics. As summarized in Fig.~\ref{fig:schematic}, partial knowledge about an unknown quantum process and a measurement of interest can be used to define a quantum kernel and, with it, a Bayesian prior. We then show how this viewpoint leads to a concrete and scalable toolbox for quantum learning. In particular, we collect the previously known few-qubit QGP priors and prove a new all-qubit family based on matchgate, or free-fermionic, evolutions. This yields a unified framework for regression, classification, and Bayesian optimization on quantum data.

\section{Results}

\subsection{Gaussian processes}

A stochastic process $f$ is a distribution over functions, and can be defined as a collection of random variables labeled by the points of a (possibly infinite) index set $\SC$. We say that $f$ is a~\acf{gp}, denoted $f\sim\mathcal{GP}(\mu,\kappa)$, if for any finite subset $\TC\subseteq\SC$ the random vector $\bm{f}_{\TC}=(f(t))_{t\in\TC}$ follows a multivariate normal distribution~\cite{rasmussen_GaussianProcessesMachine_2006,murphy_ProbabilisticMachineLearning_2023}
\begin{equation}\label{eq:GP1}
  \bm{f}_{\TC} \sim \mathcal{N}\!\bigl( \bm{\mu}_\TC,\, \bm{\Sigma}_{\TC,\TC} \bigr)\,.
\end{equation}
Here, $\bm{\mu}_\TC$ is the mean vector with entries $\mu_t$, and $\bm{\Sigma}_{\TC,\TC}$ is the covariance (kernel) matrix with entries $\kappa(t,t')$.

The pair $(\mu,\kappa)$ is the GP prior, and encodes prior statistical knowledge about $f$ before seeing any data. For any $t\in\SC$, we have $f(t)\sim \mathcal{N}\!\bigl(\mu(t),\,\kappa(t,t)\bigr)$. We can then improve our prediction of how $f(t)$ behaves by incorporating information from observations. For a finite $\TC\subseteq\SC$, suppose we record values
\begin{equation}\label{eq:GP2}
\bm{y}_{\TC} \;=\; \bm{f}_{\TC} + \bm{\varepsilon},\qquad
\bm{\varepsilon}\sim\mathcal{N}(\mathbf{0},\,R),
\end{equation}
where we assume that observations come with normally distributed noise with covariance $R$ (for independent and identically distributed noise, $R=\sigma^2\openone$). This leads to the joint Gaussian
\begin{equation}\label{eq:joint_GP}
\begin{bmatrix}
f(t) \\[2pt]
\bm{y}_{\TC}
\end{bmatrix}
\sim
\mathcal{N}\!\left(
\begin{bmatrix}
\mu(t) \\[2pt]
\bm{\mu}_{\TC}
\end{bmatrix},
\begin{bmatrix}
\kappa(t,t) & \bm{k}_{t\TC}^{\!\top} \\
\bm{k}_{t\TC} & K_{\TC\TC}+R
\end{bmatrix}
\right)\,,
\end{equation}
with $\bm{k}_{t\TC} = \bigl[\kappa(t,s)\bigr]_{s\in\TC}$ and $K_{\TC\TC}=\bigl[\kappa(s,s')\bigr]_{s,s'\in\TC}$. Conditioning then yields the posterior for $f(t)$
\begin{equation}\label{eq:GP-posterior}
\mathrm{Pr}\left(f(t) \middle| \bm{y}_{\TC} \right) = \mathcal{N}\left( 
\mu_{\text{P}}(t),
\sigma_{\text{P}}^{2}(t)
\right),
\end{equation}
where
\begin{align}
\mu_{\text{P}}(t)
&= \mu(t) + \bm{k}_{t\TC}^{\!\top}\cdot\bigl(K_{\TC\TC}+R\bigr)^{-1}\cdot \bigl(\bm{y}_{\TC}-\bm{\mu}_{\TC}\bigr),\label{eq:post-mean}\\
\sigma_{\text{P}}^{2}(t)
&= \kappa(t,t) - \bm{k}_{t\TC}^{\!\top}\cdot\bigl(K_{\TC\TC}+R\bigr)^{-1}\cdot\bm{k}_{t\TC}.\label{eq:post-var}
\end{align}
Equation~\eqref{eq:post-var} shows that the posterior variance is smaller than the prior one. In practice, the posterior mean $\mu_{\text{P}}$ serves as a predictor, while the posterior variance $\sigma^2_{\text{P}}$ quantifies uncertainty and can guide where to sample next. This uncertainty-aware feedback is what underlies active learning and Bayesian optimization.

\subsection{Quantum stochastic processes and regression problems}

Let $\HC=(\mathbb{C}^2)^{\otimes n}$ be an $n$-qubit Hilbert space with associated space of operators $\BC=\HC\otimes \HC^*$, and let $\SC=\{\rho(t)\}_t\subseteq\BC$ be a set of quantum states indexed by some label $t$ (not necessarily scalar). Then, let ${\G}\subseteq \mathbb{U}(2^n)$ be a unitary group acting on $\HC$, and denote by $d\mu$ the Haar measure over ${\G}$.

Next, consider the quantum stochastic process
\begin{equation}\label{eq:stochastic-process}
    f_U(t)=\Tr[U\rho(t)U\ad O]\,, 
\end{equation}
for some Hermitian operator $O\in\BC$, $\rho(t)\in\SC$, and $U$ some unitary from ${\G}$. We assume that the unknown transformation belongs to ${\G}$, and use $d\mu$ to encode our prior uncertainty over the possible unitaries (i.e., $U$ is as likely as any other unitary from this group). We then take a state in $\SC$, evolve it through $U$, and measure the expectation value of $O$. Once a concrete learning instance is fixed, there is a single unknown unitary $U$, and the GP framework is used to infer its action from data\footnote{In particular, using the framework does not require one to literally sample Haar-random unitaries in practice, but rather to know enough about the relevant symmetries in $U$ to establish a prior and specify the corresponding mean and kernel.}.

This naturally leads to the following regression task: \textit{During an initial phase, we are given repeated access to a finite training set $\TC\subseteq\SC$ as well as to a unitary $U$ belonging to ${\G}$. In this phase, which lasts no more than $\OC(\poly(n))$ time, we can send states through $U$, make measurements, and record outputs. Once this phase is over, we no longer have access to $U$ and the goal is to estimate the value $f_U(t)$ for all $\rho(t)\in\SC$.}

There are many ways to solve the previous problem, and their efficiency depends on how much information about the task is used. At the lowest level, one could perform full process tomography of $U$. But this requires time and memory scaling as $\OC(4^{n})$~\cite{chang2025primer}, and ignores the knowledge that $U$ comes from ${\G}$ and that we only care about states in $\SC$ and a single observable $O$. If instead $f_U(t)$ forms a quantum GP (QGP) with corresponding $\mu$ and $\kappa$, we can estimate the expectation values $\Tr[U\rho(t)U\ad O]$ during the initial data acquisition phase, build a GP model on the noisy observations $\bm{y}_{\TC}$, and perform predictions via Eqs.~\eqref{eq:GP-posterior}--\eqref{eq:post-var}. Hence, effectively leveraging the power of QGPs therefore requires two ingredients: 
\begin{enumerate}
    \item Provable QGP: Showing that $f_U\sim\mathcal{GP}(\mu,\kappa)$;
    \item Scalable QGP: The ability to estimate with $\OC(\poly(n))$ shots all relevant quantities in Eqs.~\eqref{eq:GP-posterior}--\eqref{eq:post-var} needed to perform an informed prediction.
\end{enumerate}

For instance, Ref.~\cite{garcia2023deep} showed that if ${\G}=\mathbb{U}(2^n)$ and $O$ is a traceless Pauli operator, then $f_U(t)$ forms a GP with 
\begin{equation}\label{eq:kernel-Ud}
 \mu_t=0\,, \quad\text{and}\quad \kappa(t,t')=\frac{1}{4^n}\lVert O \rVert^2_2\Tr[\rho(t)\rho(t')]\,,
\end{equation}
where $\kappa$ is, up to a factor, the fidelity kernel. This satisfies the first condition but not the second. Indeed, while $\Tr[\rho(t)\rho(t')]$ can be estimated to precision $\varepsilon\in\OC(1/\poly(n))$ with $\OC(\poly(n))$ measurements, we have that $\lVert O\rVert_2^2=2^n$ for a Pauli operator, so the kernel entries are of order $\OC(1/2^n)$ and are dominated by the noise matrix $R$, whose entries are of order $\OC(1/\poly(n))$. Clearly, restricting the unitary to act on only $m\in\OC(\log(n))$ qubits restores scalability~\cite{garcia2023deep}. Similar few-qubit constructions exist for $\mathbb{O}(2^m)$~\cite{garcia2023deep} and $\SP(2^m/2)$~\cite{garcia2024architectures}. In the next section we show that one can go beyond this few-qubit regime and obtain QGPs that are both provable and scalable even when the unknown unitary acts non-trivially on all qubits.

\subsection{Provable and scalable quantum Gaussian processes}

\begin{table*}[t]
    \centering
    \setlength{\tabcolsep}{6pt}  
    \begin{tabular}{cccccc}
        \toprule
        \textbf{Group $\G$} & \textbf{Acting on} & \textbf{Observable}  & $\mu_t$ & $\kappa(t,t')$ & \textbf{Proved in:} \\ \midrule
        Unitary $\mathbb{U}(2^m)$ & $m \in \OC(\log(n))$ &  $O\in \mathbb{C}\mathfrak{u}(2^m)$ & $0$ & $\frac{\langle \rho(t),\rho(t')\rangle_\mathfrak{u}\lVert O \rVert^2_2}{4^m}$ & Ref.~\cite{garcia2023deep} \\
        Orthogonal $\mathbb{O}(2^m)$ & $m \in \OC(\log(n))$ &  $O\in \mathbb{C} \mathfrak{o}(2^m)$ & $0$ & $\frac{\langle \rho(t),\rho(t')\rangle_\mathfrak{o}\lVert O \rVert^2_2}{4^{m-1/2}}$ & Ref.~\cite{garcia2023deep} \\
        Symplectic $\SP(2^m/2)$ & $m \in \OC(\log(n))$ &  $O\in \mathbb{C}\mathfrak{sp}(2^m/2)$ & $0$ & $\frac{\langle \rho(t),\rho(t')\rangle_\mathfrak{sp}\lVert O \rVert^2_2}{4^{m-1/2}}$ & Ref.~\cite{garcia2024architectures}\\
        Matchgate $\SPIN(2n)$ & All qubits & $O\in\BC_m$ \& $\min\{m,n-m\}\in\OC(1)$ & $0$ & $\frac{ m!\left\langle \rho(t),\rho(t') \right \rangle_{m}\lVert O \rVert^2_2}{(2n)^m}$ & This work  \\ \bottomrule
    \end{tabular}
    \caption{\textbf{Provable and scalable QGPs.} We present the groups known to lead to QGPs, as well as the conditions that make them scalable. For $\mathbb{U}(2^m)$, $\mathbb{O}(2^m)$ and $\SP(2^m/2)$, the transformation acts locally on a subset of $m\in\OC(\log(n))$ qubits, and the observable $O$ must belong to the subspaces of operator space spanned by $\mathbb{C}\mathfrak{g}$ for $\mathfrak{g}=\mathfrak{u}(2^m),\mathfrak{o}(2^m),\mathfrak{sp}(2^m/2)$, respectively. The matchgate case is quite different as the unitary $U$ acts on all qubits, while the measurement operator must be expressed as a linear combination of products of $m$ distinct Majoranas, with $\min\{m,n-m\}\in\OC(1)$.
    }
    \label{tab:groups_for_kernels}
\end{table*}

We now turn to the central question raised above: \textit{Can one obtain QGPs that are both provable and scalable when the unknown unitary acts non-trivially on all qubits?} Our main theoretical contribution is to show that the answer is yes for matchgate, or free-fermionic, circuits. As summarized in Table~\ref{tab:groups_for_kernels}, this is the first family in our framework where the unknown unitary acts on the full $n$-qubit system while still leading to a scalable kernel.

To state the result, we recall that under the Jordan-Wigner transformation the Majorana operators on $n$ fermionic modes can be expressed as $n$-qubit Pauli operators through
\begin{equation}\nonumber
                \begin{aligned}
                    C_1 &= X_1, \!\!&\!\! C_3 &= Z_1X_2, \!\!\!&\!\!\! &\ldots, \!\!&\!\! C_{2n-1} &= Z_1 \cdots Z_{n-1} X_n\,,\\
        C_2 &= Y_1, \!\!&\!\! C_4 &= Z_1Y_2, \!\!\!&\!\!\! &\ldots, \!\!&\!\! C_{2n} &= Z_1 \cdots Z_{n-1} Y_n\,.
    \end{aligned}
\end{equation}
where $X_i$, $Y_i$ and $Z_i$ denote Pauli operators on qubit $i$, and $\{C_\mu,C_\nu\}=2\delta_{\mu\nu}$. Matchgates are obtained by exponentiating quadratic Hamiltonians of the form $U=e^{-\sum_{\mu,\nu}h_{\mu\nu}C_\mu C_\nu}$, with $h$ a real anti-symmetric matrix. These unitaries form a representation of the Lie group $\SPIN(2n)$, the double cover of the special orthogonal group $\SO(2n)$\footnote{Note that while we work with the Lie group $\SPIN(2n)$, all results will hold for Bogoliubov transformations, or general matchgates in the representation of $\mathbb{O}(2n)$. }. In particular,
\begin{equation} \label{eq:Majorana-action}
    U C_\mu U\ad = \sum_{\nu=1}^{2n} Q_{\mu\nu} C_\nu\,,
\end{equation}
where $Q\in\SO(2n)$.

Under the action of matchgates the operator space decomposes into invariant subspaces
\begin{equation}\label{eq:irreps-Majoranas}
    \BC = \bigoplus_{m = 0}^{2n} \BC_m\,,
\end{equation}
where $\BC_m$ is spanned by products of $m$ distinct Majorana operators and satisfies $\dim(\BC_m) = \binom{2n}{m}$. We then define the overlap between two operators projected onto $\BC_m$ as
\begin{equation}
    \Braket{A, B}_{m} = \Tr\left[A_m^\dagger B_m \right]\,,
\end{equation}
where $A_m$ and $B_m$ denote the projections of $A$ and $B$ onto $\BC_m$.
This subspace overlap induces a norm (its square is the module purity \cite{diaz2023showcasing}) as $\lVert A \rVert^2_m = \braket{A, A}_{m}$.

Next, consider the collection of random variables $f=\{f_U(t_1),f_U(t_2),\ldots\}$, where the randomness comes from the prior: $U$ is a matchgate circuit from the representation of $\SPIN(2n)$. As proved in the Supplemental Information (see also the Methods for a sketch of the proof), the moments of $f$ asymptotically match those of a multivariate Gaussian distribution, leading to the following theorem:
\begin{theorem}[Informal]\label{theo-1}
    There exist sets of states $\SC$ and measurement operators $O\in\BC_m$ ($m\neq 0,2n$) for which $f$ forms a GP with 
    \begin{equation}
        \mu(t)\!=\!0,\,\text{and}\,\,\, \kappa(t,t')=\tfrac{m!}{(2n)^m} \lVert O \rVert^2_m \left\langle \rho(t) , \rho(t') \right \rangle_{m}.\label{eq:theo:1-main}
    \end{equation}
\end{theorem}
Thus, once the prior is restricted to matchgate evolutions, the covariance between two data points is controlled by the overlap of the corresponding states inside the relevant Majorana sector $\BC_m$.

Theorem~\ref{theo-1} guarantees the provable QGP condition, but scalability is not automatic. Indeed, with coherent access and Bell measurements, we can evaluate these overlaps $\left\langle \rho(t), \rho(t') \right \rangle_{m}$ up to precision $\varepsilon$ with a shot number scaling as $\OC(\poly(\dim(\BC_m))/\varepsilon^2)$ shots (see Supp.~Info.). Consequently, to achieve efficient polynomial scaling in the total number of qubits $n$ while ensuring that the kernel entries remain of the same order as those in $R$, one simply needs $\min\{m,n-m\}\in\OC(1)$.

Table~\ref{tab:groups_for_kernels} compiles the currently known provable and scalable QGPs: the cases $\mathbb{U}(2^m)$ and $\mathbb{O}(2^m)$ of~\cite{garcia2023deep}, $\SP(2^m/2)$ from~\cite{garcia2024architectures}, and the matchgate $\SPIN(2n)$ derived here. We highlight the fact that matchgates constitute the only known case where the unitary acts on all qubits, rather than on a small subset of them. Moreover, we find that, as expected~\cite{ragone2023lie}, in all cases the observable $O$ must belong to some subspace $\BC_\lambda$ of operator space, and the kernels take the unified form
\begin{equation}
\kappa(t, t') = \left\langle \rho(t),\rho(t') \right \rangle_{\lambda}\lVert O \rVert^2_2/d_\lambda,
\end{equation}
where $d_\lambda$ denotes the leading-order scaling of $\dim(\BC_\lambda)$ for large $n$.
To finish, we also note that in Table~\ref{tab:FF_GP_scenarios} we present cases of state datasets and observables for which the expectation values do not form GPs, showing that care must be taken to guarantee that the conditions for provable QGPs are actually met.

\begin{figure*}[th]
\centering
\includegraphics[width=\linewidth]{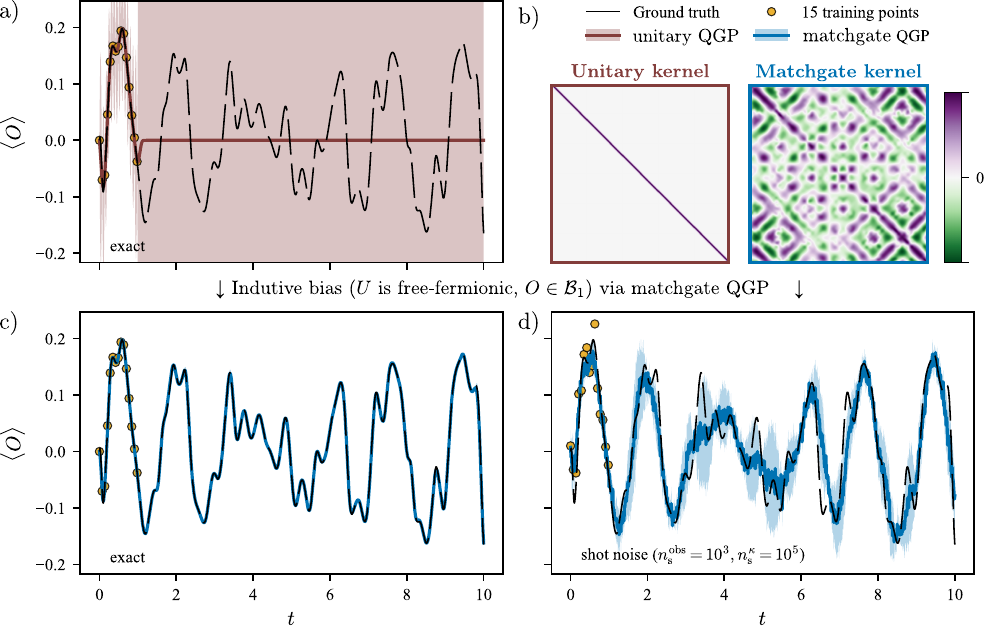}
 \caption{\textbf{Quantum Gaussian process extrapolation and the role of inductive bias.} 
We consider a 50-qubit regression task where the model is trained on 15 points in the interval $t\in[0,1]$ and then asked to extrapolate the signal out to $t\leq 10$. The unknown transformation is a random matchgate circuit from $\SPIN(2n)$, and the observable is a Pauli string in $\BC_1$. 
(a) If we ignore this structure and instead use the unitary kernel associated with $U\in\mathbb{U}(2^n)$, the QGP fails to extrapolate. (b) The corresponding kernel matrix is nearly band-diagonal, with rapidly vanishing off-diagonal entries, so long-range correlations are lost. When the correct inductive bias is used, namely $U\in\SPIN(2n)$ and $O\in\BC_1$, the resulting matchgate kernel retains strong off-diagonal structure and enables accurate extrapolation.
In the bottom row, we show extrapolation without shot noise (c) and with shot noise (d). Even for realistic budgets, the QGP remains robust, captures the overall signal, and reflects the additional uncertainty through wider posterior bands.
}
\label{fig:B1_extrapolation}
\end{figure*}

\subsection{Applications \& Experiments}\label{sec:applications}

Borrowing inspiration from the versatility of GPs in classical machine learning, we now demonstrate large-scale numerical simulations of our QGP framework across supervised learning and Bayesian optimization, spanning both synthetic and physical tasks. In all cases, we refer the reader to the Methods section for additional details on the experimental setups. Notably, while some of the experiments presented extend beyond strictly provable QGP regimes, the success of these simulations highlights the framework's value as a strong, practical heuristic.

\subsubsection{Regression}

We begin by showcasing QGPs for regression tasks. Each example highlights a different aspect of the framework. We start with a controlled extrapolation problem, where the main goal is to isolate the role of inductive bias. We then move to physically motivated many-body problems, where the task is to reconstruct phase diagrams and order parameters from extremely sparse quantum data. We focus here on QGP performance, and refer the reader to Supplemental Information for benchmarks against alternative regression methods. These then reveal shot efficiency of QGPs over non-data-driven approaches and establish that QGP active learning outperforms quantum learners that lack uncertainty quantification.

\paragraph{Extrapolation.}\label{sec:B1_extrapolation} Extrapolation is a particularly stringent test in supervised learning, since the model must make predictions well outside the region covered by the training data.

\begin{figure}[th]
    \centering{
    \includegraphics[width=1\linewidth]{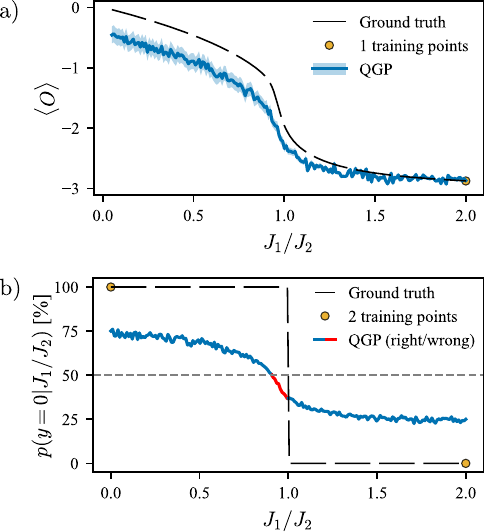}
    }
    \caption{    \textbf{QGPs for phase diagram prediction and classification in a bond-alternating XXX model.} This model exhibits a phase transition at $J_1=J_2$. 
    (a) We showcase the regression of the phase diagram using just one training point.
    (b) A QGP classifier achieves accurate phase classification, trained on a minimal set of a single training point per phase. The training labels are no longer quantitative order parameter values but qualitative binary phase labels.
    }
    \label{fig:xxx}
\end{figure}

Here we consider a 50-qubit synthetic dataset where 15 training states $\rho(t)$, with $t\in\SC=[0,1]$, are sent through an unknown matchgate circuit $U$ from $\SPIN(2n)$, and where we measure a Pauli observable $O\in\BC_1$. The task is to predict the signal out to $t\leq 10$, in a regime where the dynamics are not periodic and success is not trivially guaranteed. To emphasize the role of inductive bias, we first deliberately use the wrong prior and ``forget'' that $U$ is a matchgate, replacing the correct kernel by the unitary kernel of Eq.~\eqref{eq:kernel-Ud}, as if $U$ could have been any element of $\mathbb{U}(2^n)$. As shown in Fig.~\ref{fig:B1_extrapolation}, extrapolation is then unattainable. The QGP predictions collapse to the prior mean and the uncertainty rapidly grows once we leave the training region. This can be directly understood from the corresponding kernel matrix, whose off-diagonal entries quickly vanish and therefore fail to encode the long-range correlations needed for extrapolation.

Next, we impose the correct inductive bias, namely that $U$ belongs to the representation of $\SPIN(2n)$ and that $O\in\BC_1$, so that the process forms a GP as per Theorem~\ref{theo-1} with kernel given by Eq.~\eqref{eq:theo:1-main}. The corresponding kernel now captures meaningful similarities among data points over long time scales, and this is precisely what allows the QGP to extrapolate. In Fig.~\ref{fig:B1_extrapolation} we also study the effect of shot noise. In the noiseless case, the QGP extrapolates perfectly. More importantly, this behavior remains robust when we include realistic finite-shot estimates, taking $\shotsobs=10^3$ shots per training expectation value and $\shotsker=10^5$ shots per kernel entry. The additional noise broadens the uncertainty bands, but the overall trend of the signal is still correctly captured. Thus, the first example shows that successful interpolation can generally be achieved, here with only 15 training points in a narrow region, but extrapolation becomes impossible in the absence of a suitable inductive bias.

\paragraph{Phase transition in a one-dimensional XXX bond-alternating model.} Next, we move to a physically motivated problem and use QGP regression to recover the ground-state phase diagram of a bond-alternating XXX model~\cite{Kitazawa1996phase} on $n=50$ qubits. The target is the expectation value of the order-parameter proxy
\begin{equation}\label{eq:measurement-XXX}
    O=X_{\frac{n}{2}}X_{\frac{n}{2}+1} + Y_{\frac{n}{2}}Y_{\frac{n}{2}+1}+Z_{\frac{n}{2}}Z_{\frac{n}{2}+1}\,.
\end{equation}
Unlike the extrapolation example above, this observable has support on more than one invariant sector: the first two terms belong to $\BC_2$ in Eq.~\eqref{eq:irreps-Majoranas}, while the last term belongs to $\BC_4$. We therefore compute the matchgate kernel of Eq.~\eqref{eq:theo:1-main} separately on these two sectors and then sum the resulting kernels to perform the regression. While Theorem~\ref{theo-1} does not directly cover this combined observable, the resulting kernel is still positive semidefinite and hence defines a valid GP model.

The result is shown in Fig.~\ref{fig:xxx}(a). Using only an informative training point placed at $J_1=2J_2$, and shot budgets $\shotsobs=\shotsker=10^4$, the QGP correctly reconstructs the change in the order parameter across the transition near $J_1=J_2$, as predictions converge towards the prior for small $J_1/J_2$. This shows that the QGP framework can already be used in a realistic many-body setting, even when the observable is not confined to a single subspace.

\begin{figure*}[th]
    \centering
    \includegraphics[width=\linewidth]{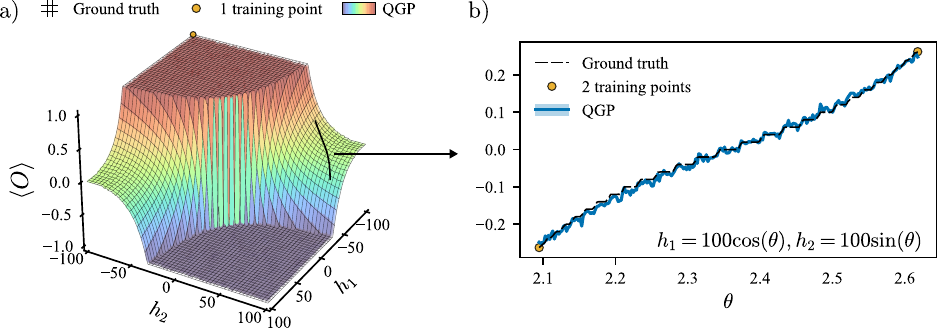}
\caption{\textbf{QGP learning the magnetization phases of an XXZ Hamiltonian.} 
(a) Regressed magnetization landscape for an $n=100$ XXZ chain, obtained from a single training point placed in the maximal magnetization plateau. The wireframe shows the true magnetization profile, while the colored surface is the QGP prediction. The training point is marked by a yellow circle.  (b) Regression along the black cut shown in the left panel. Using only the two extrema of the cut as training points, and shot budgets $n_{\rm s}^{\rm obs}=n_{\rm s}^{\kappa}=10^5$, the QGP reproduces the staircase structure of the magnetization plateaus.}
    \label{fig:xxz}
\end{figure*}

\paragraph{Phase transition in the XXZ model.} Finally, we ask whether the previous picture survives when both the system size and the complexity of the phase landscape are increased. To this end, we apply QGP regression to the magnetization diagram of a one-dimensional $n=100$ qubit XXZ chain with longitudinal and transversal fields $(h_1,h_2) \in\SC=[-100,100]^2$~\cite{cerezo2017factorization}. We take as observable the magnetization
\begin{equation}
    O=\frac{1}{n}\sum_k Z_k,
\end{equation}
which belongs to $\BC_2$.

The results are presented in Fig.~\ref{fig:xxz}. In the left panel, the training set consists of a single point chosen in the region of maximal magnetization, and we use shot budgets $\shotsobs=\shotsker=10^4$. Even under this extremely sparse supervision, the regressed magnetization landscape closely matches the true one. This behavior can be understood from the fact that when $h_1,h_2\rightarrow +\infty$, the ground state approaches $|0\rangle^{\otimes n}$, so kernel rows involving this state are already strongly tied to the magnetization profile itself. In other words, the chosen training point is not just convenient numerically, but physically informative. In the right panel, we further probe the transition between magnetization plateaus by restricting the phase diagram to the curve $h_1=100\cos(\theta)$, $h_2=100\sin(\theta)$ with $\theta\in[2\pi/3,5\pi/6]$. Taking the two extrema of this cut as training points, and increasing the shot budget to $\shotsobs=\shotsker=10^5$, the QGP accurately reproduces the staircase structure of the magnetization. This final regression example shows that QGPs remain scalable and effective at $n=100$ qubits and over a genuinely two-dimensional parameter landscape, even when trained on only one or two carefully chosen quantum data points.

\subsubsection{Classification}

Beyond regression, classification is the other primary task in supervised machine learning. Here we showcase how QGPs can directly classify the ground states of the aforementioned bond-alternating XXX model, without first having to learn the order parameter as in the previous section. The kernel is the same as before, but now the (noiseless) labels are binary classes $y\in\{0,1\}$ assigned according to which side of the critical point the state belongs to. 

To probe the limits of the QGP, we consider a minimal training set with a single training point per class, as shown in Fig.~\ref{fig:xxx}(b). Even in this extremely sparse regime, the model achieves a test accuracy of $95\%$. The only misclassifications occur close to the critical point, leading to a slight mismatch between the learned decision boundary and the true transition point at $J_1/J_2=1$. This is not unexpected, since the point $J_1/J_2=1$ sharply separates the two phases only in the thermodynamic limit, while for the finite system size considered here ($n=50$), the crossover already takes place over a finite region~\cite{bermejo2024quantum}. What makes this example particularly interesting is that in GP classification one first models a latent Gaussian variable via GP regression, and only then maps it through a sigmoid to obtain the Bernoulli probabilities for the class labels. An analysis of this latent QGP reveals a remarkably strong correlation with the order parameter proxy of Eq.~\eqref{eq:measurement-XXX}, with correlation coefficient $0.998$. Thus, the classifier is not only able to separate the phases from qualitative labels, but also recovers the underlying order-parameter structure in its latent representation.

\subsubsection{Optimization}\label{sec:bo}

Next, we validate the utility of QGPs as surrogate functions within Bayesian optimization (BO)~\cite{shahriari2015taking,Mockus1978,jonesEfficient1998}. This adds a new ingredient to the framework. In the previous tasks, the QGP was used to predict labels or observables from limited data. Here, instead, its uncertainty estimates determine where the next query should be made. By balancing exploration and exploitation, BO iteratively refines the surrogate only where needed, and can therefore locate optima using remarkably few evaluations of the objective function. We illustrate this in two settings: first in a synthetic two-dimensional task designed to isolate the optimization dynamics, and then in a quantum sensing application where the query budget is itself the central experimental resource. To keep the focus on the optimization mechanism, we neglect shot noise in the simulations.

\paragraph{Synthetic 2D optimization example.}\label{sec:bo:2D}

\begin{figure*}[tb]
    \centering
    \includegraphics[width=\linewidth]{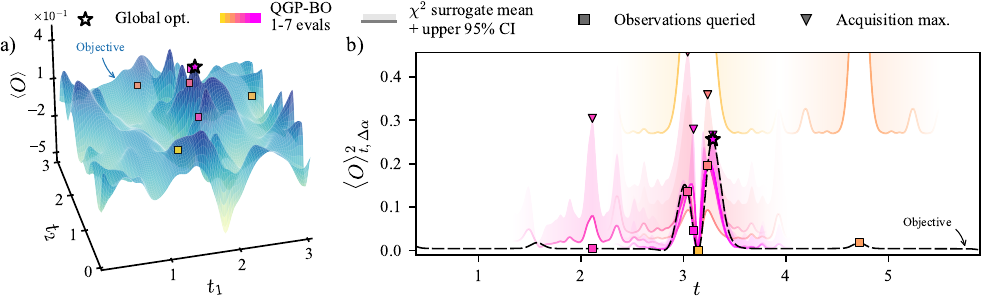}
    \caption{\textbf{Bayesian optimization using QGP surrogates (QGP-BO).}
(a) Synthetic two-parameter optimization landscape built from a 30-qubit matchgate dataset, with an objective function derived from the expectation values of an $\BC_1$ observable on states $\rho_{\bm{t}}$. The QGP surrogate uses the corresponding $\BC_1$ kernel. Arrows trace the BO trajectory from a random initial point to the global optimum (star).
(b) QGP-BO applied to a 55-qubit quantum sensing task. The black line denotes the true objective, while the colored curves show the surrogate mean together with the upper $95\%$ confidence bound at successive BO iterations. Squares denote queried observations, and triangles mark the maxima of the acquisition function that determine the next query. Despite broad flat regions in the objective landscape, the optimizer reaches the global optimum in only 7 evaluations. The elevated surrogate mean in early iterations is a consequence of the large initial uncertainty, which becomes strictly positive after squaring the Gaussian surrogate.}
    \label{fig:bo}
\end{figure*}

BO is particularly well suited for rugged objective functions with only a few parameters, where each function evaluation is expensive and where standard optimizers can be misled by local optima. To place our QGP-BO scheme in this regime, we construct a two-parameter optimization task with $\bm{t}\in[0,3]^2$ by extending the synthetic matchgate dataset used in the extrapolation example of Fig.~\ref{fig:B1_extrapolation} to two dimensions. The resulting objective is highly non-convex and contains many sub-optimal local optima.

A representative QGP-BO trajectory is shown in Fig.~\ref{fig:bo}. Starting from a random initial point, the optimizer reaches the global optimum in only 7 function evaluations, up to a resolution of $\Delta t = 0.06$. What is most notable here is not only the small number of evaluations, but also the fact that the optimization path avoids getting trapped in nearby local optima. This is precisely where the QGP surrogate becomes useful, as the posterior uncertainty keeps the acquisition function from collapsing too early onto locally promising, but globally sub-optimal, regions. Supplemental Info. provides a detailed optimization benchmark.

\paragraph{Quantum sensing application.}\label{sec:bo:QS}

We now turn to a practically motivated task where query efficiency is not just convenient but essential. Quantum sensing uses quantum systems to estimate physical quantities with precisions beyond those accessible to purely classical probes~\cite{giovannetti2006quantum,degen2017quantum,huerta2022inference,ijaz2024more}. Here we consider a variational magnetometry problem, where the goal is to optimize a parametrized probe state $\rho(t)$~\cite{koczor2020variational,beckey2020variational,kaubruegger2021quantum,meyer2020variational,thurtell2022optimizing,ijaz2024more} so as to maximize its sensitivity to a small change $\Delta \alpha$ in the magnetic field.

In particular, we consider a fixed-budget optimization task, where every evaluation used for probe-state optimization consumes the same resource later needed for the actual sensing protocol~\cite{ijaz2024more}. This makes BO a natural fit. It also makes QGP surrogates especially attractive, since the field evolution is a matchgate transformation and the chosen observable satisfies $O\in\BC_2$, so the surrogate can be built from the corresponding $\BC_2$ matchgate kernel. In contrast to the previous BO example, however, the objective is no longer the expectation value of the observable itself, but rather its sensitivity with respect to the field parameter $\alpha$. We therefore build a composite surrogate from the underlying QGP, and replace the absolute value of the sensitivity by its square, which preserves the location of the optimum.

The resulting QGP-BO trajectory is shown in Fig.~\ref{fig:bo} for a magnetometry task with $n=55$ qubits. The optimizer identifies the maximally sensitive probe state in only seven evaluations, reaching the global optimum with precision $\Delta t = 0.006$. By the time the optimum is found, the surrogate is still not an especially accurate global model of the underlying expectation value landscape. Yet this does not prevent BO from succeeding. The surrogate already contains enough information to point the acquisition function toward the right region and to jump efficiently across wide flat portions of the landscape. This highlights an important distinction between BO and regression: the goal of the surrogate is not to reconstruct the full function everywhere, but to supply enough local guidance to identify where improvement is most likely.

\section{Discussion}

In this work we introduced quantum Gaussian processes as a simple and structured framework for learning from quantum systems. Rather than approaching quantum learning through increasingly flexible parametric models, we showed how it can be formulated through Bayesian priors tailored to the underlying quantum dynamics, observables and datasets. 
This physics-informed treatment encodes strong inductive biases directly into the mathematical framework, allowing the models to generalize efficiently without requiring exhaustive data.
In particular, we established conditions under which quantum stochastic processes form Gaussian processes, proved that matchgate, or free-fermionic, evolutions give rise to a provable and scalable all-qubit QGP family, and demonstrated how the resulting kernels can be used for regression, classification, and Bayesian optimization on quantum data.

Several natural directions follow from this work. On the theory side, it will be important to broaden the range of quantum processes known to yield useful QGPs, and to better understand how far the present framework extends beyond the settings covered by our current theorems. On the practical side, a particularly promising direction is to use QGPs as surrogate models inside more standard variational and parametrized quantum-circuit workflows. In that context, Bayesian optimization could provide a principled way of navigating costly quantum objective landscapes with far fewer circuit evaluations, while the kernel itself could encode physically motivated inductive biases that are otherwise absent in generic variational ansätze.  Ultimately, this framework lays the groundwork for a broader paradigm of probabilistic and Bayesian quantum machine learning.

\section{Methods}\label{sec:methods}           

\begin{table*}[htbp]
                \setlength{\tabcolsep}{6pt}
                \begin{tabular}{@{}lcccccc@{}}
                \toprule
                \textbf{}      & \multicolumn{2}{c}{\textbf{Odd $m$}} &  & \multicolumn{3}{c}{\textbf{Even $m$}}                          \\ \cmidrule{2-3} \cmidrule{5-7} 
                \textbf{}      & \textbf{$m=1$} & \textbf{$m \geq 3$} &  & \textbf{$m = 2$} & \textbf{$m \equiv 0 \pmod 4$} & \textbf{$m \equiv 2 \pmod 4$} \\ \midrule
                \textbf{General}   
                    &   
                    \cmark
                    & 
                    \makecell[tl]{\begin{varwidth}[t]{\linewidth}
                    \begin{itemize}[leftmargin=5pt, topsep=0pt, partopsep=0pt, parsep=0pt, itemsep=0pt]
                    \item[\cmark] 
                     if Gaussian\\
                     pairings dominate 
                    \end{itemize}
                    \end{varwidth}
                    }
                    &
                    % This cell is intentionally left blank to have a blank column for horizontal spacing reasons
                    &
                    \makecell[tl]{\begin{varwidth}[t]{\linewidth}
                    \begin{itemize}[leftmargin=5pt, topsep=0pt, partopsep=0pt, parsep=0pt, itemsep=0pt]
                        \item[\cmark] if Gaussian\\
                     pairings dominate 
                    \end{itemize}
                    \end{varwidth}
                    }
                    &
                    \makecell[tl]{\begin{varwidth}[t]{\linewidth}
                    \begin{itemize}[leftmargin=5pt, topsep=0pt, partopsep=0pt, parsep=0pt, itemsep=0pt]
                        \item[\cmark] if Gaussian\\
                     pairings dominate \textbf{and}\\ if odd moments vanish
                    \end{itemize}
                    \end{varwidth}
                    }
                    &
                    \makecell[tl]{\begin{varwidth}[t]{\linewidth}
                    \begin{itemize}[leftmargin=5pt, topsep=0pt, partopsep=0pt, parsep=0pt, itemsep=0pt]
                        \item[\cmark] if Gaussian\\
                     pairings dominate\\
                     \textbf{and} if odd $(k\geq 5)$\\
                     moments vanish
                    \end{itemize}
                    \end{varwidth}
                    }
                    \\
                    \vspace{-6pt}
                    \\
                \multirow[t]{2}{*}{\makecell[tl]{\textbf{Special}\\\textbf{cases}}} &   &  && 
                \makecell[tl]{\begin{varwidth}[t]{\linewidth}
                    \begin{itemize}[leftmargin=5pt, topsep=0pt, partopsep=0pt, parsep=0pt, itemsep=0pt]
                            \item[\cmark] Gaussian states (*)
                            \item[\omark] Pauli observables 
                            \item[\cmark] Random observables \\ (Coeffs. $\sim \mathcal{N}(0,1)$)
                            \item[\cmark] Magic (extent) state
                        \end{itemize}
                \end{varwidth}}
                & 
                \makecell[tl]{\begin{varwidth}[t]{\linewidth}
                    \begin{itemize}[leftmargin=5pt, topsep=0pt, partopsep=0pt, parsep=0pt, itemsep=0pt]
                            \item[\omark] Gaussian states $(\BC_4)$
                            \item[\omark] Pauli observables $(\BC_4)$
                            \item[\xmark] Gaussian states and \\ Pauli observable $(\BC_4)$
                            \item[\cmark] \textit{Random fermionic state}\\\textit{and Pauli observable} $(\BC_4)$
                        \end{itemize}
                \end{varwidth}}
                &                      \\ \cmidrule{2-7}
                &\multicolumn{6}{c}{\makecell[tl]{\begin{varwidth}[t]{\linewidth}\begin{itemize}[leftmargin=5pt, topsep=0pt, partopsep=0pt, parsep=0pt, itemsep=0pt] \item[\cmark] For any fixed $m$, a special class of observables constructed from mutually disjoint Majorana products in $\BC_m$ \end{itemize}\end{varwidth}}} 
                \\\bottomrule
                \end{tabular}
                \caption{
                \textbf{Summary of QGPs from matchgate transformations in the large-$n$ limit.} 
                While ``\cmark'' (``\xmark'') indicate that such state and/or observable choices are shown to fulfill (violate) the GP conditions, ``\omark'' marks that such families of states or observables alone do not suffice to guarantee GP emergence.  \textit{Italics} imply that the conclusion is based on numerical evidence rather than a formal proof. Here, by Gaussian state we mean any state that we can obtain by applying a $U$ from $\SPIN(2n)$ to the all-zero state, whereas a fermionic state is any pure state with well-defined parity $Z^{\otimes n}$. (*) A set of Gaussian states provably forms a multivariate GP provided their subspace overlaps satisfy an asymptotic lower bound. For all remaining cases and subspaces $\BC_m$ ($m>1$), verifying that the states satisfy dominance conditions is required to yield a multivariate GP.
                }
                \label{tab:FF_GP_scenarios}
            \end{table*}

\subsection{Sketch of the proof of Theorem~\ref{theo-1}}

We sketch here the logic behind Theorem~\ref{theo-1}. The goal is to show that, for suitable families of states $\rho_i$ and observable $O$, the random variables
\begin{equation}
X_i = \operatorname{Tr}[U \rho_i U^\dagger O] ,
\end{equation}
with $U$ drawn from the Haar measure over the matchgate representation of $\SPIN(2n)$, have moments $\mathbb{E}_{U\sim \SPIN(2n)}[X_{i_1}\cdots X_{i_k}]$ that asymptotically coincide with those of a multivariate Gaussian distribution.

Unlike the classical neural network infinite-width correspondence~\cite{lee2018deep}, the proof does not follow from a central limit theorem. The reason is that the matrix elements of the orthogonal transformation induced by a matchgate are correlated. As such, a more direct approach must be taken, where one computes all infinitely many moments for random variables such as $X_i$. To do so, we exploit the decomposition of operator space into invariant Majorana sectors as in Eq.~\eqref{eq:irreps-Majoranas}, and restrict the observable to a single subspace $O\in \BC_m$. Here, one finds that expanding both $O$ and the projected states $\rho(t_{i_j})$ in the subspace turns them into anti-symmetric rank-$m$ tensors. For instance, we can define the $2n\times2n\times2n\times2n$ tensor $\tstate_{4}^{(i_j)}$, corresponding to projecting $\rho(t_{i_j})$ onto $\BC_4$, with entries
\begin{equation}
    [\tstate_{4}^{(i_j)}]_{\mu,\nu,\alpha,\beta}=-\tfrac{1}{\sqrt{24d}}\Tr[\rho(t_{i_j})C_\mu C_\nu C_\alpha C_\beta]\,, 
\end{equation}
if $\mu\neq\nu\neq\alpha\neq\beta$, and zero otherwise, and check that the anticommutation relations of Majoranas force $\tstate_4^{(i_j)}$ to be a rank-$4$ anti-symmetric tensor. Importantly, a direct application of Eq.~\eqref{eq:Majorana-action} shows that under the action of a matchgate unitary $U$, the tensor $\tstate_{4}^{(i_j)}$ transforms under $Q \in \SO(2n)$ as
\begin{align}\label{eq:tensor-transform}
    &[\tstate_{4}^{(i_j)}]_{\mu,\nu,\alpha,\beta}\nonumber\\
    &\quad=\sum_{\mu',\nu',\alpha',\beta'} q_{\mu,\mu'}q_{\nu,\nu'}q_{\alpha,\alpha'}q_{\beta,\beta'}[\tstate_{4}^{(i_j)}]_{\mu',\nu',\alpha',\beta'}\,.
\end{align}

From here, we find it extremely useful to work on the vectorization picture. For convenience, we recall that given an operator $A\in\BC=\HC\otimes\HC^*$ with expansion
\begin{align}
A=\sum_{i,j=1}^{d} A_{i,j}\ketbra{i}{j},
\end{align}
we define the vectorization map $\operatorname{Vec}:\HC\otimes\HC^*\to \HC^{\otimes 2}$, whose action on $A$ is:
\begin{align}
|A\rangle\!\rangle=\operatorname{Vec}[A]=\sum_{i,j} A_{i,j}\ket{i}\ket{j}.
\end{align}
This representation is particularly useful because it turns left and right multiplication of operators into ordinary matrix multiplication on the doubled Hilbert space. In particular, for any linear operators $A$, $B$, and $C$, one has
\begin{align}\label{eq:kroneker_vectorisation}
    |ABC\rangle\!\rangle = A\otimes C^\top |B\rangle\!\rangle,
\end{align}
where $C^\top$ denotes the transpose of $C$. As a direct consequence, the Hilbert--Schmidt inner product between operators becomes the standard inner product between their vectorizations,
\begin{align}
    \langle\!\langle A|B\rangle\!\rangle = \Tr[A^\dagger B].
\end{align}
The same construction naturally extends to channels. Given a linear map $\mathcal{M}:\HC\otimes\HC^*\to\HC\otimes\HC^*$ of the form
\begin{align}
\mathcal{M}(\cdot)=\sum_{i,j} A_i (\cdot) B_j^\dagger,
\end{align}
its vectorized representation is
\begin{align}
    M=\operatorname{Vec}\!\bigl[\Phi(\cdot)\bigr] = \sum_{i,j} A_i\otimes B_j^*,
\end{align}
where $B_j^*$ denotes the complex conjugate of $B_j$. Hence, vectorization recasts superoperators as operators.

After vectorization, we find that Eq.~\eqref{eq:tensor-transform} simply becomes 
\begin{equation}
    |\tstate_{4}^{(i_j)}\rangle\!\rangle\xrightarrow{U}Q^{\otimes 4}|\tstate_{4}^{(i_j)}\rangle\!\rangle\,,
\end{equation}
with generalization to the component in $\BC_m$ and, hence, corresponding projection tensor $\tstate$ of rank $m$
\begin{equation}
    |\tstate^{(i_j)}\rangle\!\rangle\xrightarrow{U}Q^{\otimes m}|\tstate^{(i_j)}\rangle\!\rangle\,.
\end{equation}
Then, we obtain
\begin{align}
&\E_{U\sim \SPIN(2n)}\left[ X_{i_1} X_{i_2} \cdots X_{i_k} \right] \nonumber\\
&\quad=\E_{Q\sim \SO(2n)}\left[ \langle\!\bra{\tobs}^{\otimes k} Q^{\otimes mk}\bigotimes_{j=1}^{k} \ket{\tstate^{(i_j)}}\!\rangle \right]\nonumber\\
&\quad= \langle\!\bra{\tobs}^{\otimes k} \E_{Q\sim \SO(2n)}\left[Q^{\otimes mk}\right]\bigotimes_{j=1}^{k} \ket{\tstate^{(i_j)}}\!\rangle \,,
\end{align}
where we can identify 
\begin{align}
    M^{(mk/2)}_{\SO(2n)}&=\E_{Q\sim \SO(2n)}\left[Q^{\otimes mk}\right]\nonumber\\
    &=\E_{Q\sim \SO(2n)}\left[Q^{\otimes mk/2}\otimes (Q^*)^{\otimes mk/2}\right]\,,
\end{align}
as 
the $k$-th moment operator (the vectorization of the $k$-th order twirl channel)~\cite{garcia2023deep,mele_IntroductionHaarMeasure_2024} for the standard representation of $\SO(2n)$. In particular, we have used the fact that for any $Q\in\SO(2n)$, then $Q=Q^*$. Putting it all together we find 
\begin{equation}\nonumber
\mathbb{E}_{U\sim \SPIN(2n)}[X_{i_1}\cdots X_{i_k}]
=
\langle\!\langle \tobs|^{\otimes k}
\,M^{(mk/2)}_{\SO(2n)}\,
\bigotimes_{j=1}^k |\tstate^{(i_j)}\rangle\!\rangle \,,
\end{equation}
and it directly follows that the moment vanishes immediately when $mk$ is odd, whereas it is non-zero if $mk$ is even. Hence the free-fermionic problem is reduced to understanding moment operators of $\SO(2n)$.

In the large-$n$ limit, the relevant moment operator admits an asymptotic Weingarten expansion over the Brauer algebra~\cite{garcia2023deep,braccia2026commutant}. Each Brauer diagram defines a pattern of contractions among the $k$ copies of the observable tensor and the $k$ copies of the state tensors. Here, one distinguishes Brauer elements by those leading to pairwise tensor matchings, where tensors are connected pairwise, and those giving rise to non-pairwise matchings, which couple three or more tensors at once. The pairwise matchings are precisely the contributions that generate products of covariances, whereas the non-pairwise ones encode deviations from Gaussianity.

The proof, therefore, reduces to showing that the pairwise tensor matchings dominate asymptotically. Concretely, this is ensured by a set of conditions derived in the Supplemental Information, where we find that the observable must have full support on the chosen subspace, and every non-pairwise contraction on the dataset must be subleading compared to the product of pairwise overlaps. We refer to this condition as ``Gaussian pairings dominating''.   Under these assumptions, the even moments take the form
\begin{align}
&\mathbb{E}_{U\sim \SPIN(2n)}[X_{i_1}\cdots X_{i_k}]\nonumber\\
&\quad=
\sum_{\sigma\in B_{k/2}}
\prod_{(j,j')\in \sigma}
\frac{m!}{(2n)^m}\,\|O\|_m^2\,
\langle \rho_{i_j},\rho_{i_{j'}}\rangle_m ,
\end{align}
which is exactly the Isserlis~\cite{isserlis1918formula} form for a zero-mean Gaussian process with ($k=2$) covariance
\begin{equation}
\kappa(\rho_i,\rho_j)
=
\frac{m!}{(2n)^m}\,\|O\|_m^2\,\langle \rho_i,\rho_j\rangle_m .
\end{equation}
This proves the even-moment part of the theorem.

At that stage, the only remaining obstruction is the odd moments. Here, we need to perform a specific case-by-case analysis depending on the value of $m$. This is precisely where the general case distinction summarized in Table~\ref{tab:FF_GP_scenarios} comes from. If $m$ is odd, then every odd moment vanishes trivially because $mk$ is odd. If $m=2$, the even-moment argument above is already enough once Gaussian pairings dominate. For higher even $m$, the third moment becomes the first nontrivial obstruction. After reshaping the anti-symmetric tensors $\tobs$ into square matrices $\tobsmat$, one finds a parity effect: when $m\equiv 2 \pmod 4$, the reshaped matrices are anti-symmetric, so $\operatorname{Tr}[O^3]=0$ and the third moment vanishes automatically; when $m\equiv 0 \pmod 4$, the reshaped matrices are symmetric, and odd moments do not vanish automatically, so extra structure is needed.

This is why Table~\ref{tab:FF_GP_scenarios} is organized first by the parity of $m$, and only then by concrete families of states and observables. The row labeled ``General'' states the sufficient conditions under which the moment argument closes. The row labeled ``Special cases'' collects explicit settings in which the Gaussian pairing domination check (for even moments) and the odd-moment analysis are either satisfiable, indefinite, or explicitly violated due to persistent non-Gaussian contributions. In this sense, we present Table~\ref{tab:FF_GP_scenarios} as a bookkeeping device that records, subspace by subspace, whether the general program above certifies QGP emergence.

A fully rigorous version of this argument is given in the Supplementary Information, where the asymptotic moment operator, the Brauer-algebra expansion, and the odd-moment analysis are derived explicitly.

        \subsection{Noisy kernel correction for QGPs}\label{sec:methods:noisy_kernels}

            Due to the fact that we can only extract information from quantum computers using a finite-shot budget,  the (symmetrically) measured training kernel matrix $\widetilde{K}_{\TC\TC}$ may not be positive semidefinite (PSD). Crucially, this phenomenon typically does not occur when using classical kernels, as these are computed exactly (up to numerical precision). Thus, the literature on PSD corrections in classical kernel methods is limited. In this section we review some standard methods, and discuss their use for QGPs.

            We begin by considering the geometric approach of finding the closest PSD matrix to the noisy kernel. Under the  Frobenius and spectral norms~\cite{boyd2004convex_refs}, this is optimally achieved by clipping negative eigenvalues of $\widetilde{K}_{\TC\TC}$ to zero. Despite the relative simplicity of this geometrically optimal method, it has the disadvantage that it leads to no uncertainty being propagated from the kernel. That its, the QGP assumes the clipped kernel to be correct with full certainty.

            An alternative standard approach to restore the PSD-ness of the kernel matrix is to shift its diagonal, and hence implicitly the entire eigenvalue spectrum, by the lowest (negative) eigenvalue $\lambda_{\min}$. This shift can be realized via adding $\lvert \lambda_{\min} \rvert I$ to the observation covariance $R$ in Eq.~\eqref{eq:joint_GP}. Hence, this PSD repair can be interpreted as increasing the observation variances independently by $\lvert \lambda_{\min} \rvert$, thereby establishing a proper propagation of uncertainty. However, it constitutes only a ``minimal uncertainty propagation'', just sufficient to ensure that the GP is well-defined. 

            Given the limitations of the previous approaches, we propose a PSD correction alternative to achieve robust uncertainty propagation. Specifically, one shifts the kernel diagonal (of $N$ training points) by the eigenvalue bound of the Wigner semicircle law  
            \begin{equation}
                \lambda_{\mathrm{Wigner}} = 2\sqrt{N}\sigma_{\kappa}
                ,
            \end{equation}
            where $\sigma^2_{\kappa}$ denotes the kernel entry measurement variance. 
            In contrast to the $\lambda_{\min}$-shift, this $\lambda_{\mathrm{Wigner}}$-shift approach is Bayesian. Shifting by $\lambda_{\min}$ is merely a point estimate (based on a single measured kernel matrix) and, hence, inherently frequentist. The $\lambda_{\mathrm{Wigner}}$-shift establishes the shift that is highly probable to restore the PSD property and becomes part of the generative prior. That is, when seen as an extension of $R$, it redefines the distribution of training observations (Eq.~\eqref{eq:GP2}) as 
            \begin{equation}\label{eq:GP2_wigner}
                \bm{y}_{\TC} \sim \NC\left( \bm{f}_\TC, \left(\sigma_{\rm obs}^2 + \lambda_{\mathrm{Wigner}} \right) I \right)
            \end{equation}
            before any measurements of either the observations $\bm{y}_\TC$ or the kernel matrix $K_{\TC\TC}$ have been performed. Clearly, the a priori uncertainty of training data is now a function of both shot budgets $\shotsobs$ and $\shotsker$.

            The Wigner semicircle law is applicable in the noisy kernel context when the noise is assumed to be additive, i.i.d., with zero mean and variance $\sigma^2_\kappa$. That is, $\widetilde{K} = K + W$ with $W$ is a corresponding Gaussian orthogonal ensemble (GOE). With high probability, $\lambda_{\rm Wigner}$ bounds the spectrum of $W$ (and of $\widetilde{K}$ from below via Weyl's inequality as $\lambda_{\min}(\widetilde{K}) \geq \lambda_{\min}(K) + \lambda_{\min}(W) \geq  \lambda_{\min}(W)$) and, almost surely, in the large-$N$ limit. Hence, both shifting techniques coincide asymptotically in the number of training points. However, for finite (potentially small) $N$, the $\lambda_{\rm Wigner}$-shift is more robust via a minimax prior with the lowest possible shift to fix the worst-case spectral noise of $\widetilde{K}$ with high probability.
            The Supplemental Information provides further details and a benchmark on the noisy kernel correction techniques.

        \subsection{Datasets and physical model systems}\label{sec:methods:models}

            This section defines the quantum datasets used in the main text, comprising both synthetic data and datasets constructed from the ground states of physical models.

            \subsubsection{Synthetic quantum datasets}\label{sec:methods:models:synthetic}

                The input states $\rho_t$ of the synthetic dataset are generated by sending the state $\ket{0}^{\otimes n}$ into a layer of single-qubit rotations $R_i(t)$, parametrized by $t\in \mathbb{R}$, and a subsequent ladder of nearest-neighbor CNOT gates. The rotation axes are initially picked at random along with scaling factors, i.e., $R_i(t) = \exp\left(-i\pi\alpha_i  t P_i\right)$ where $P_i$ is uniformly sampled from the set $\{X, Y, Z\}$ and $\alpha_i \sim \NC(0,1)$ for the $i$-th qubit. Two such layers are applied.                 No periodicity is present in the signal; conversely, the non-commensurable rotation frequencies in the quantum data (induced by $\alpha_i$) yield a highly non-periodic signal, making extrapolation stringent.

                A 2D extension (with $\bm{t} \in \mathbb{R}^2$) of this synthetic dataset is achieved through the additional parameter $t_2$ now dictating the entanglement in the input states $\rho_{\bm{t}}$, while the single-qubit rotations, as above, are parameterized by $t_1$. Concretely, instead of fixed CNOT gates, we apply a ladder of nearest-neighbor controlled-$Z$ rotations ($CR_Z$) that depend on $t_2$. This preparation is restricted to a single layer. 

            \subsubsection{Heisenberg bond-alternating XXX model}\label{sec:methods:models:xxx}

                The bond-alternating XXX model \cite{Kitazawa1996phase} corresponds to the following family of Hamiltonians
                \begin{align}\label{eq:xxx:hamiltonian}
                    H &= J_1\sum_{i=1,3,5,\ldots}\left(X_iX_{i+1}+Y_iY_{i+1}+Z_iZ_{i+1}\right)\\ &+ J_2\sum_{i=2,4,6,\ldots}\left(X_iX_{i+1}+Y_iY_{i+1}+Z_iZ_{i+1}\right)\,.
                \end{align}
                This model exhibits a phase transition (in the thermodynamic limit $n\rightarrow\infty$) when $J_1/J_2=1$. 
                We prepare a dataset $\DC=\{(\rho_i, y_i)\}_{i=1}^N$ of ground states $\rho_i$ and observations $y_i$ for $N=200$ different equispaced configurations $0<J_1/J_2\leq2$. The observations $y_i$ serving as the labels for the training of the QGP regression correspond to expectation values of the operator 
                \begin{equation}\label{eq:xxx:order_parameter_proxy}
                 O  = X_{\frac{n}{2}}X_{\frac{n}{2}+1} + Y_{\frac{n}{2}}Y_{\frac{n}{2}+1} + Z_{\frac{n}{2}}Z_{\frac{n}{2}+1},
                \end{equation}
                which is known to serve as a local proxy of the model's order parameter in the asymptotic limit. Notice that we can also express this operator as $\langle O \rangle = 2\langle {\rm SWAP}_{\frac{n}{2},\frac{n}{2}+1}\rangle - 1$.

            \subsubsection{Heisenberg XXZ model}\label{sec:methods:models:xxz}
                The specific XXZ model considered is represented by the Hamiltonian
                \begin{equation}\label{eq:xxz:hamiltonian}
                    H = \sum_{i=1}^n h^i Z_i-\sum_{i=1}^{n-1}J(X_iX_{i+1}+Y_iY_{i+1})+J_z Z_iZ_{i+1}\,,
                \end{equation}
                where $h^i$ is the magnetic field component at site $i$. Specifically, we consider alternating fields of the form
                \begin{equation}
                h^i=
                \begin{cases}
                    h_1 & \text{for}\; i \in \text{odd} \\ 
                    h_2  & \text{for}\; i \in \text{even}.
                \end{cases}
                \end{equation}
                Then, the observable corresponds to the total magnetization along $z$, which determines the magnetic ordering of the system
                \begin{equation}\label{eq:xxz:order_parameter}
                     O = \frac{1}{n}\sum_{i=1}^n  Z_i \,.
                \end{equation}
                The dataset, indexed by $(h_1,h_2)\in\mathbb{R}^2$, is created from finding the ground states of the model at different values of the magnetic fields.

            \subsubsection{Quantum sensing magnetometry}\label{sec:methods:models:qs}

                In magnetometry, a system interacts with a static magnetic field via the Hamiltonian
                \begin{equation}
                    H = \frac{\alpha}{2} \sum_{i=1}^n Z_i
                    ,
                \end{equation}
                where $\alpha$ represents the magnetic field magnitude. Our goal is to find the best probe state $\rho_t$ \cite{thurtell2022optimizing,ijaz2024more} that is maximally sensitive in the measured quantity with respect to changes in the magnetic field $\Delta \alpha$. In particular, we want to maximize the sensitivity when the observable of interest is 
                \begin{equation}\label{eq:bo:qs_observable}
                    \O = \sum_{i=1}^{n-1} X_iX_{i+1} \in \BC_2
                    .
                \end{equation}
                
                For the probe state, we consider squeezed spin states $\rho_t$ for $t \in [\pi/8, 15\pi/8]$, which can be easily prepared from $\rho_0 = \ketbra{0}{0}$ via one-axis twists \cite{kitagawa1993squeezed}, i.e., 
                \begin{equation}\label{eq:bo:probe_state}
                \rho_t=e^{t\sum_{i<j} X_i X_j}
                .
                \end{equation}
                The sensitivity of the probe state $\rho_t$ is quantified by the partial derivative magnitude
                \begin{equation}\label{eq:bo:qs_orig_objective}
                \left\lvert \frac{\partial \langle\O\rangle}{ \partial \alpha }\Big\rvert_{\alpha = \alpha_0} \right\rvert
                ,
                \end{equation}
                which is maximized (objective function) when identifying a desired probe state.

        \subsection{QGPs for classification}\label{sec:methods:gpc}

        To employ our QGP framework for classification tasks, we follow a standard approach that uses GPs for probabilistic classification \cite{rasmussen_GaussianProcessesMachine_2006}, in which predictions are obtained via class probabilities.
        For GP classification, we introduce a continuous latent function $z(t)$. This  formulates the discrete predictive distribution on the classification outcome $y \in \{0, 1\}$ for input $t$ as 
        \begin{equation}\label{eq:classification_posterior}
            \Pr(y \mid t) = \int \Pr\left(y \mid z \right) \Pr(z \mid t) dz
            ,
        \end{equation}
        where we omitted additional conditioning on the training dataset. For inputs $\TC$, this provides  binary labels $\bm{y}_{\TC}$, and here $z$ denotes the latent function at input $t$, i.e., $z = z(t)$.
        The first distribution in the integral of Eq.~\eqref{eq:classification_posterior} is a Bernoulli distribution over the class variable $y$ given the latent variable $z$ via a logistic sigmoid link function $\sigma(z) = 1/(1+\exp(-z))$ as
        \begin{equation}
            \Pr(y \mid z) = \sigma(z)^y (1 - \sigma(z))^{1-y}.
        \end{equation}
        The second distribution in the integral of Eq.~\eqref{eq:classification_posterior} is the predictive distribution of the latent variable. By marginalizing over the latent variables corresponding to the training inputs $\bm{z}_{\TC}$ this can be expanded as
        \begin{equation}
        \Pr(z \mid t) = \int \Pr(z \mid \bm{z}_{\TC}, t) \Pr(\bm{z}_{\TC} \mid \TC, \bm{y}_{\TC}) d\bm{z}_{\TC}
        ,
        \end{equation}
        where $\Pr(\bm{z}_{\TC} \mid \TC, \bm{y}_{\TC})$ is the posterior over the latent training variables \cite{bishop2006patternCORRECT}.
        While $\Pr(z \mid \bm{z}_{\TC}, t)$ is obtained via GP regression on the continuous latent function according to Eqs.~\eqref{eq:GP-posterior}-\eqref{eq:post-var}, a Laplace approximation of $\Pr(\bm{z}_{\TC} \mid \TC, \bm{y}_{\TC})$ yields a Gaussian distribution \cite{rasmussen_GaussianProcessesMachine_2006,bishop2006patternCORRECT}. 
        Now, with a Gaussian approximation for $\Pr(z \mid t)$, we can further approximate the integral in Eq.~\eqref{eq:classification_posterior}, which involves a standard approximation of the convolution of the logistic sigmoid and Gaussian \cite{bishop2006patternCORRECT}.
        Although our framework is implemented in Julia based on the \texttt{AbstractGPs.jl} package used for regression, the results were verified to match those of the GP classifier in the \texttt{scikit-learn} Python package \cite{pedregosa2011scikit}.

        In QGP regression, the derived quantum kernels feature a scaling factor dependent on the explicit physical observable $O$. Because this observable is absent in the classification setting, we instead renormalize the kernel function to appropriately scale the variance of the latent variable. This renormalization ensures that the latent function values fall within the active region of the logistic sigmoid link function, preventing the resulting probabilities from either remaining uncertain near $0.5$ or saturating prematurely at $0$ and $1$. Given the theoretical correspondence between the kernel scaling factor and the observable, this scaling can be physically interpreted as a property of a latent observable, whose expectation values define the latent continuous function. Importantly, the choice of this scaling factor does not alter the location of the decision boundary at $\Pr(y=0 \mid t) = 0.5$.
        Beyond supervised learning, the probabilistic formulation of QGPs naturally extends to unsupervised anomaly detection (one-class classification) by leveraging predictive uncertainties as anomaly scores~\cite{kemmler2013one}, a strategy that has proven effective for automated phase discovery in quantum systems~\cite{kottmann2020unsupervised}.

        \subsection{QGPs for optimization}\label{sec:methods:bo}

        In this section, we briefly outline the framework of BO and detail how QGPs naturally integrate into it as probabilistic surrogate models. For a comprehensive review of standard BO, we refer the reader to Refs.~\cite{shahriari2015taking}. BO is a surrogate optimization technique, which is based on the main assumption that the objective function $g$ is extremely costly to evaluate. For example, no analytical expression may be available as evaluations involve heavy simulations or even real experiments. Hence, we aim to maximize the objective function $g^* = \max g(\bm{t})$ and find the corresponding maximizer $\bm{t}^* \in \arg\max g(\bm{t})$ with as few objective function queries as possible. Global optimality is a common complementary goal of BO, besides evaluation cost efficiency.

        What makes BO ``Bayesian'' is the fact that one uses a surrogate function capable of quantifying its uncertainty, and which guides the optimization and iteratively update the probabilistic belief about the objective function as data from the proposed evaluations are observed (Bayes' theorem). Probabilistic surrogate functions can also adequately account for noisy evaluations, as in finite-shot or real experiments. Most commonly, GPs provide probabilistic surrogate functions by placing a probabilistic prior on the objective function. While we typically cannot hope for convexity/concavity in the objective function to facilitate optimization, we may very well rely on assumptions about the objective function or general task to impose an inductive bias in the surrogate modeling by encoding such in the prior of the GP.

        The so-called acquisition function $a(\bm{t})$  orchestrates the search for an optimum by proposing the next iterate to query via its maximum $\bm{t}_{\rm next} \in \arg\max_{\bm{t}} a(\bm{t})$. The key principle is that $a(\bm{t})$ quantifies an exploration-exploitation trade-off of the current surrogate function. Exploration aims to improve the model by querying in uncertain regions, a practice known as active learning. With an accurate surrogate function as a secondary goal, exploitation pushes toward the current model's optimum in accordance with the main goal of the optimization. 

        A standard acquisition function choice is the expected improvement \cite{Mockus1978,jonesEfficient1998}, given by
        \begin{equation}
            a_{\mathrm{EI}}(\bm{t}) = \E\left[I(\bm{t}, f, y^+)\right]\,,
        \end{equation}
        with the improvement (by at least $\xi$) over the best objective value observed thus far $y^+$ defined as $I(\bm{t}, f, y^+) = \max\{0, f(\bm{t}) - y^+ - \xi\}$. For a Gaussian surrogate ${f\sim\mathcal{GP}}$, the expectation can be analytically computed \cite{shahriari2015taking}.
        We use the expected improvement in the 2D synthetic optimization task with standard hyperparameter choice $\xi = 0.01$.
        For the quantum sensing task, squaring the QGP breaks its Gaussianity and yields a non-central $\chi^2$ surrogate.
        The closed-form mean and variance are derived from the underlying QGP \cite{johnson_ContinuousUnivariateDistributions_1994} as 
        \begin{align}
            \mu_{\chi^2} &= \mu_{\mathcal{GP}}^2 + \sigma_{\mathcal{GP}}^2 \,,\\
            \sigma^2_{\chi^2} &= 2 \sigma_{\mathcal{GP}}^2 (2\mu_{\mathcal{GP}}^2 + \sigma_{\mathcal{GP}}^2 )%
            .
        \end{align}
        To avoid complex numerical integration over the $\chi^2$ posterior, the simpler upper confidence bound acquisition function \cite{srinivas_GaussianProcessOptimization_2010} is adequate because it is solely defined in terms of the mean and variance as
        \begin{equation}
            a_{\mathrm{UCB}}(\bm{t}) = \mu(\bm{t}) + \beta\sigma(\bm{t})\,,
        \end{equation}
        with hyperparameter $\beta = 1.96$. While motivated by the 95\% confidence level of a Gaussian distribution, it loses this exact probabilistic interpretation in this $\chi^2$ regime and instead serves as an exploration-exploitation trade-off parameter.

        In both demonstrations, we optimize the acquisition function via grid search ($\Delta t = 0.06$ and $0.006$, respectively). This relies strictly on quantum kernel evaluations, which are considered computationally cheap relative to the objective function and can be cached across iterations. 
        Practical, potentially higher-dimensional, settings may benefit from a classical numerical optimization of the acquisition function. Additionally, we do not implement an automated stopping criterion and present the BO trajectories until the global optimum is reached within grid precision.

        In the quantum sensing problem, we are interested in sensing the presence of a weak magnetic field, i.e., perturbations around $\alpha_0 = 0$, which corresponds to Eq.~\eqref{eq:bo:qs_orig_objective} as an objective function. Since the expectation value is a function of both probe state twisting parameter $t$ and magnetic field $\alpha$, one could build a joint GP surrogate over $t$ and $\alpha$ and compute the derivative in $\alpha$, evaluated at $\alpha_0$, to obtain a GP in $t$ due to the linearity of differentiation. While the observable dependence with $t$ should be clearly modeled with a $\BC_2$ matchgate kernel, the dependence with $\alpha$ could be simply captured by a local linear kernel to model the first-order information in $\alpha$ of interest for the downstream differentiation. The joint kernel would then be of product form.
        We employ a simplified approach here: In the absence of the magnetic field ($\alpha = 0$) the observations for any $t$ vanish, since the observable and probe state (Eqs.~\eqref{eq:bo:qs_observable} and ~\eqref{eq:bo:probe_state}) commute and $\braket{0 | O |0} = 0$. Therefore, the expectation value for a small magnetic field (here $\Delta \alpha = 0.05$) directly provides the partial derivative of Eq.~\eqref{eq:bo:qs_orig_objective} evaluated in $\alpha_0 = 0$ under (forward) finite difference approximation. Consequently, the QGP surrogate can simply model the expectation values for $\alpha = \Delta \alpha$.

        \subsection{Numerical simulation}\label{sec:methods:numerical_simulation_details}

            In this section, we provide details on how the numerical experiments presented in the main text were performed. Specifically, we first describe the Matrix Product State (MPS) setup used to compute all exact quantities, such as the kernel entries and the training observation values employed in the QGP experiments. 
            We then explain how, starting from these  exact values, we simulated the shot noise that would arise when the required measurements are performed on a real quantum computer. In particular, we consider two distinct shot budgets per measurement, $\shotsobs$ and $\shotsker$, which control the precision of the observation measurements and the subspace-overlap measurements, respectively.
            
            \subsubsection{MPS framework}

                Since a thorough introduction to Tensor Networks and MPS methods is beyond the scope of this work, we refer the reader to the many introductory courses on the topic, see e.g.~\cite{orus2014practical}.
                
                Let us just recall that in the MPS framework, both the accuracy and the complexity of the computations are controlled by a single parameter known as the bond-dimension and denoted as $\chi$.
                All the operations between MPSs and MPOs, such as evolution and computation of expectation values can be performed through efficient algorithms that exploit local updates of the MPS tensors, and scale at worst as $\chi^3$. 
                By setting a maximum value threshold for $\chi$, one can approximate in a controlled way the calculations at hand. However, we stress that no such truncation was used in our numerical examples.
                
                In the synthetic regression example, we prepared a dataset of simple MPS states by applying a shallow $t$-dependent circuit to the all-zero MPS state. To implement the random matchgate circuit in combination with a $\BC_1$ Pauli measurement, the dataset states can directly be measured in the $2n$ Majorana (Pauli) operators forming a basis of $\BC_1$ via MPS-MPO contraction, and finally weighted by an initially randomly drawn coefficient vector of unit length.

                For the experiments where physical models were involved, we first prepared the ground state of the model at hand as an MPS by resorting to the density matrix renormalization group algorithm~\cite{schollwock2005density}. The latter works by sequentially updating the local tensors of a randomly initialized MPS sweeping from one side to the other until convergence in measured energy is attained. Then, observations are obtained by MPS-MPO contraction with the few Pauli operators appearing in the desired observable. 
                The subspace overlaps used to construct the kernels are obtained by two different techniques, depending on the subspace at hand, $\BC_2$ or $\BC_4$. In the former case, we measured all the Majorana (Pauli) operators in the subspace, obtaining vectors of coefficients of the states in $\BC_2$, and then built the kernel as the Gram matrix of those vectors. In the case of the subspace $\BC_4$ we directly extracted the subspace overlaps. To do so, we first promote the states to their density matrices and vectorize them to be encoded as MPSs, which squares the bond-dimension $\chi \to \chi^2$. Then, we compute the subspace overlap by computing the transition amplitude of the projector onto the subspace with respect to any ordered couple of states. The projector onto $B_4$ is conveniently encoded as an MPO of bond dimension $\chi=5$ which we developed in a similar way as for the projectors used in~\cite{braccia2024computing}.\footnote{Notice that this is not exclusive to the $B_4$ subspace, we developed MPO projectors onto any $\BC_m$ with bond dimension $\chi=\min(m,2n-m)+1$.}

                All the code we used to perform these simulation were built on top of the \texttt{ITensors.jl} Julia package~\cite{fishman2022itensor}.

            \subsubsection{Shot noise simulation} 

                While the MPS framework previously described allows us to compute all the exact quantities needed to perform our QGP experiments, we also considered the effect that performing those measurements on an actual quantum device would have. Specifically, while not considering the noise that could affect the preparation of the states and their evolution, we considered the statistical noise induced by measuring quantum expectation values with a finite number of shots. Particularly, we consider two different per-measurement shot budgets, $\shotsobs$ and $\shotsker$, respectively controlling the precision of the measurements of the observations and those of the subspace overlaps.

                For the latter, namely the kernel entries, we added shot-noise to the base exact overlaps using a worst-case Gaussian approximation. Namely, we add independent Gaussian fluctuations with a variance chosen from a conservative upper bound, assuming that the overlaps are measured via a Bell-measurement scheme.
                More precisely, we assume a uniform variance for all kernel-entry estimates given by
                \begin{equation}
                    \sigma_{\kappa}^2 = \frac{\binom{2n}{m}}{\shotsker}\,,
                \end{equation}
                for $n$ the number of qubits and $\BC_m$ the subspace at hand.
                Notice that every kernel entry is assumed to be measured with the same uncertainty, regardless of its actual value. Instead of modeling the variance of each overlap estimation, we assign to all entries the same shot-noise scale coming from the upper bound given by the ratio between the subspace dimension and the number of shots.
                Let us stress that this shot-noise model provides a simple and controlled way of injecting finite-shot noise into the kernel matrix.
                Ultimately then, the kernels used in our numerical experiments read
                \begin{equation}
                    \tilde{k}_{ij} = \tilde{k}_{ji} = k_{ij} + \varepsilon_{ij} \quad \text{with} \quad  \varepsilon_{ij} \sim \NC(0, \sigma^2_{\kappa})\,.
                \end{equation}

                As far as the observations are concerned, we instead employed a more refined shot-noise model. Namely, each Pauli measurement is treated as a binomial distribution over outcomes $\pm 1$. Operationally, for an exact expectation value $\mu \in [-1,1]$ and a finite shot budget $\shotsobs$, we model the corresponding estimator as the average of $\shotsobs$ independent single-shot outcomes distributed according to that expectation value. Indeed, a Pauli observable has only two possible outcomes, and therefore the variance depends on the underlying expectation value itself. In particular, measurements whose expectation value is close to $\pm 1$ are nearly deterministic and therefore almost noiseless, whereas measurements with expectation value close to $0$ are the noisiest.
                Practically, if $y_{\rm obs}$ is an exact Pauli expectation value, we replace it with
                \begin{equation}
                    \tilde y_{\mathrm{obs}} = \frac{2}{\shotsobs} X - 1\,, \quad X \sim \mathrm{Bin}\left(\shotsobs,\frac{1+y_{\mathrm{obs}}}{2}\right)\,,
                \end{equation}
                for $X$ the number of $+1$ outcomes in the $\shotsobs$ samples from the binomial distribution $\mathrm{Bin}\left(\shotsobs,\frac{1+y_{\mathrm{obs}}}{2}\right)$. The corresponding observation variances are given by
                \begin{equation}
                    \sigma_{\rm obs}^2 = \frac{1-y_{\mathrm{obs}}^2}{\shotsobs}\,.
                \end{equation}

    \section*{Data Availability}
        All data needed to evaluate the conclusions in this work are present in the main text and the supplementary materials. The raw data are deposited in the public code repository as specified below.

    \section*{Code Availability}
        The code implementing the QGP framework and to run the numerical simulations presented in this work is publicly available on GitHub at:
        \begin{center}
        \url{https://github.com/Jonas-Jaeger/provable-scalable-quantum-gaussian-processes}
        \end{center}

    \section*{Acknowledgments}
        We thank Andrea Palermo, Antonio Anna Mele, and Martín Larocca for useful discussions.

        J.J. was supported by the U.S. Department of Energy (DOE) through a quantum computing program sponsored by the Los Alamos National Laboratory (LANL) Information Science \& Technology Institute. 
        J.J. also acknowledges support from the Natural Sciences and Engineering Research Council (NSERC) of Canada, specifically the NSERC CREATE in Quantum Computing Program (grant number 543245).  
        P.Br. and P.Be. were supported by Laboratory Directed Research and Development (LDRD) program of LANL under project number 20260043DR. P.Be. acknowledges ongoing support from DIPC. 
        M.G.A. was supported by the U.S. Department of Energy (DOE) through a quantum computing program sponsored by the Los Alamos National Laboratory (LANL) Information Science \& Technology Institute.  
         D.G.-M. acknowledges financial support from the European Research Council (ERC) via the Starting grant q-shadows (101117138) and from the Austrian Science Fund (FWF) via the SFB BeyondC (10.55776/FG7). 
        M.C. was supported by LANL ASC Beyond Moore’s Law project.
       This work was also supported by the Quantum Science Center (QSC), a National Quantum Information Science Research Center of the U.S. Department of Energy (DOE).

    \section*{Author Contributions}
        The project was conceived by D.G.-M. and M.C.
        The theoretical results were derived and proven by J.J., P.Br., and P.Be., and checked by D.G.-M. and M.C. 
        The applications were implemented by J.J., P.Br., P.Be., and M.G.A., with leading efforts in algorithm implementations by J.J. and numerical simulations by P.Br. Extensions to classification and optimization were conceived by J.J. All authors analyzed and interpreted the data. All authors contributed to writing and revising the manuscript.

    \section*{Competing Interests}
        The authors declare no financial or non-financial competing interests.

    \bibliography{references,quantum}
    \makeatletter
\close@column@grid
\makeatother
\cleardoublepage
\newpage
\onecolumngrid
\renewcommand\appendixname{Supp. Info.}
\appendix

\renewcommand\figurename{Supplemental Figure}
\setcounter{figure}{0}
\setcounter{lemma}{0}
\setcounter{theorem}{0}
\setcounter{corollary}{0}

\renewcommand{\theHlemma}{app.\arabic{lemma}}
\renewcommand{\theHtheorem}{app.\arabic{theorem}}
\renewcommand{\theHcorollary}{app.\arabic{corollary}}

\numberwithin{figure}{section}

\newcounter{supfig}
\setcounter{supfig}{\value{figure}}

\newtheorem{supdefinition}{Supplemental Definition}
\newtheorem{supproposition}{Supplemental Proposition}
\newtheorem{suptheorem}{Supplemental Theorem}
\newtheorem{suplemma}{Supplemental Lemma}
\newtheorem{supcorollary}{Supplemental Corollary}

\section*{Supplemental Information for\texorpdfstring{\\}{ }``\textit{Provable and scalable quantum Gaussian processes for quantum learning}''}

Analogous to the main text, this supplemental information is divided into two parts. First, we present the theoretical framework and proofs for the matchgate quantum Gaussian processes (QGPs). Second, we provide additional details and numerical results regarding the applications of QGPs.

\section*{Contents}

\makeatletter

\newcommand\tableofsupplements{\@starttoc{app}}

\let\oldaddcontentsline\addcontentsline
\renewcommand{\addcontentsline}[3]{%
    \oldaddcontentsline{#1}{#2}{#3}%
    \def\@tempa{#1}\def\@tempb{toc}%
    \ifx\@tempa\@tempb
        \oldaddcontentsline{app}{#2}{#3}%
    \fi
}
\makeatother
\tableofsupplements

\FloatBarrier
\clearpage

    \section{Theory preliminaries}

        \subsection{Notation}\label{app:notation}

            A concise overview of the general notation used in this work is provided in the following:
            \begin{center}
            \begin{tabularx}{\textwidth}{l@{\hspace{1em}}X}
                \toprule
                \textbf{Symbol \& Expr.} & \textbf{Definition} \\
                \midrule
                $n$ & Number of qubits \\
                $m$ & Number of Majorana operators (index of matchgate/free-fermionic module $\BC_m$) \\
                $k$ & Moment order/ordinal \\
                $X_i, Y_i, Z_i$ & Pauli operators on $i$-th qubit \\
                $\openone$ & Identity (operator)\\
                $\DC = \{(\rho_i, y_i)\}_{i=1}^N$ & Training dataset of input states and observations \\
                $N$ & Number of training points (and input quantum states for proofs)\\
                $\rho_t$ & Input quantum state at time $t$ (e.g., for inference in QGP regression) \\
                $y_t$ & (True) Observed scalar output at time $t$ \\
                $\mu_H$ & Haar measure over group $\G$ \\
                $U$ & Unknown quantum evolution (Haar-random over group $\G$) \\
                $\O$ or $\hat{O}$& Observable for measurement \\
                $\G$ & Group of transformations (e.g., $\SPIN(2n)$, $\SO(d)$, etc.) \\
                $\kappa(\rho_t, \rho_{t'})$ or $\kappa_{\G}(\rho_t, \rho_{t'})$ & Kernel function for QGP, optionally with respect to group $\G$\\
                $C_\mu$ & Majorana fermionic operator ($\mu = 1, \ldots, 2n$) \\
                $C_{\idxset}$ & Product of $m$ Majorana operators for index set $\idxset \subseteq \{1,\dots,2n\}, \lvert \idxset \rvert = m$; order implicit \\
                $C_{\bm{\nu}}$ & Product of $m$ Majorana operators ordered by index vector $\bm{\nu} \in \{1,\dots,2n\}^m $ \\
                $\mathcal{B}$ & Complete linear operator space over $n$-qubits \\
                $\BC_m$ & Module (subspace) spanned by products of $m$ Majorana operators \\
                $A_{\BC_m}$ or $A_m$ & Projection of operator $A \in \BC$ onto module $\BC_m$ \\
                $\langle A, B \rangle_{{\BC}_m}$ or $\langle A, B \rangle_m$ & Module ($\BC_m$) overlap/inner product of $A, B \in \BC$\\
                $\|A\|_{\BC_m}^2$ or $\|A\|_m^2$ & Module purity/norm of $A \in \BC$ in $\BC_m$ \\
                $h_m$ & Prefactor $( -i )^{m(m-1)/2}$ to ensure Hermiticity of $m$-Majorana products \\
                $\delta_{\alpha, \beta}$ or $\delta_{\alpha}^{\beta}$ & Kronecker delta\\
                $\delta_{\alpha_1 \dots \alpha_m}^{\beta_1, \dots \beta_m}$ & Generalized (anti-symmetric) Kronecker delta of order $2m$ \\
                $S_\ell$ & Symmetric group consisting of all permutations of $\ell$ items \\
                $\frak{B}_{\ell}(\delta)$ or $\frak{B}_{\ell}$ & Brauer algebra spanned by all $\ell$ pairings of $2\ell$ items with (loop) dimension parameter $\delta$\\
                $R$ & Matchgate (spinor) representation of the special orthogonal group\\
                $F_\delta$ & (Orthogonal) Representation of the Brauer algebra $\frak{B}_{\ell}(\delta)$ on the $\ell$-fold tensor product Hilbert space of dimension $\delta^\ell$ (analogous definition for symmetric group) \\
                \bottomrule
            \end{tabularx}
            \end{center}
            
            \noindent
            More specifically, the derivations and proofs for the matchgate kernels involve precise notational definitions when arranging the coefficients of operator module projections for the states $\rho^{(i_j)}$ and observable $\O$ as compared here:
            
            \begin{center}
                \begin{tabularx}{\textwidth}{lcX}
                    \toprule
                     Symbol
                     & Vectorized & Definition \\
                     \midrule
                     $\tstate^{(i_j)}$,\, $\tobs$
                     & \xmark & Rank-$m$, $2n$-dimensional tensors of module frame (projection) coefficients\\
                     $\tstateelem_{\bm{\alpha}}^{(i_j)}, \, \tobselem_{\bm{\mu}}$ & \xmark & Scalar module frame coefficient at multi-index $\bm{\alpha}$ and $\bm{\mu}$ for the state $\rho_{i_j}$ and observable $O$, respectively\\
                     $\tstatemat^{(i_j)}, \, \tobsmat$  & \xmark & Matrix special case for $m=2$ or matrix reshape for $m>2$ of module frame coefficient tensors\\
                     $\tstatebb^{(i_j)}, \, \tobsbb$  & \xmark & Positive semidefinite squares $\tstatebb^{(i_j)} = \tstatemat^{(i_j)\top}\tstatemat^{(i_j)}$ and $\tobsbb = O^\top O$ of matrix special case for $m=2$ or matrix reshape for $m>2$ of module frame coefficient tensors\\
                     $\tstatevec^{(i_j)}, \, \tobsvec$  & \xmark & Vector special case for $m=1$ of module frame coefficient tensors\\
                     $\kket{\tstate^{(i_j)}}$,\, $\kket{\tobs}$& \cmark & $(2n)^m$-dimensional vectorizations of the module frame coefficient tensors\\
                     $\bigotimes_{j=1}^k \tstate^{(i_j)}$,\, $\tobs^{\otimes k}$ & \xmark & $k$-fold tensor products of module frame coefficient tensors \\
                     $\bigotimes_{j=1}^{k} \kket{\tstate^{(i_j)}}$,\, $\kket{\tobs}^{\otimes k}$ & \cmark & $(2n)^{mk}$-dimensional vectorizations of the $k$-fold tensor products of module frame coefficient tensors \\
                     $\MC_{\G}^{(t)}$ & \xmark & Twirl / moment superoperator of order $t$ over the group $\G$ with respect to the Haar measure $\mu_{\rm H}$ \\
                     $M_{\G}^{(t)}$ & \cmark & Moment operator of order $t$ over the group $\G$ with respect to the Haar measure $\mu_{\rm H}$ \\ 
                     \bottomrule
                \end{tabularx}
            \end{center}
            Note that for the univariate (i.e., single input state) case, the superscript indices are omitted, matching the notation for the single observable. Moreover, while the observable operator and its matrix representation of the module frame coefficient tensor $\tobs$ share the same symbol $O$, the difference should either be clear from context or the observable operator is explicitly denoted as $\hat{O}$.

    \subsection{The group of matchgate (or free-fermionic) transformations and its module decomposition}

        We begin by briefly recalling the definitions of the group of free-fermionic, or matchgate, unitaries~\cite{valiant2001quantum,jozsa2008matchgates}.
        The set
        \begin{equation}\label{eq:app:Majorana_ops}
            \begin{alignedat}{5}
                C_1      &= X_1 \,,           &\qquad C_3      &= Z_1 X_2\,,            &\qquad\ldots\qquad&& C_{2n-1} &= Z_1 \cdots Z_{n-1} X_n \,,\\
                C_2      &= Y_1 \,,           &\qquad C_4      &= Z_1 Y_2 \,,           &\qquad\ldots\qquad && C_{2n}   &= Z_1 \cdots Z_{n-1} Y_n\,,
                \end{alignedat}
        \end{equation}
        defines the $2n$ Majorana fermionic operators acting on the Hilbert space $\HC=(\mathbb{C}^2)^{\otimes n}$ of $n$ qubits (or $n$ fermionic modes under the Jordan-Wigner transformation).
        Majorana operators obey the anti-commutation relations $\left\lbrace C_\mu, C_\nu \right\rbrace = 2 \delta_{\mu\nu}$, yielding a Clifford algebra.
        
        The group of free-fermionic unitaries is a representation of the Lie group $\SPIN(2n) = e^{\mathfrak{g}}$, with associated \ac{dla} spanned by the products of two distinct Majorana operators,
        \begin{equation}
            \mathfrak{g} = \operatorname{span}_{\mathbb{R}}\left\lbrace C_\mu C_\nu  \right\rbrace_{1 \leq \mu < \nu \leq 2n} \cong \mathfrak{so}(2n)\,.
        \end{equation}
        This algebra is isomorphic to the special orthogonal Lie algebra $\mathfrak{so}(2n)$, and  it is generated through nested commutation by the set $\lbrace iZ_j\rbrace_{j = 1, \ldots, n}  \cup \lbrace iX_jX_{j+1} \rbrace_{j = 1, \ldots, n-1}$ (see e.g., ~\cite{diaz2023showcasing}), which corresponds to the typical generators for single- and two-qubit Pauli rotation gates in the context of matchgate circuits. The adjoint action of these circuits is a representation of the special orthogonal group $\SO(2n)$ (the matchgate representation), via the isomorphism
        \begin{equation}\label{eq:spinor_to_SO}
            U C_\mu U^\dagger = \sum_\nu q_{\mu \nu} C_\nu\,,
        \end{equation}
        with $q_{\mu\nu}$ the matrix elements of a matrix $Q \in \SO(2n)$ acting on the (real) vector space spanned by the $2n$ Majorana operators in Eq.~\eqref{eq:app:Majorana_ops}.

        Products of Majorana operators form an orthogonal basis with respect to the \ac{hs} inner product for the space of linear operators $\mathcal{B}$ acting on $\HC$. In particular, $\mathcal{B}$ decomposes into irreducible representations as~\cite{braccia2026commutant}
        \begin{equation} \label{eq:app:modules}
            \mathcal{B} = \bigoplus_{m = 0}^{2n} \mathcal{B}_m,
        \end{equation}
        where each of the $2n + 1$ subspaces $\mathcal{B}_m$ is spanned by the products of $m$ (distinct) Majorana operators, and $\mathcal{B}_0 = \operatorname{span}_{\mathbb{C}} \{I\}$. Hence, they have dimension $\dim(\BC_m)=\binom{2n}{m}$. 
        Importantly, for any $M_m \in \mathcal{B}_m$, one has that $VM_mV^\dagger \in \mathcal{B}_m$ for all $V\in \SPIN(2n)$ \cite{diaz2023showcasing}. Note that while the main text refers to these structures simply as subspaces, this supplemental material adopts the more rigorous algebraic term \emph{modules}.

        \subsection{Elementary properties of the modules\texorpdfstring{ $\BC_m$}{}}

        We consider two different spanning sets for each module $\mathcal{B}_m$, of which one is an orthonormal basis and one a frame. The technical difference is that the former is a set of linearly independent operators with respect to the \ac{hs} inner product, while the latter permits linear dependences. Even though both span the module subspace and are composed of products of $m$ distinct Majorana operators, they possess different properties\footnote{Note that for a single Majorana (module $\mathcal{B}_1$), both spanning sets are identical.} that can be beneficial depending on the specific use case.
        \medskip
        
        \noindent\textbf{Module basis.}
            We define the module basis as the following set of operators,
            \begin{equation}\label{supp:eq:module_basis}
                \left\lbrace 
                C_{\idxset }
                \right\rbrace_{\idxset }
                =
                \left\lbrace 
                \frac{1}{\sqrt{d}}
                h_m C_{\nu_1}  C_{\nu_2} \cdots  C_{\nu_m}
                \;\middle|\;
                \nu_1 < \nu_2 < \dots < \nu_m \in \VC \subseteq [2n] 
                \text{ with } \lvert \VC \rvert = m
                \right\rbrace\,,
            \end{equation}
            with $h_m = (-i)^{\frac{m\left( m - 1 \right)}{2}}$ and $d = 2^n$.
            This basis is orthonormal with respect to the \ac{hs} inner product, and is guaranteed to be Hermitian by the choice of the pre-factor $h_m$. This basis is minimal since its size matches the module dimension, $\mathrm{dim}\left(\mathcal{B}_m \right) = \binom{2n}{m}$. This is a consequence of the fact that we are only considering ordered products of Majorana operators. Indeed, permuting the Majorana operators within a product 
            can at most result in a minus sign due to the anti-commutation relations. 
            As the order is irrelevant, the basis elements can be identified by index sets $\VC \subseteq [2n] = \left\lbrace 1, \ldots, 2n \right\rbrace$ of size $\lvert \VC \rvert = m$, via the short-hand notation $C_{\idxset }$.
            \medskip
            
        \noindent\textbf{Module frame.}
            We define the module frame as the following multiset of operators,
            \begin{equation}\label{supp:eq:module_frame}
                \left\lbrace 
                C_{\bm{\nu}}
                \right\rbrace_{\bm{\nu}}
                =
                \left\lbrace 
                \begin{array}{ll}
                    \dfrac{1}{\sqrt{d \cdot m!}}
                    h_m 
                    C_{\nu_1}  C_{\nu_2} \cdots  C_{\nu_m} & \text{if all $\nu_j$ distinct}\\[1.25em]
                    0  & \mathrm{otherwise}
                \end{array}
                \;\middle|\;
                (\nu_1, \nu_2, \ldots, \nu_m) = \bm{\nu} \in [2n]^{m}
                \right\rbrace\,.
            \end{equation}
            While $h_m = (-i)^{\frac{m\left( m - 1 \right)}{2}}$ and $d = 2^n$ match the pre-factors in the basis definition above, an additional factor of ${1}/{\sqrt{m!}}$ appears in the definition of the module frame elements.             Again, this frame is guaranteed to be Hermitian by the choice of the pre-factor $h_m$. 
            In contrast to the basis, however, this frame is neither completely orthogonal nor normalized with respect to the \ac{hs} inner product and norm, and is technically to be understood as a multiset since duplicates may arise.

            The size of the frame is $(2n)^m$ since it contains an operator for every possible combination of $m$ Majorana operators. 
            These combinations are identified by index vectors $\bm{\nu} \in [2n]^{m}$ via the short-hand notation $C_{\bm{\nu}}$.
            The elements of the frame can be arranged in a rank-$m$ tensor of dimension $2n$ in each index, and so can the (scalar) projections of an operator onto these frame elements with respect to the \ac{hs} inner product. Crucially, these tensors are fully anti-symmetric in all their indices. 

            Because the module frame is overcomplete, a decomposition of an operator as a linear combination of the frame elements is not unique.
            As a convention, the following constraints are imposed on the coefficients of such decompositions: coefficients of frame elements involving products of the same Majorana operators must have equal magnitude, and coefficients associated with zero operators in the module frame are set to zero. These otherwise non-unique decompositions then yield completely anti-symmetric rank-$m$ tensors.

        \medskip
        
        Next, we recall the definitions of projections onto the modules, module overlaps (i.e., inner products in the module) and module purities (i.e., norms in the module). These definitions are formulated with respect to both the module basis and frame. 
        
        \noindent\textbf{Module projection.}
            Consider the space of linear operators $\mathcal{B}$ acting on $\mathbb{C}^{d}$ with $d = 2^n$, which can be decomposed into the modules $\mathcal{B}_m$ as in Eq.~\eqref{eq:app:modules}.
            Given an operator $A \in \mathcal{B}$, the orthogonal projection onto the modules is denoted as $A_{\mathcal{B}_m}$, or more concisely, $A_m$.
            The projection onto the module expanded using the module basis or module frame is given by
            \begin{equation}\label{supp:eq:projection_via_basis}
                A_{\mathcal{B}_m} 
                = 
                \sum_{\substack{\VC \subseteq [2n] \\\lvert \VC \rvert = m}} 
                a_{\idxset } C_{\idxset }
                \qquad
                \text{with}
                \qquad
                a_{\idxset }
                = 
                \Braket{C_{\idxset }, A} 
                =
                \tr\left[ C_{\idxset } A \right]
                \,,
            \end{equation}
            and
        \begin{equation}\label{supp:eq:projection_via_frame}
                A_{\mathcal{B}_m} 
                = \sum_{\bm{\nu}\in [2n]^m}
                a_{\bm{\nu}} C_{\bm{\nu}}
                \qquad
                \text{with}
                \qquad
                a_{\bm{\nu}} 
                = 
                \Braket{C_{\bm{\nu}}, A} 
                =
                \tr\left[ C_{\bm{\nu}} A \right]
                \,,
            \end{equation}
            respectively,
            where $\Braket{\cdot, \cdot}$ denotes the \ac{hs} inner product. Recall that $C_{\idxset }$ with the index set $\VC$ relates to the product of Majorana operators in the unique ascending order $\nu_1 < \cdots < \nu_m \in \VC$. In contrast, the index vector $\bm{\nu}$ determines the order of the Majorana operators explicitly.
            Note that obtaining the coefficients for the frame as $a_{\bm{\nu}} = \tr\left[ C_{\bm{\nu}} A \right]$ implicitly enforces that all coefficients corresponding to index vectors $\bm{\nu}$ with non-distinct entries are zero.
            Importantly, the sums in Eq.~\eqref{supp:eq:projection_via_basis} and Eq.~\eqref{supp:eq:projection_via_frame} run over $\binom{2n}{m}$ and $(2n)^m$ terms, respectively. Note as well that the coefficients with respect to both module basis and frame are real if and only if the decomposed or projected operator is Hermitian.
            
            The coefficients obtained by projections via the module basis relate to the coefficients of the module frame as
            \begin{equation}
                a_{\bm{\nu}} 
                = 
                \begin{cases}
                    \frac{1}{\sqrt{m!}}\,\mathrm{sgn}(\pi) \,a_{\left\lbrace \nu_1, \nu_2, \ldots, \nu_m \right\rbrace}    
                    & \text{if } \nu_1, \nu_2, \ldots, \nu_m \text{ distinct}\\
                    0 & \text{otherwise}\,,
                \end{cases}
            \end{equation}
            where $\mathrm{sgn}(\pi) \in \lbrace -1, 1 \rbrace$ denotes the sign of the permutation $\pi$ of the index vector $\bm{\nu} = \pi(\bm{\nu}')$,  with respect to the reference index vector $\bm{\nu}'$ with the same but ordered components $\nu_1' < \nu_2' < \cdots < \nu_m'$.
            Equivalently, the generalized (anti-symmetric) Kronecker delta $\delta$ yields a concise definition by combining both cases of distinct and non-distinct indices as
            \begin{equation} \label{eq:app:levi-civita}
                a_{\bm{\nu}} 
                = 
                \frac{1}{\sqrt{m!}} 
                a_{\left\lbrace \nu_{1}, \nu_{2}, \ldots, \nu_{m} \right\rbrace}   
                \delta_{\nu_1\phantom{'} \nu_2 \ldots \nu_m}^{\nu'_1 \nu'_2\dots\nu'_m} 
                \,,
            \end{equation}
            where $\bm{\nu}'$ is the ordered reference $\nu_{1}' < \nu_2' < \cdots < \nu_{{m}}'$.
            The generalized Kronecker delta is a rank-$m$ tensor and evaluates its components $\delta_{\alpha_1 \dots \alpha_m}^{\beta_1 \cdots \beta_m}$ as $1$ ($-1$) if the index values are distinct integers and $\alpha_1 \dots \alpha_m$ is an even (odd) permutation of $\beta_1 \cdots \beta_m$, and zero otherwise.
            Formally, generalized Kronecker delta can be defined in terms of the symmetric group $S_m$ of degree $m$ or an $m \times m$ determinant and reads as
            \begin{equation}\label{eq:generalized_antisym_kronecker}
                \delta_{\alpha_1 \dots \alpha_m}^{\beta_1 \cdots \beta_m} 
                = 
                \sum_{\pi \in S_m} \operatorname{sign}(\pi) \delta_{\alpha_{\sigma(1)}}^{\beta_1} \cdots \delta_{\alpha_{\sigma(m)}}^{\beta_m}
                = 
                \begin{vmatrix}
                \delta^{\beta_1}_{\alpha_1} & \dots & \delta^{\beta_1}_{\alpha_m} \\
                \vdots & \ddots & \vdots \\
                \delta^{\beta_m}_{\alpha_1} & \dots & \delta^{\beta_m}_{\alpha_m}
                \end{vmatrix}
                \,,
            \end{equation}
            reducing to the standard (two-index) Kronecker delta, which is $1$ if $\alpha = \beta$ and otherwise zero, denoted as $\delta_{\alpha}^{\beta} = \delta_{\alpha, \beta}$.
            
            Vice versa, the module frame coefficients determine the coefficients of projections in the module basis as 
            \begin{equation}
                a_{\idxset } = \frac{1}{\sqrt{m!}} \sum_{\pi \in S_m} \mathrm{sgn}(\pi) a_{\pi(\VC)}\,,
            \end{equation}
            where $\pi \in S_m$ denotes the possible permutations (on $m$ indices) such that $\pi(\VC)$ provides an index vector of the index set $\VC \subseteq \lbrace 1, \ldots, 2n\rbrace$ ($\lvert \VC \rvert = m$) in a certain fixed order. In common special cases, e.g., when the module frame coefficients were obtained by projecting into the module using all frame operators, the single frame coefficient with the index vector $\bm{\nu}'$ corresponding to ordered indices $\nu_1' < \nu_2' < \cdots < \nu_m' \in \VC$ determines the basis coefficient
            \begin{equation}
                a_{\idxset } = \sqrt{m!} \, a_{\bm{\nu}'} 
                \,.
            \end{equation}

        \medskip
            
        \noindent\textbf{Module overlap (module inner product).} The module overlap between two operators $A, B\in\BC$ is defined as the HS inner products of their module projections, i.e.,
            \begin{equation}\label{supp:eq:module_overlap}
                \Braket{A, B}_{\mathcal{B}_m} 
                =
                \Braket{A_{m}, B_{m}}
                =
                \tr\left[ A_{m}^\dagger B_{m} \right]\,,
            \end{equation}
            where $\Braket{\cdot, \cdot}$ denotes the \ac{hs} inner product. 
            Utilizing the decompositions of the projected operators into the orthonormal module basis and frame, given in Eqs.~\eqref{supp:eq:projection_via_basis} and \eqref{supp:eq:projection_via_frame} respectively, the module overlaps become
            \begin{equation}\label{supp:eq:module_overlap_coeffs}
                \Braket{A, B}_{\mathcal{B}_m} 
                =
                \Braket{\bm{a}, \bm{b}}
                =
                \sum_{\idxset } a_{\idxset}^* b_{\idxset}
                = 
                \sum_{\bm{\nu}} a_{\bm{\nu}}^* b_{\bm{\nu}}\,,
            \end{equation}
            where $\Braket{\cdot, \cdot}$ denotes the standard scalar product on complex vectors.            Notice that the sum $\sum_{\vec{\nu}}\vec{a}_{\vec{\nu}}^* b_{\vec{\nu}}$ corresponds to the full contraction of the coefficient tensors $\vec{a}_{\vec{\nu}}$ and $\vec{b}_{\vec{\nu}}$.

            Equivalently, the module overlap can be expressed as\begin{equation}\label{supp:eq:module_overlap_basis_op}
                \Braket{A, B}_{\mathcal{B}_m} 
                =
                \sum_{\idxset }\tr\left[C_{\idxset }^{\otimes 2} \left(A^\dagger \otimes B\right) \right]
                = 
                \tr\left[\left(\sum_{\idxset }C_{\idxset }^{\otimes 2}\right) \left(A^\dagger \otimes B\right) \right]
                = 
                \Braket{\sum_{\idxset }C_{\idxset }^{\otimes 2},\; A^\dagger \otimes B}_{\!\mathrm{HS}}
                .
            \end{equation}
            Analogously, in terms of the module frame operators,
            \begin{equation}\label{supp:eq:module_overlap_frame_op}
                \Braket{A, B}_{\mathcal{B}_m} 
                =\sum_{\bm{\nu}}\tr\left[C_{\bm{\nu}}^{\otimes 2} \left(A^\dagger \otimes B\right) \right]
                = 
                \tr\left[\left(\sum_{\bm{\nu}}C_{\bm{\nu}}^{\otimes 2}\right) \left(A^\dagger \otimes B\right) \right]
                = 
                \Braket{\sum_{\bm{\nu}}C_{\bm{\nu}}^{\otimes 2},\; A^\dagger \otimes B}_{\!\mathrm{HS}}
                .
            \end{equation}
            Note that, despite their direct connection to state overlaps, module overlaps can attain negative values. For instance, consider the two single-qubit states $\rho_{\pm} = \tfrac{1}{2}\left(I \pm X\right)$, which are orthogonal. 
            However, the projection onto $\BC_1$ does not preserve orthogonality and even results in a negative module overlap, as $\Braket{\rho_+, \rho_-}_{\BC_1} = -\tfrac{1}{2}$. 

        \medskip
            
        \noindent\textbf{Module purities (module norm).} The module purity of an operator $A\in\BC$ is defined as the squared HS inner norm of its module projection, i.e.,
            \begin{equation}\label{supp:eq:module_purity}
                \left\lVert A \right\rVert_{\mathcal{B}_m}^2 
                =
                \left\lVert A_{m} \right\rVert^2 
                =
                \Braket{A, A}_{\mathcal{B}_m}
                =
                \tr\left[ A_{m}^\dagger A_{m} \right]\,,
            \end{equation}
            where $\left\lVert \cdot \right\rVert $ denotes the \ac{hs} norm. 
            Alternative expressions for the module purity follow immediately from Eqs.~\eqref{supp:eq:module_overlap_coeffs}, \eqref{supp:eq:module_overlap_basis_op} and \eqref{supp:eq:module_overlap_frame_op}. Notably, since
            \begin{equation}\label{supp:eq:module_purities_coeffs}
                \left\lVert A \right\rVert_{\mathcal{B}_m}^2 
                =
                \lVert \bm{a} \lVert^2 
                =
                \sum_{\idxset } \lvert a_{\idxset} \rvert^2
                = 
                \sum_{\bm{\nu}} \lvert a_{\bm{\nu}} \rvert^2,
            \end{equation}
            the module frame is a Parseval frame, i.e., the norm of the projection coefficients in the module frame matches the norm of the operator in the module subspace. This implies that the module purity is upper bounded for quantum states as 
            $\left\lVert \rho \right\rVert_{\mathcal{B}_m}^2\leq 1$.

            Furthermore, we can apply the Cauchy-Schwarz inequality to obtain an upper bound for the module overlap,
            \begin{equation}
                \left|\Braket{A, B}_{\mathcal{B}_m} \right|\leq \left\lVert A \right\rVert_{\mathcal{B}_m}^2 \left\lVert B \right\rVert_{\mathcal{B}_m}^2  \,.
            \end{equation}
            In particular, for quantum states $\rho,\rho'$, the absolute value of their module overlap is bounded as $\Braket{\rho, \rho'}_{\mathcal{B}_m} \leq d_m/d$.

            \subsection{The Haar measure and twirl superoperator}

             The normalized Haar measure $\mu_{\mathrm{H}}$ on a semisimple Lie group $\G$       is uniquely defined by its left-and-right-invariance property,
        \begin{equation}\label{app:eq:left_right_invariance_haar}
            \int_{\G} f(U) d\mu_{\mathrm{H}}(U) 
            = \int_{\G} f(V U) d\mu_{\mathrm{H}}(U) 
            = \int_{\G} f(U V) d\mu_{\mathrm{H}}(U)\,, 
        \end{equation}
        for any integrable function $f(U)$ and $V \in \G$.
        For the sake of brevity, we usually omit the label $\mu_{\mathrm{H}}$.

            The $t$-th moment superoperator $\MC_{\G}^{(t)} (\cdot)$ induced by the adjoint action under the Haar measure, referred to as the twirl, is given by~\cite{mele_IntroductionHaarMeasure_2024}
            \begin{equation}
                \MC_{\G}^{(t)} (\cdot) = \E_{U \sim_{\mu_{\rm H}} \G} \left[ U^{\otimes t} (\cdot) (U\ad)^{\otimes t} \right]
                .
            \end{equation}
            This linear superoperator is a map between linear operators on $\HC^{\otimes t}$, and can be vectorized to yield the (vectorized) moment operator $M_{\G}^{(t)}$, as
             \begin{equation}
                M_{\G}^{(t)}
                =
                \E_{U\sim_{\mu_{\rm H}} \G}\!\left[ U^{\otimes t} \otimes (U^*)^{\otimes t} \right] \label{eq:app:moment-operator}
                .
            \end{equation}
            Th moment operator $M_{\G}^{(t)}$ acts on vectors of dimension ${2^{2nt}}$, which are the vectorizations of linear operators acting on $\mathcal{H}^{\otimes t}$.

\clearpage
        \section{Theoretical framework for matchgate Gaussian processes}\label{sec:supp:theoretical_framework}

        In this section, we develop the framework for analyzing the emergence of  Gaussian processes under random matchgate unitaries. In particular, we provide the proof for Theorem~\ref{thm:supp:GP}.

        We  begin by introducing the random variables\footnote{Note the common convention of using upper-case letters to denote random variables, which should not be confused with upper-case letters used to denote operators or matrices. The random variables studied here have scalar-valued realizations.} 
        \begin{equation} \label{eq:app:random_variables}
            X_{i} = \tr\left[U \rho_i U^\dagger \O\right]\,,
        \end{equation}
        with $U$ sampled from the Haar measure on $\SPIN(2n)$~\cite{braccia2025optimal}.
        Our goal is to analyze the distributions of          the expectation values in Eq.~\eqref{eq:app:random_variables} over quantum states $\{\rho_i\}_i$ that have been evolved under some free-fermionic transformation $U$, and determine conditions for which they converge to Gaussian processes under the prior that $U$ had been uniformly sampled from $\SPIN(2n)$, i.e. according to the latter's Haar measure.
        Hence, we are interested in computing the $k$-th moments
        \begin{equation}\label{eq:app:k_moment}
                \E_{U\sim \SPIN(2n)}\left[ X_{i_1} X_{i_2} \cdots X_{i_k} \right] 
                = \E_{U\sim \SPIN(2n)}\left[ \prod_{j=1}^k \tr\left[U \rho_{i_j} U^\dagger \O\right] \right]\,,
            \end{equation}
             corresponding to the index multiset $\lbrace i_j \rbrace_{j=1}^k$ identifying the random variables $X_{i_j}$ (and states $\rho_{i_j}$). We will consider both the univariate (single input state) and multivariate (multiple input states) cases.
           For univariate distributions, where a single random variable $X$ is analyzed, the moment of order $k$ reduces to
            $\E\left[X^k\right]$. For multivariate distributions of $N$ random variables or, equivalently, $N$-dimensional random vectors, 
            we will collect the random variables in a length-$N$ vector  $\left( X_1, X_2, \dots , X_N \right)^\top$.

        We restrict our derivations to an observable $\O$ from a specific module $\BC_m$ ($m \neq 0, 2n$) , i.e., a linear combination of products of $m$ distinct Majorana operators. Consequently, the primary mathematical objects in the derivation of this theoretical framework are the module frame projection coefficients for the states $\rho_i$ and the observable $O$ (as defined in Eq.~\eqref{supp:eq:projection_via_frame}). These coefficients are organized into rank-$m$ tensors, denoted by $\tstate^{(i)}$ and $\tobs$, respectively. Each of the $m$ tensor indices, collected in the index vector $\bm{\nu}$, corresponds to a single Majorana index within the frame operator product $C_{\bm{\nu}}$. For the cases of tensor rank $m=2$ and $m=1$, these coefficient tensors simplify to matrices (denoted $\tstatemat^{(i)}$ and $\tobsmat$) and vectors (denoted $\tstatevec^{(i)}$ and $\tobsvec$), respectively. Further details regarding notation of tensors and their form variants are provided in Supp.~Info.~\ref{app:notation}.
        
        The derivation is split into four parts:
        \begin{enumerate}
            \item We derive general expressions for the moments in Eq.~\eqref{eq:app:k_moment}, connecting them to the standard representation of the special orthogonal group $\SO(2n)$.
           Importantly, we either obtain vanishing moments or recover the (vectorized) $mk/2$-th moment operator with respect to $\SO(2n)$ \textrightarrow{}~Supp.~Info.~\ref{app:theory:1}.

            \item We employ asymptotic Weingarten calculus \cite{mele_IntroductionHaarMeasure_2024,garcia2023deep,garcia2024architectures} to analyze the moment operator in terms of Brauer algebras and in the asymptotic large-$n$ limit \textrightarrow{}~Supp.~Info.~\ref{app:theory:2}.

            \item We establish formal conditions on the observable $\O$ and input states $\rho_i$ under which the random variables $X_i$ converge to a Gaussian process (Theorem~\ref{thm:supp:GP}) or not (Theorem~\ref{thm:supp:no_GP}) \textrightarrow{}~Supp.~Info.~\ref{app:theory:3}.

            \item We provide equivalent and sufficient formulations of the conditions established in the previous step, as well as general implications and simplifications of these conditions for certain modules $\BC_m$ \textrightarrow{}~Supp.~Info.~\ref{app:theory:4}.
        \end{enumerate}
        Supplemental Info.~\ref{app:examples_families_GPs} identifies concrete families of states and observables that obey or do not obey Gaussian processes under Haar-random matchgate transformations.

        \subsection{Vectorization reveals the moment operator of the special orthogonal group}\label{app:theory:1}

            We will now connect the moments in Eq.~\eqref{eq:app:k_moment} to the standard representation of $\SO(2n)$, as stated in the following useful lemma. This will allow us to leverage well-known Weingarten-calculus results from the special orthogonal group in the next subsections.

            \begin{lemma}
                Consider an 
                observable $\O \in \mathcal{B}_m$ and a set of input states $\left\lbrace \rho_i \right\rbrace_{i=1}^N$. The $k$-th moments of the joint distribution of random variables $X_i= \tr[U \rho_i U^\dagger \O]$ are given by 
\begin{equation}\label{eq:multi_majo_multi_moment_op_vectorizations}
\E_{U\sim \SPIN(2n)}\left[ X_{i_1} X_{i_2} \cdots X_{i_k} \right]  =  \langle\!\bra{\tobs}^{\otimes k} \mathbb{E}_{Q\sim \SO(2n)}\left[Q^{\otimes mk}\right]\bigotimes_{j=1}^{k} \ket{\tstate^{(i_j)}}\!\rangle \,,
                \end{equation}
                where $\ket{\tobs}\!\rangle=\sum_{\bm{\mu}}  \tobselem_{\bm{\mu}} \Ket{\bm{\mu}}$ and $\ket{\tstate^{(i_j)}}\!\rangle=\sum_{\bm{\alpha}} \tstateelem^{(i_j)}_{\bm{\alpha}}  \ket{\bm{\alpha}}$ are the vectorization (in $\mathbb{R}^{2nm}$) of the tensors $\tobs$ and $\tstate^{(i_j)}$ containing the module frame projection coefficients $\tobselem_{\vec{\mu}}$ and $\tstateelem^{(i_j)}_{\vec{\alpha}}$, respectively.
            \end{lemma} 
            \begin{proof}
                Employing the frame decomposition of the observable 
                $\O = \O_m = \sum_{\bm{\mu}} \tobselem_{\bm{\mu}} C_{\bm{\mu}}$ in Eq.~\eqref{eq:app:k_moment}
                yields
                \begin{equation}
                    \E_{U\sim \SPIN(2n)}\left[ X_{i_1} X_{i_2} \cdots X_{i_k} \right]
                    = 
                    \E_{U\sim \SPIN(2n)}\left[ \prod_{j=1}^k \sum_{\bm{\mu} \in [2n]^m}
                    \tobselem_{\bm{\mu}} 
                    \tr\left[\rho_{i_j} U 
                    C_{\bm{\mu}} 
                    U^\dagger \right] \right]
                    .
                \end{equation}
                The matchgate representation, as defined in Eq.~\eqref{eq:spinor_to_SO}, now directly relates the above expression to the moments over the special orthogonal group $\SO(2n)$, as follows
                \begin{equation}
                    UC_{\mu_1}C_{\mu_2}\cdots C_{\mu_m}U^\dagger = UC_{\mu_1}U^\dagger UC_{\mu_2}^\dagger U \cdots C_{\mu_m}U^\dagger = \left(\sum_{\alpha_1=1}^{2n} q_{\mu_1 \alpha_1} C_{\alpha_1}\right) \left(\sum_{\alpha_2=1}^{2n} q_{\mu_2 \alpha_2} C_{\alpha_2}\right)\cdots \left(\sum_{\alpha_m=1}^{2n} q_{\mu_m \alpha_m} C_{\alpha_m}\right) \,,
                \end{equation}
                which implies             \begin{equation}\label{eq:multi_moment_mulit_majo_spinor_to_SO}
                    \E_{U\sim \SPIN(2n)}\left[ X_{i_1} X_{i_2} \cdots X_{i_k} \right]
                    = 
                    \E_{Q\sim \SO(2n)}\left[ \prod_{j=1}^k \left(\sum_{
                    \bm{\mu}, \bm{\alpha} \in [2n]^m} \tobselem_{\bm{\mu}} \, q_{\mu_1 \alpha_1} \cdots q_{\mu_m \alpha_m} \tr\left[\rho_{i_j} C_{\bm{\alpha}} \right] \right) \right]\,.                    
                \end{equation}
                The expectation values $\tr\left[\rho_{i_j} C_{\bm{\alpha}} \right]$
                are the coefficients in the projection of the state $\rho_{i_j}$ onto the module $\mathcal{B}_m$ via the module frame (see  Eq.~\eqref{supp:eq:projection_via_frame}), which we denote  $\tstateelem_{\bm{\alpha}}^{(i_j)} = \tr\left[\rho_{i_j} C_{\bm{\alpha}} \right]$. We  collect them  in a (real) anti-symmetric rank-$m$ tensor $\tstate^{(i_j)}$.                 Analogously, the coefficients $\tobselem_{\bm{\mu}}$ of the observable $\O$ are collected in another (real) anti-symmetric rank-$m$ tensor $\tobs$. 
      
                Next, we notice that $q_{\mu_l \alpha_{l}}=\Braket{\mu_l | Q |\alpha_{l}} $, where $|\mu_l\rangle$ and $|\alpha_{l}\rangle$ are canonical basis vectors of $\mathbb{R}^{2n}$ on which the fundamental representation of $\SO(2n)$ acts\footnote{It is important to not confuse the basis vectors $\ket{\mu_l}$ and $\ket{\alpha_l}$ corresponding to $\SO(2n)$ with computational (qubit) basis states, as the basis is not binary $\{\ket{0},\ket{1}\}$ but composed of $2n$ basis elements $\{C_1, \ldots, C_{2n}\}$.}. This allows us to write
                \begin{align}
                    q_{\mu_1 \alpha_1} \cdots q_{\mu_m \alpha_m} 
                    &= \prod_{l=1}^m q_{\mu_l \alpha_{l}} 
                    = \prod_{l=1}^m \Braket{\mu_l | Q |\alpha_l}\nonumber \\
                    &= \Tr\left[\left( \bigotimes_{l=1}^m \bra{\mu_l} \right)  Q^{\otimes m} \left( \bigotimes_{l=1}^m \ket{\alpha_l} \right)\right] \label{eq:product_trace_trick}\\
                    &=  \Braket{\bm{\mu} | Q^{\otimes m} | \bm{\alpha}}\,, \label{eq:product_q_elems_vectorized}
                \end{align}
                where Eq.~\eqref{eq:product_trace_trick} can be seen as a consequence of the trace trick $a = \tr[a]$ for any scalar $a$, and the                 identity $\tr[A \otimes B] = \tr[A]\tr[B]$. Furthermore, we introduced the notation $\ket{\vec{\mu}}=\bigotimes_{l=1}^{m} \ket{\mu_l}$ (and analogously for $\ket{\vec{\alpha}}$).
                Plugging Eq.~\eqref{eq:product_q_elems_vectorized} into~\eqref{eq:multi_moment_mulit_majo_spinor_to_SO} leads to 
               \begin{align}
                    \E_{U\sim \SPIN(2n)}\left[ X_{i_1} X_{i_2} \cdots X_{i_k} \right] &=  \E_{Q\sim \SO(2n)}\left[ \prod_{j=1}^k \left(\sum_{\bm{\mu}, \bm{\alpha} \in [2n]^m} \tobselem_{\bm{\mu}} \, q_{\mu_1 \alpha_1} \cdots q_{\mu_m \alpha_m} \rho^{(i_j)}_{\bm{\alpha}}\right) \right] \nonumber \\ 
                    &= \E_{Q\sim \SO(2n)}\left[ \prod_{j=1}^k \left(\sum_{\bm{\mu}, \bm{\alpha} \in [2n]^m} \tobselem_{\bm{\mu}} \Braket{\bm{\mu} | Q^{\otimes m} | \bm{\alpha}} \rho^{(i_j)}_{\bm{\alpha}}\right)\right] \nonumber \\ 
                    & =  \E_{Q\sim \SO(2n)}\left[ \prod_{j=1}^k \left(\left(\sum_{\bm{\mu}\in [2n]^m}  \tobselem_{\bm{\mu}} \bra{\bm{\mu}} \right)  Q^{\otimes m}  \left( \sum_{\bm{\alpha}\in [2n]^m} \rho^{(i_j)}_{\bm{\alpha}}  \ket{\bm{\alpha}} \right)\right)\right] \nonumber \\ &= \E_{Q\sim \SO(2n)}\left[  \left(\sum_{\bm{\mu}\in [2n]^m}  \tobselem_{\bm{\mu}} \bra{\bm{\mu}} \right)^{\otimes k}  Q^{\otimes mk}  \left( \bigotimes_{j=1}^k \sum_{\bm{\alpha}\in [2n]^m} \rho^{(i_j)}_{\bm{\alpha}}  \ket{\bm{\alpha}} \right)\right]\,.\label{eq:vectorization_y_tensor_single_X} 
                \end{align}

                Finally, using the notation $\ket{\tobs}\!\rangle=\sum_{\bm{\mu}}  \tobselem_{\bm{\mu}} \Ket{\bm{\mu}}$ and $\ket{\tstate^{(i_j)}}\!\rangle=\sum_{\bm{\alpha}} \tstateelem^{(i_j)}_{\bm{\alpha}}  \ket{\bm{\alpha}}$ for the vectorization of the tensors $\tobs$ and $\tstate^{(i_j)}$, respectively, we end up with
                \begin{equation}
                    \E_{U\sim \SPIN(2n)}\left[ X_{i_1} X_{i_2} \cdots X_{i_k} \right] = \E_{Q\sim \SO(2n)}\left[ \langle\!\bra{\tobs}^{\otimes k} Q^{\otimes mk}\bigotimes_{j=1}^{k} \ket{\tstate^{(i_j)}}\!\rangle \right] = \ \langle\!\bra{\tobs}^{\otimes k} E_{Q\sim \SO(2n)}\left[Q^{\otimes mk}\right]\bigotimes_{j=1}^{k} \ket{\tstate^{(i_j)}}\!\rangle \,.\nonumber
                \end{equation}
            \end{proof}

            \noindent
            In Supp.~Fig.~\ref{fig:tensor_network_diagram_vectorization} we present a tensor-network diagram visualization of Eq.~\eqref{eq:multi_majo_multi_moment_op_vectorizations}, which will be very useful for our purposes.
            \begin{figure}[t]
                \centering
                \includegraphics[width=0.75\linewidth]{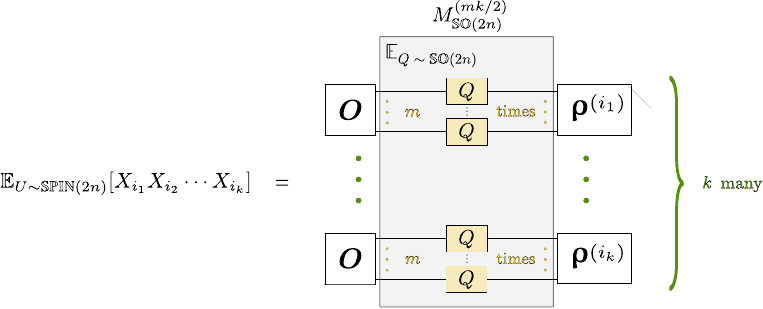}
                \caption{\textbf{Tensor network diagram for moment expression.} Depiction in vectorized form with Haar-random $Q \sim \SO(2n)$ as in Eq.~\eqref{eq:multi_majo_multi_moment_op_vectorizations}, which yields the moment operator as in Eqs.~\eqref{eq:supp:expectation_via_moment_op} and ~\eqref{supp:eq:moment_operator_definition}.}
                \label{fig:tensor_network_diagram_vectorization}
            \end{figure}

            Next, we show that the moments in Eq.~\eqref{eq:multi_majo_multi_moment_op_vectorizations} either vanish (when $mk$ is odd) or else correspond to the action of the (vectorized) $\frac{mk}{2}$-th moment operator in $\SO(2n)$ as defined in Eq.\eqref{eq:app:moment-operator}.
  
            \begin{lemma}\label{lemma:vanish_vs_moment_op}
                The $k$-th moments~\eqref{eq:app:k_moment} either \emph{a)} vanish
                \begin{equation}
                    \E_{U\sim \SPIN(2n)}\left[ X_{i_1} X_{i_2} \cdots X_{i_k} \right] = 0 \qquad \text{\emph{if $mk$ is odd}}\,,
                \end{equation}
                or \emph{b)} can be expressed through the (vectorized) $mk/2$-th twirl operator in $\SO(2n)$ w.r.t. the Haar measure,
                \begin{equation}\label{eq:supp:expectation_via_moment_op}
                    \E_{U\sim \SPIN(2n)}\left[ X_{i_1} X_{i_2} \cdots X_{i_k} \right] = \langle\!\bra{\tobs}^{\otimes k} M_{\SO(2n)}^{(mk/2)}\bigotimes_{j=1}^{k} \ket{\tstate^{(i_j)}}\!\rangle \qquad \text{\emph{if $mk$ is even}}\,,
                \end{equation}
                where                \begin{equation}\label{supp:eq:moment_operator_definition}
                    M_{\SO(2n)}^{(mk/2)} = 
                    \E_{Q\sim \SO(2n)}\!\left[ Q^{\otimes mk/2} \otimes Q^{\otimes mk/2} \right].
                \end{equation}
            \end{lemma}
            \begin{proof}
                First, if both the moment $k$ and the number of Majoranas $m$ are odd, then  $mk$ is odd, and this immediately yields vanishing moments,
                \begin{align}
                    \E\left[ X_{i_1} X_{i_2} \cdots X_{i_k} \right]
                    &= \langle\!\bra{\tobs}^{\otimes k} \left( \int_{\SO(2n)} Q^{\otimes mk} d\mu_H(Q) \right)\bigotimes_{j=1}^{k} \ket{\tstate^{(i_j)}}\!\rangle\nonumber\\
                    &= \langle\!\bra{\tobs}^{\otimes k} \frac{1}{2}\left( \int_{\SO(2n)} Q^{\otimes mk} d\mu_H(Q) + \int_{\SO(2n)} Q^{\otimes mk} d\mu_H(Q) \right)\bigotimes_{j=1}^{k} \ket{\tstate^{(i_j)}}\!\rangle\nonumber\\
                    &= \langle\!\bra{\tobs}^{\otimes k} \frac{1}{2}\left( \int_{\SO(2n)} Q^{\otimes mk} d\mu_H(Q) + \int_{\SO(2n)} (-IQ)^{\otimes mk} d\mu_H(Q) \right)\bigotimes_{j=1}^{k} \ket{\tstate^{(i_j)}}\!\rangle\label{eq:multi_majo_odd_moment_multi_variate_left_invar}\\
                    &= \langle\!\bra{\tobs}^{\otimes k} \frac{1}{2}\left( \int_{\SO(2n)} Q^{\otimes mk} d\mu_H(Q) + \int_{\SO(2n)} (-1)^{mk}Q^{\otimes mk} d\mu_H(Q) \right)\bigotimes_{j=1}^{k} \ket{\tstate^{(i_j)}}\!\rangle\nonumber\\
                    &= \langle\!\bra{\tobs}^{\otimes k} \frac{1}{2}\left( \int_{\SO(2n)} Q^{\otimes mk} d\mu_H(Q) - \int_{\SO(2n)} Q^{\otimes mk} d\mu_H(Q) \right)\bigotimes_{j=1}^{k} \ket{\tstate^{(i_j)}}\!\rangle\label{eq:multi_majo_odd_moment_multi_k_odd}\\
                    &= 0\,.
                \end{align}
                Here, Eq.~\eqref{eq:multi_majo_odd_moment_multi_variate_left_invar} utilizes the left-invariance  of the Haar measure $\mu_H$ (see Eq.~\eqref{app:eq:left_right_invariance_haar}) under left multiplication by an element of the group, which is $-I \in \SO(2n)$ in this case. In Eq.~\eqref{eq:multi_majo_odd_moment_multi_k_odd}, we used that $mk$ is odd, so that $(-1)^{mk} = -1$. This proves case \emph{a)}.
    
                For case \emph{b)}, we notice that if at least the moment $k$ or the number of Majoranas $m$ is even,  $mk$ is even, and the expectation in Eq.~\eqref{eq:multi_majo_multi_moment_op_vectorizations} can be re-written as 
                \begin{equation}
                    \E_{Q\sim \SO(2n)}\!\left[ Q^{\otimes mk} \right] = \E_{Q\sim \SO(2n)}\!\left[ Q^{\otimes mk/2} \otimes (Q^*)^{\otimes mk/2} \right] = M_{\SO(2n)}^{(mk/2)},
                \end{equation}
                which is the $mk/2$-th moment (super-)operator, i.e., the vectorized version of the $mk/2$-th twirl~\cite{mele_IntroductionHaarMeasure_2024}. This is possible as $Q$ is real, i.e., $Q\in\SO(2n)$, and hence $Q^* = Q$.
            \end{proof}

        \subsection{Asymptotic moments via Weingarten calculus}\label{app:theory:2}

            Now that we have connected the moments in Eq.~\eqref{eq:app:random_variables} to the moment super-operator $M_{\SO(2n)}^{(mk/2)}$, we can leverage standard Weingarten-calculus techniques to compute them. Let us review the necessary machinery.

            First we recall that the moment super-operator $M_{\SO(2n)}^{(mk/2)}$ is a projector onto the $mk/2$-fold commutant of $\SO(2n)$, that is, the vector space spanned by the set of matrices 
            \begin{equation}
                \{M \;|\; MQ^{\otimes mk/2} = Q^{\otimes mk/2}M \quad\forall Q\in\SO(2n)\}\,.
            \end{equation}
            In particular, the $mk/2$-fold commutant of $\SO(2n)$ consists of a representation of the Brauer algebra $\mathfrak{B}_{mk/2}(2n)$~\cite{collins2017weingarten_refs} when $\frac{mk}{2} < n$. When $\frac{mk}{2}\geq n$, an additional generator  must be added to those of $\mathfrak{B}_{mk/2}(2n)$ to obtain the commutant~\cite{grood1999brauer}. Note that this additional generator does not appear in the commutant of the orthogonal group $\mathbb{O}(2n)$.

            The Brauer algebra $\mathfrak{B}_{mk/2}(2n)$ is a vector space spanned by a basis consisting of all possible pairings of a set of size $mk$, endowed with an additional bilinear operation (product). That is, given a set of $mk$ items, the basis elements of the Brauer algebra correspond to all possible ways of splitting them into pairs. 
            This implies that all the permutations  in the symmetric group $S_{mk/2}$ are also in $\mathfrak{B}_{mk/2}(2n)$, as these correspond to the pairings that can only connect the first $mk/2$ items to the remainder ones. A straightforward calculation reveals that there are  $\frac{(mk)!}{2^{mk/2} (mk/2)!}=(mk-1)!!$  elements in the aforementioned basis of the Brauer algebra. Here we also note that  every basis element $\sigma\in\mathfrak{B}_{mk/2}(2n)$ can  be completely specified by $mk/2$ disjoint pairs, as 
\begin{equation}
    \sigma=\{\{\lambda_1, \sigma(\lambda_1)\}\cup\dots\cup\{\lambda_{mk/2}, \sigma(\lambda_{mk/2})\}\}\,.
\end{equation}

In Supp.~Fig.~\ref{fig:brauer_diagrams}, we diagrammatically show all the elements $\sigma\in\mathfrak{B}_{mk/2}(2n)$ for $\frac{mk}{2}=1,2$ (as well as some for $\frac{mk}{2}=3$) using tensor representation. The product is bilinear and associative (so the Brauer algebra is an associative algebra), and can be defined graphically on the basis elements by concatenation, see Supp.~Fig.~\ref{fig:brauer_diagrams}. . When closed loops appear, the resulting basis element gets multiplied by $2n$ raised to the number of loops. Furthermore, the transpose of an element $\sigma$ of the Brauer algebra is defined as $\sigma^\top=\{\{\lambda_{1}+mk/2, \sigma(\lambda_{1})+mk/2\}\cup\dots\cup\{\lambda_{mk/2}+mk/2 , \sigma(\lambda_{mk/2})+mk/2\}\}$, where the sum is taken mod $mk/2$
 (graphically, $\sigma^\top$ is the mirror image of $\sigma$).

\begin{figure}
    \centering
    \includegraphics[width=1\linewidth]{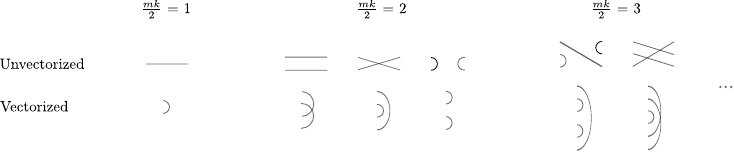}
    \caption{\textbf{Tensor representation of the basis elements of the Brauer algebra.} Tensor diagram showing the unvectorized and vectorized forms of all the elements $\sigma\in\mathfrak{B}_{mk/2}(2n)$ for $\frac{mk}{2}=1,2$, and some for $\frac{mk}{2}=3$.}
    \label{fig:brauer_diagrams}
\end{figure}

While the previous paragraph determines how the abstract Brauer algebra is defined, we still need to specify how its elements are represented and how they act on $\HC^{\otimes mk}$. In particular, we here consider the representation $F_{2n}:\mathfrak{B}_{mk/2}(2n)\rightarrow \BC(\HC^{\otimes mk/2})$ such that 
\begin{equation} \label{eq:rep-brauer_orig}
F_d(\sigma) = \sum_{i_1,\dots,i_{mk}=1}^{2n} \ket{i_{mk/2+1},i_{t+2},\dots,i_{mk}} \bra{i_1,i_2,\dots,i_{mk/2}} \,\prod_{\gamma=1}^{mk/2} \delta_{i_{\lambda_\gamma}, i_{\sigma(\lambda_\gamma)}}  \,, 
\end{equation} 
and its vectorization
\begin{equation}
    \ket{F_d(\sigma)}\!\rangle = \sum_{i_1,\dots,i_{mk}=1}^{2n} \ket{i_{mk/2+1},i_{t+2},\dots,i_{mk}} \ket{i_1,i_2,\dots,i_{mk/2}} \,\prod_{\gamma=1}^{mk/2} \delta_{i_{\lambda_\gamma}, i_{\sigma(\lambda_\gamma)}}   \,.  \label{eq:rep-brauer}
\end{equation}
Note that the transpose $\sigma^\top$ in the vectorization $\ket{F_d(\sigma^\top)}\!\rangle$ is graphically achieved by exchanging the top and bottom half indices/legs of the vectorization $\ket{F_d(\sigma)}\!\rangle$.

Given these definitions, the following important lemma holds.

            \begin{lemma}[Asymptotic moment operator]\label{supp:cor:asymp_moment_operator}
                The (vectorized) moment operator as defined in Eq.~\eqref{supp:eq:moment_operator_definition} is given by
                \begin{equation}\label{supp:eq:asymptotic_moment_operator_definition}
                    M_{\SO(2n)}^{(mk/2)} 
                    =
                    \frac{1}{(2n)^{mk/2}} \left( \sum_{\sigma \in \mathfrak{B}_{mk/2}} \ket{ F_{2n}(\sigma^\top)}\!\rangle\!\langle\!\bra{F_{2n}(\sigma) } 
                    + 
                    \!\!\!\sum_{\sigma, \pi \in \mathfrak{B}_{mk/2}}\!\!\! w_{\sigma, \pi} \ket{F_{2n}(\sigma)}\!\rangle\!\langle\!\bra{F_{2n}(\pi) } \right) \,,
                \end{equation}
                when $mk<2n$, with $w_{\sigma, \pi} \in \mathcal{O}\left(\frac{1}{n}\right)$. 
            \end{lemma}
            \begin{proof}
                The moment operator for the orthogonal group can be found in e.g., Ref.~\cite{garcia2023deep} (see Supplementary Theorem 6),                  and the proof that it is equal to that of the special orthogonal group when $mk<2n$ in Ref.~\cite{grood1999brauer}.
                Writing it in vectorized form readily yields Eq.~\eqref{supp:eq:asymptotic_moment_operator_definition}.
            \end{proof}

            In conclusion, the moment operator in Lemma~\ref{supp:cor:asymp_moment_operator} expresses the $k$-th moments (for $mk$ even) over $\mathbb{SO}(2n)$ in terms of Brauer algebra representations. Combining this result with Lemma~\ref{lemma:vanish_vs_moment_op}, we arrive at 
          \begin{align}\label{eq:supp:moment_via_Brauer_moment_op}
    \E\left[ X_{i_1} \cdots X_{i_k} \right] 
    = \frac{1}{(2n)^{mk/2}} \Biggl(&
    \sum_{\sigma \in \mathfrak{B}_{mk/2}} \langle\!\braket{ \tobs |^{\otimes k}| F_{2n}(\sigma^\top)}\!\rangle\!\langle\!\braket{F_{2n}(\sigma) | \bigotimes_{j=1}^{k} |\tstate^{(i_j)}}\!\rangle  \nonumber\\
   &+ 
    \!\sum_{\sigma, \pi \in \mathfrak{B}_{mk/2}} w_{\sigma, \pi} \langle\!\braket{ \tobs | ^{\otimes k}| F_{2n}(\sigma)}\!\rangle\!\langle\!\braket{F_{2n}(\pi) | \bigotimes_{j=1}^{k} |\tstate^{(i_j)}}\!\rangle  \Biggr)\,,
            \end{align}            
            as obtained from direct substitution of Eq.~\eqref{supp:eq:asymptotic_moment_operator_definition} in Eq.~\eqref{eq:supp:expectation_via_moment_op}.
            A tensor-network diagram of this  moment operator expression is provided in Supp.~Fig.~\ref{fig:tensor_network_Brauer_moment}, analogous to the previous $\SO(2n)$ Haar-random diagram in Supp.~Fig.~\ref{fig:tensor_network_diagram_vectorization}. 
            This diagrammatic approach will become useful to prove the main theorem in the subsequent section.

            \begin{figure}
                \centering
                \includegraphics[width=0.75\linewidth]{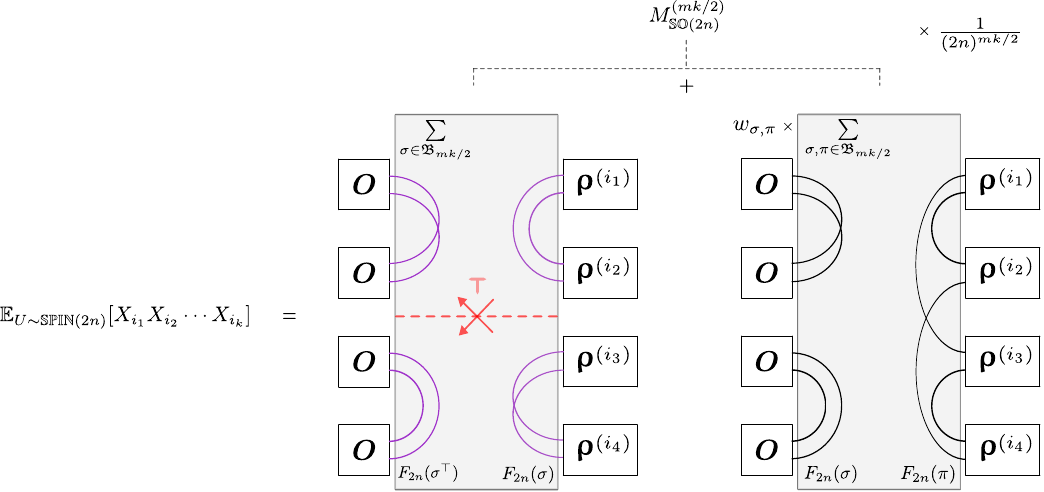}
                \caption{\textbf{Tensor-network diagram for asymptotic Weingarten-calculus moment expression.} Depiction in vectorized form with sums over basis elements of the Brauer algebra $\mathfrak{B}_{mk/2}$ representation, $F_{2n}(\sigma)$, as in Eq.~\eqref{eq:supp:moment_via_Brauer_moment_op}. While the first sum applies Brauer elements $F_{2n}(\sigma^T)$ and $F_{2n}(\sigma)$ to the observable and state tensors $\tobs$ and $\tstate^{(i_j)}$, related via transposition, independent Brauer elements $\sigma,\pi$ act on them in the second, asymptotically subordinate, sum. One specific pair of Brauer elements is represented for each sum.}
                \label{fig:tensor_network_Brauer_moment}
            \end{figure}

    \subsection{From moments to Gaussian processes}\label{app:theory:3}

        Next, we determine conditions on the observable $\O$ and states $\rho_{i_j}$ under which Haar-random matchgate transformations induce quantum Gaussian processes in the large-$n$ limit.
        This connection is established through Isserlis's theorem \cite{isserlis1918formula}, which provides the moments of a \ac{gp} (multivariate Gaussian distribution) of any order $k$. Consider a zero-mean multivariate Gaussian random vector $(X_i)_{i=1}^N \sim \NC(0, \Sigma)$. If $k$ is odd, the moment vanishes
        \begin{equation}
            \E\left[ X_{i_1} X_{i_2} \cdots X_{i_k} \right] = 0\,,   
        \end{equation}
        and if $k$ is even, the moments are uniquely determined by the covariances (second-order moments) of possible partitionings of the $k$ random variables into pairs. This reads as
        \begin{equation}\label{eq:isserlis_even_moments_form}
            \E\left[ X_{i_1} X_{i_2} \cdots X_{i_k} \right] 
            = \sum_{\sigma \in \mathfrak{B}_{k/2}} \prod_{(j, j') \in \sigma} \Cov\left[X_{i_j}, X_{i_{j'}}\right]\,,
        \end{equation}
        where $\Cov\left[X_{i}, X_{i'}\right] = \E\left[ X_{i} X_{i'} \right]$.

        To establish the conditions and proofs for the emergence of GPs, we analyze the moments in terms of Brauer algebra contractions of the observable and state tensors, $\tobs$ and $\tstate^{(i_j)}$, respectively. This analysis relies on the following contraction property:
        \begin{definition}[Pairwise tensor matchings]
            Consider a contraction over $k$ rank-$m$ tensors $\lbrace \bm{A}^{(j)}\rbrace_{j=1}^k$ with index pairings indicated by $\sigma \in \mathfrak{B}_{mk/2}$, denoted as 
            \begin{equation}
                \operatorname{contr}_{\sigma}\left(\bigotimes_{j=1}^k \bm{A}^{(j)} \right).
            \end{equation}
            We say that $\sigma$ provides a \emph{pairwise tensor matching} if all the indices of each tensor are contracted with the indices of exactly one other tensor. Formally, $\sigma$ can be decomposed into a tensor-level pairing, described by $\bm{\sigma} \in \mathfrak{B}_{k/2}$, and an index-level permutation $\pi_{j,j'} \in S_m$ for each of the $k/2$ tensor pairs $(j, j') \in \bm{\sigma}$, such that
            \begin{equation}
            \label{supp:eq:brauer_perfect_match_decomp}
                \sigma = \bm{\sigma} \circ \bigotimes_{(j, j')\in\bm{\sigma}} \pi_{j,j'}\,.
            \end{equation}
            Otherwise, $\sigma$ is referred to as a \emph{non-pairwise tensor matching}\footnote{Also denoted as $\neg$pairwise, $\neg$pair, or $\neg\,\sigma\,\mathrm{pairwise}$.}. 
        \end{definition}
        \noindent
        Supplemental~Fig.~\ref{fig:perfect_imperfect_brauer_decomp} illustrates tensor network diagrams for both pairwise (a) and non-pairwise (b) matchings, as well as the decomposition (c) into tensor-level pairings and index-level permutations as per Eq.~\eqref{supp:eq:brauer_perfect_match_decomp}.
        
        We now state two formal conditions for the formation of matchgate quantum Gaussian processes:
        \begin{condition}[Full module support of observable] \label{cond:full_support}
            The observable must have full support on the module $\BC_m$ ($m \neq 0, 2n$), such that
            \begin{equation} 
                \O \in \mathcal{B}_m
                \quad\iff\quad
                \left\lVert \O \right\rVert^2_{\mathcal{B}_m} 
                = 
                \left\lVert \O \right\rVert^2_{\mathcal{B}}\,.
            \end{equation}
        \end{condition}
        This first condition simplifies the analysis by constraining the observable to a single module. The next condition guarantees that in the asymptotic limit of large Hilbert-space dimension, the moments converge to those of a multivariate Gaussian distribution.
        
        \begin{condition}[Gaussian pairings dominate] \label{cond:gaussian_dominate}
             For any even $k\geq 4$ and all multisets $\lbrace \rho_{i_j}\rbrace_{j=1}^k$, let 
            \begin{alignat}{2}
                \Sperf &= \quad \sum_{\mathclap{\bm{\sigma} \in \mathfrak{B}_{k/2}}} && \left\lVert O \right\rVert_{\mathcal{B}_m}^k \prod_{(j,j') \in \bm{\sigma}} \braket{ \rho_{i_{j}} \,, \rho_{i_{j'}} }_{\BC_m}, \label{eq:Sperf} \\
                \Simp &= \quad \sum_{\mathclap{\substack{\sigma \in \mathfrak{B}_{mk/2}\\\neg\,\sigma\,\mathrm{pairwise}}}} \quad && \operatorname{contr}_{\sigma^\top}\!\left( \tobs^{\otimes k} \right) \operatorname{contr}_{\sigma}\!\left( \bigotimes_{j=1}^k \tstate^{\left(i_j\right)}\right)\,, \label{eq:Simp} \\
                \Smix &= \quad \sum_{\mathclap{\sigma, \pi \in \mathfrak{B}_{mk/2}}} && w_{\sigma, \pi}\operatorname{contr}_{\sigma}\!\left( \tobs^{\otimes k} \right) \operatorname{contr}_{\pi}\!\left( \bigotimes_{j=1}^k \tstate^{\left(i_j\right)}\right) 
                \,, \label{eq:Smix}
            \end{alignat}
            with $w_{\sigma, \pi} \in \OC\left(\frac{1}{n}\right)$ as in Eq.~\eqref{eq:supp:moment_via_Brauer_moment_op}.

            \noindent
            In the limit $n\to\infty$, the sum $\Sperf$ must then asymptotically (strictly) dominate\footnote{Recall that the asymptotic bound applies to the magnitude; i.e., $f(n) \in o(g(n)) \iff |f(n)| \in o(g(n))$.} the sums $\Simp$ and $\Smix$, i.e.,
            \begin{equation}
                \frac{\Simp}{\Sperf} \in o(1) \,,
                \qquad 
                \text{and}
                \qquad
                \frac{\Smix}{\Sperf} \in o(1) \,
                .
            \end{equation}
            
        \end{condition}
        \noindent 
        This is the formal statement of the informal ``Gaussian pairings dominate'' condition in the main text. 
        Note that Supp.~Info.~\ref{app:theory:4} provides simpler alternative, both equivalent and sufficient, formulations to Condition~\ref{cond:gaussian_dominate}.

            \begin{figure}[tbp]
                \centering
                \includegraphics[width=0.6\linewidth]{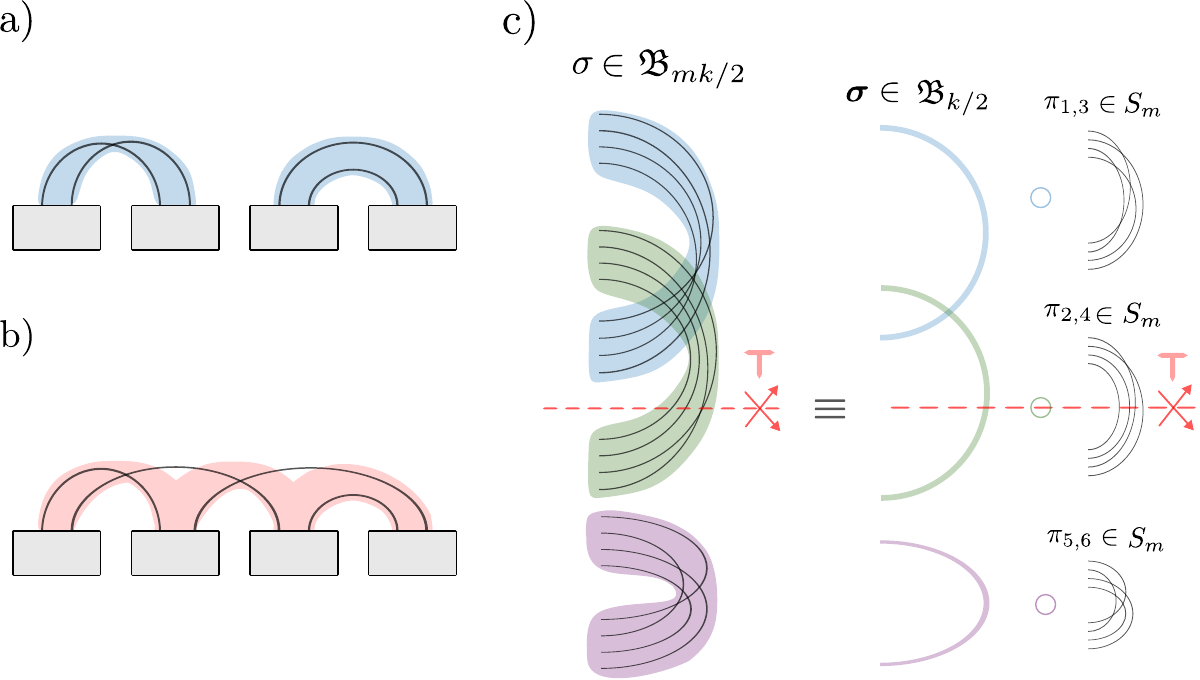}
                \caption{
                \textbf{Pairwise and non-pairwise tensor matchings and pairwise matching decomposition.} Lines indicate connected tensor legs at the index level, whereas shaded backgrounds represent tensor-level matchings. (a, b) Index pairing examples ($m=2$, $k=4$) yielding (a) a pairwise tensor matching between the first and last tensor pairs, and (b) a non-pairwise matching that connects all four tensors without pairwise factorization.
                (c) Decomposition of a Brauer element $\sigma \in \mathfrak{B}_{mk/2}$ representing a pairwise tensor matching ($m=4$, $k=6$), as defined in Eq.~\eqref{supp:eq:brauer_perfect_match_decomp}. The horizontal dashed red line and arrows indicate the action of the transposition, which swaps the top and bottom halves of the indices. Transposing $\sigma \in \mathfrak{B}_{mk/2}$ inverts only those index-level permutations $\pi_{j,j'} \in S_m$ with corresponding tensor pairings $j,j' \in \mathfrak{B}_{k/2}$ that cross this line.}
                \label{fig:perfect_imperfect_brauer_decomp}
            \end{figure}

        Under these conditions, we derive explicit expressions for the moments, distinguishing between even and odd moments.
        \begin{lemma}[Even moments]\label{thm:supp:general_asym_even_moments}
            Under Conditions~\ref{cond:full_support} and~\ref{cond:gaussian_dominate}, and in the large-$n$ limit, the even $k$-th moments are
            \begin{equation}\label{eq:general_asymptotic_k_moment}
                \E_{U\sim \SPIN(2n)}\left[ X_{i_1} X_{i_2} \cdots X_{i_k} \right] 
                = 
                \sum_{\bm{\sigma} \in \mathfrak{B}_{k/2}} \prod_{(j, j') \in \bm{\sigma}} \underbrace{\frac{m!}{(2n)^m} \left\lVert \O \right\rVert^2_{m} \left\langle \rho_{i_j} , \rho_{i_{j'}} \right \rangle_{m}}_{= \Cov(X_{i_j}, X_{i_{j'}})}\,.
            \end{equation}
        \end{lemma}
        \begin{proof}
            Given Condition~\ref{cond:full_support}, the observable $O$ is fully supported on the module $\BC_m$ and, hence, allows us to apply Lemma~\ref{lemma:vanish_vs_moment_op} to the even $k$-th moments $\E_{U\sim \SPIN(2n)}\left[ X_{i_1} X_{i_2} \cdots X_{i_k} \right]$. Specifically, since $k$ is even, $mk$ is even, such that Lemma~\ref{lemma:vanish_vs_moment_op} applies to express the moment through the moment operator,
            \begin{equation}
                \E\left[ X_{i_1} \cdots X_{i_k} \right] 
                = \langle\!\bra{ \tobs }^{\otimes k} \; M_{\SO(2n)}^{(mk/2)} \; \bigotimes_{j=1}^{k} \ket{\tstate^{(i_j)}}\!\rangle
                .
            \end{equation}
            Recall that $\ket{\tobs}\!\rangle^{\otimes k}$ and $\bigotimes_{j=1}^{k} \ket{\tstate^{(i_j)}}\!\rangle$ are the vectorized module ($\mathcal{B}_m$) coefficient tensors for the observable $O$ and the states $\rho_{i_j}$, respectively.
            As we are interested in the large-$n$ limit, we further replace the moment operator $M_{\SO(2n)}^{(mk/2)}$ by its asymptotic expansion given by Lemma~\ref{supp:cor:asymp_moment_operator}, as
            \begin{align} 
                \E\left[ X_{i_1} \cdots X_{i_k} \right] 
                &= \frac{1}{(2n)^{mk/2}} \Biggl( \phantom{{}+{}}
                \sum_{\substack{\sigma \in \mathfrak{B}_{mk/2}\\\sigma \, \mathrm{pairwise}}} \langle\!\braket{ \tobs |^{\otimes k}| F_{2n}(\sigma^\top)}\!\rangle\!\langle\!\braket{F_{2n}(\sigma) | \bigotimes_{j=1}^{k} |\tstate^{(i_j)}}\!\rangle \label{eq:supp:asymptotic_moment:Sperf} \\
                &\phantom{{}= \frac{1}{(2n)^{mk/2}} \Biggl(} +
                \sum_{\substack{\sigma \in \mathfrak{B}_{mk/2}\\\neg\,\sigma \, \mathrm{pairwise}}} \langle\!\braket{ \tobs |^{\otimes k}| F_{2n}(\sigma^\top)}\!\rangle\!\langle\!\braket{F_{2n}(\sigma) | \bigotimes_{j=1}^{k} |\tstate^{(i_j)}}\!\rangle \label{eq:supp:asymptotic_moment:Simp}
                \\
                &\phantom{{}= \frac{1}{(2n)^{mk/2}} \Biggl(} + 
                \!\sum_{\sigma, \pi \in \mathfrak{B}_{mk/2}} w_{\sigma, \pi} \langle\!\braket{ \tobs | ^{\otimes k}| F_{2n}(\sigma)}\!\rangle\!\langle\!\braket{F_{2n}(\pi) | \bigotimes_{j=1}^{k} |\tstate^{(i_j)}}\!\rangle \Biggr)\,,
                \label{eq:supp:asymptotic_moment:Smix}
            \end{align}
            where we additionally split the single summation over the Brauer algebra $\mathfrak{B}_{mk/2}$ into a sum with only pairwise tensor matchings on both $\kket{\tobs}^{\otimes k}$ and $\bigotimes_{j=1}^k \kket{\tstate^{(i_j)}}$ (Eq.~\eqref{eq:supp:asymptotic_moment:Sperf}), and another sum for the remaining matchings (Eq.~\eqref{eq:supp:asymptotic_moment:Simp}), which are non-pairwise.
            
            The main part of the proof focuses on showing that such split is possible and that the sum in Eq.~\eqref{eq:supp:asymptotic_moment:Sperf} matches, up to a factor $m!$, $\Sperf$ in Condition~\ref{cond:gaussian_dominate}, Eq.~\eqref{eq:Sperf}.             Simultaneously, this sum is shown to pick up the correct pre-factor to precisely match the even $k$-th moment in Eq.~\eqref{eq:general_asymptotic_k_moment}, which is the objective of this proof. To then conclude the proof, we apply Condition~\ref{cond:gaussian_dominate} to asymptotically strictly bound the other two sums in Eqs.~\eqref{eq:supp:asymptotic_moment:Simp} and~\eqref{eq:supp:asymptotic_moment:Smix} by $\Sperf$. Note that the dimension of the Brauer algebra, which determines the number of terms involved in these sums, does not depend on $n$ and, hence, cannot impact the asymptotic order.

            We note that for any Brauer algebra element $\sigma \in \mathfrak{B}_{mk/2}$ that induces a pairwise tensor matching in $\bigotimes_{j=1}^{k} \ket{\tstate^{(i_j)}}\!\rangle$ its transpose $\sigma^\top$ simultaneously induces a pairwise tensor matching in $\ket{\tobs}\!\rangle^{\otimes k}$. Otherwise, $\sigma$ and $\sigma^\top$ simultaneously lead to non-pairwise tensor matchings.             To show this property, consider the decomposition of a pairwise tensor matching $\sigma \in \mathfrak{B}_{mk/2}$ into a tensor-level pairing $\bm{\sigma} \in \mathfrak{B}_{k/2}$ and index-level permutations $\pi_{j,j'} \in S_m$ for each tensor pair $(j, j') \in \bm{\sigma}$ in Eq.~\eqref{supp:eq:brauer_perfect_match_decomp}. 
            The transposition $\sigma^\top$ affects the decomposition as follows
            \begin{equation}
                \sigma^\top = \bm{\sigma}^\top 
                \circ 
                \bigotimes_{(j, j')\in\bm{\sigma}^\top}\bar{\pi}_{j,j'}\,,
            \end{equation}
            with
            \begin{align}
                \bm{\sigma}^\top &= \left\lbrace (j + k/2, j' + k/2)\!\!\!\! \mod{k} \; \middle| \; (j,j') \in \bm{\sigma} \right\rbrace\,
                , \\ 
                \bar{\pi}_{j,j'} &= \begin{cases}
                    \pi_{j,j'} \quad\text{if} \quad j,j' \leq k/2 \; \lor j,j' > k/2 \; \\
                    \pi_{j,j'}^{-1} \quad \text{else}
                \end{cases} \,
                , \label{supp:eq:index_level_permutation_inversion}
            \end{align}
            maintaining the structure of a pairwise matching of tensors because a decomposition is still possible. Supplemental Fig.~\ref{fig:perfect_imperfect_brauer_decomp} further exemplifies the effect of the transposition within the decomposition of the pairwise tensor matching.  

            Furthermore, multiple Brauer elements $\sigma \in \mathfrak{B}_{mk/2}$ can resemble the same pairwise tensor matching $\bm{\sigma} \in \mathfrak{B}_{k/2}$, differing only in their index-level permutations $\pi_{j,j'} \in S_m$. Remarkably, the contribution of any such $\sigma$ in Eq.~\eqref{eq:supp:asymptotic_moment:Sperf} is invariant under the choice of $\pi_{j,j'}$. 
            To see why, recall that the contracted tensors are anti-symmetric, meaning an index permutation $\pi_{j,j'}$ introduces a (parity) sign, $\mathrm{sign}(\pi_{j,j'})$. However, the overall contribution depends on both $\sigma$ and its transpose $\sigma^\top$. As established in Eq.~\eqref{supp:eq:index_level_permutation_inversion}, the transpose acts on the index level by applying $\bar{\pi}_{j,j'}$, which is either $\pi_{j,j'}$ or its inverse $\pi_{j,j'}^{-1}$. Because the sign of a permutation is invariant under inversion ($\mathrm{sign}(\pi) = \mathrm{sign}(\pi^{-1})$), the combined sign from both contractions always trivially cancels as $\mathrm{sign}(\pi_{j,j'}) \mathrm{sign}(\bar{\pi}_{j,j'}) = 1$. 
            
            Consequently, the permutation signs vanish entirely, and every choice of index-level permutations collapses identically to the non-permuted (id) contraction:
            \begin{align}
                \operatorname{contr}_{\pi_{j,j'}}(\tobs, \tobs)
                \operatorname{contr}_{\bar{\pi}_{j,j'}}(\tstate^{({i_j})}, \tstate^{({i_{j'}})}) 
                &= 
                \operatorname{contr}_{\mathrm{id}}(\tobs, \tobs)
                \operatorname{contr}_{\mathrm{id}}(\tstate^{({i_j})}, \tstate^{({i_{j'}})}) \nonumber \\
                &= 
                \left(\sum_{\bm{\nu}} \lvert \tobselem_{\bm{\nu}} \rvert^2\right) 
                \left( \sum_{\bm{\nu}} (\tstateelem_{\bm{\nu}}^{(i_j)})^* \tstateelem_{\bm{\nu}}^{(i_{j'})} \right) \nonumber \\
                &= 
                \left\lVert \O \right\rVert^2_{m} \langle \rho_{i_j} , \rho_{i_{j'}}  \rangle_{m} \,\label{eq:supp:single_perfect_matching_contribution}
                .
            \end{align}
            
            This scalar contribution occurs independently for each of the $k/2$ pairs $(j,j') \in \bm{\sigma}$. Since the symmetric group $S_m$ governs the index pairings, there are exactly $m!$ valid choices for each $\pi_{j,j'}$. Summing over all such permutation choices for a fixed tensor pairing $\bm{\sigma}$ naturally scales the contribution by $(m!)^{k/2}$,
            \begin{equation}
                (m!)^{k/2} \prod_{(j,j') \in \bm{\sigma}} \left\lVert \O \right\rVert^2_{m} \langle \rho_{i_j} , \rho_{i_{j'}}  \rangle_{m}\,
                .
            \end{equation}
            Including the global pre-factor $1/(2n)^{mk/2}$ from Eq.~\eqref{eq:supp:asymptotic_moment:Sperf}, the total sum over all pairwise tensor matching contributions perfectly matches $\Sperf$ (up to the scaling factor) and precisely recovers the even $k$-th moment given in Eq.~\eqref{eq:general_asymptotic_k_moment} as
            \begin{equation}
                \frac{1}{(2n)^{mk/2}} \sum_{\bm{\sigma} \in \mathfrak{B}_{k/2}} (m!)^{k/2} \prod_{(j,j') \in \bm{\sigma}} \left\lVert \O \right\rVert^2_{m} \langle \rho_{i_j} , \rho_{i_{j'}}  \rangle_{m}
                = 
                \sum_{\bm{\sigma} \in \mathfrak{B}_{k/2}}  \prod_{(j,j') \in \bm{\sigma}} \frac{m!}{(2n)^{m}} \left\lVert \O \right\rVert^2_{m} \langle \rho_{i_j} , \rho_{i_{j'}}  \rangle_{m}\,
                .
            \end{equation}

            To conclude the proof, observe that by the definition of the Brauer algebra representation $F_{2n}$ in Eq.~\eqref{eq:rep-brauer_orig} and its vectorization $\kket{F_{2n}}$ in Eq.~\eqref{eq:rep-brauer}, the vectorized inner product evaluates to $\bbrakket{\bm{A}}{F_{2n}(\sigma)} = \operatorname{contr}_{\sigma}(\bm{A})$ for any rank-$mk$ $2n$-dimensional tensor (network) $\bm{A}$ and $\sigma \in \mathfrak{B}_{mk/2}$. This equivalence is also visually apparent in the tensor network diagram of Suppl.~Fig.~\ref{fig:tensor_network_Brauer_moment}. 
            Therefore, the remaining sums, in Eqs.~\eqref{eq:supp:asymptotic_moment:Simp} and~\eqref{eq:supp:asymptotic_moment:Smix}, constituting the $k$-th even moment map directly to $\Simp$ and $\Smix$, respectively. By Condition~\ref{cond:gaussian_dominate}, these terms are asymptotically strictly dominated by $\Sperf$. Thus, in the large-$n$ limit, the contributions from $\Simp$ and $\Smix$ vanish faster, leaving only the $\Sperf$ term of leading order and allowing us to recover the general asymptotic form of the even $k$-th moments in Eq.~\eqref{eq:general_asymptotic_k_moment}.

        \end{proof}

An important corollary can be obtained from Lemma~\ref{thm:supp:general_asym_even_moments}, regarding the covariance matrix of the variables $\{X_i\}_{i=1}^n$.
        
        \begin{corollary}[Covariances]
            The covariance of $X_{i}$ and $X_{i'}$ is
            \begin{equation}
                \Sigma_{i, i'} = \Cov\left(X_{i}, X_{i'}  \right) = \frac{m!}{(2n)^m} \lVert O \rVert^2 \left\langle \rho_{i} , \rho_{{i'}} \right \rangle_{m}\,.
            \end{equation}
        \end{corollary}
        \begin{proof}
            Direct implication of Lemma~\ref{thm:supp:general_asym_even_moments} with $k=2$, leading to the trivial Brauer algebra $\mathfrak{B}_{1}$.
        \end{proof}

    Next, we present a proposition that automatically guarantees zero odd moments for some values of $m$ and $k$.

        \begin{proposition}[Guarantees for vanishing odd moments]\label{prop:supp:vanish_odd_moments}
            Under Condition~\ref{cond:full_support}, odd $k$ moments vanish ($ \E_{U\sim \SPIN(2n)}\left[ X_{i_1} \cdots X_{i_k} \right] = 0 $) in any of the following cases:
            \begin{enumerate}[label=\emph{(\alph*)}]
                \item When  $k = 1$ (first moment).
                \item When $m$ is odd.
                \item When $k=3$  (third moment) and $m \equiv 2 \pmod{4}$   (i.e., $m$ is even but not a multiple of four).
            \end{enumerate}
        \end{proposition}
        \begin{proof}
            {(a)} The first moment ($k=1$)   trivially vanishes as directly implied by Lemma~\ref{lemma:vanish_vs_moment_op} for odd $m$, and for even $m$, because all contractions in Eq.~\eqref{eq:supp:moment_via_Brauer_moment_op} are over individual anti-symmetric tensors, which are zero by definition. 
            
            {(b)} Vanishing odd moments generally in the odd $m$ case  are directly implied by Lemma~\ref{lemma:vanish_vs_moment_op}.

            {(c)} Focusing on the third moment ($k=3$), any contraction in Eq.~\eqref{eq:supp:moment_via_Brauer_moment_op} defined by $\sigma \in \mathfrak{B}_{{3m}/{2}}$ either introduces self-loops (which vanish due to anti-symmetry) or splits the indices of each tensor into two groups of $m/2$ indices, according to which other tensor they are connected to. By then reshaping the tensors into $\tfrac{m}{2} \times \tfrac{m}{2}$ square matrices, we can express the contractions as matrix traces (up to a sign). Crucially, we only need to examine the contraction over the observable tensors, which corresponds to the trace $\lvert\tr\left[\tobsmat^3 \right]\rvert$. Matrix transposition here corresponds to $m/2$ swaps of tensor indices. Because $m \equiv 2 \pmod{4}$, the number of swaps $m/2$ is odd, meaning the reshaped matrix $\tobsmat$ is anti-symmetric. An anti-symmetric matrix raised to an odd power remains anti-symmetric and is therefore traceless, i.e., $\tr\left[\tobsmat^3\right] = 0$. Since every term in Eq.~\eqref{eq:supp:moment_via_Brauer_moment_op} contains an observable contraction factor, the entire third moment exactly vanishes.
        \end{proof}
        \noindent
        Note that all odd moments are also algebraically guaranteed to vanish for a special class of observables within any module $\BC_m$, as those constructed as linear combinations of $\BC_m$ basis operators associated with mutually disjoint Majorana indices (see Proposition~\ref{thm:special_obs_odd_moments} in Supp.~Info.~\ref{sec:special_obs}).
        
        Proposition~\ref{prop:supp:vanish_odd_moments} outlines where odd moments inherently vanish due to the algebraic structure of their contractions, but these guarantees are incomplete for even (fermionic) modules $\BC_m$. For the third moment ($k=3$), reshaping the tensors into matrices allows the contractions to be evaluated as traces. While this reshaping yields anti-symmetric matrices with naturally vanishing traces when $m \equiv 2 \pmod{4}$, modules with $m \equiv 0 \pmod{4}$ involve an even  number of index swaps ($m/2$). This results in symmetric matrices whose traces do not inherently vanish, meaning that a zero third moment relies entirely on specific system properties (e.g., $\tr\left[\tobsmat^3 \right] = 0$). Furthermore, this matrix reshaping technique is strictly limited to $k=3$, while for higher odd moments ($k \geq 5$), a single tensor can contract with more than two others, breaking the cyclic structure required to express the contraction as a matrix trace and preventing straightforward algebraic arguments, thus necessitating a case-by-case analysis. 
        Comparing the explicit expressions for the moments in Lemma~\ref{thm:supp:general_asym_even_moments} to the moments of a Gaussian process (via Isserlis theorem \cite{isserlis1918formula}), we derive the following two main theorems that certify whether Gaussian processes are formed or not, respectively (with some cases covered by Prop.~\ref{prop:supp:vanish_odd_moments}):
        \begin{theorem}[Gaussian Process]\label{thm:supp:GP}
            Under Conditions~\ref{cond:full_support}, \ref{cond:gaussian_dominate} and in the large-$n$ limit, $\left\lbrace X_{i} \right\rbrace_{i=1}^N$ forms a \ac{gp} if and only if all odd moments vanish.
        \end{theorem}
        \begin{proof}
            Clearly, since Lemma~\ref{thm:supp:general_asym_even_moments} holds under Conditions \ref{cond:full_support} and~\ref{cond:gaussian_dominate} and in the large-$n$ limit, the form of the even moments as specified in Eq.~\eqref{eq:general_asymptotic_k_moment} then matches the form in Eq.~\eqref{eq:isserlis_even_moments_form} of Isserlis's theorem \cite{isserlis1918formula}. Furthermore, the zero odd moments align with Isserlis's theorem. 
            Finally, to prove that these matching moments uniquely determine the joint distribution of the $X_i$ random variables, we can use Carleman's condition.

            \begin{lemma}
                [Carleman's condition, Hamburger case~\cite{kleiber2013multivariate}]\label{lem:carleman}
Let $\gamma_k$ be the (finite) moments of the distribution of a random variable $X$ that can take values on the real line $\mathbb{R}$. These moments determine uniquely the distribution of $X$ if
\begin{equation}
    \sum_{k=1}^\infty \gamma_{2k}^{-1/2k} = \infty \, .
\end{equation}
\end{lemma}

In our case, let
\begin{equation}
X=\operatorname{Tr}\left[U\rho U^\dagger O\right]\,,
\end{equation}
and assume that $O\in B_m$ (Condition~\ref{cond:full_support}). Under the Gaussian-pairing-dominance condition (Condition~\ref{cond:gaussian_dominate}), the limiting moments of $X$ are
\begin{equation}
\gamma_{2k-1}=0 \,, \qquad  \gamma_{2k} = (2k-1)!!\,\sigma^{2k}\,,
\end{equation}
with
\begin{equation}
\sigma^2 = \frac{m!}{(2n)^m}\,\|O\|_m^2\,\|\rho\|_m^2 \,.
\end{equation}
Then, for any $n$,
\begin{equation}
\begin{aligned}
\sum_{k=1}^{\infty}\gamma_{2k}^{-1/(2k)} &= \frac{1}{\sigma} \sum_{k=1}^{\infty}\left((2k-1)!!\right)^{-1/(2k)}
\\
&\geq \frac{1}{\sigma} \sum_{k=1}^{\infty}(2k)^{-1/2} = \infty\,,
\end{aligned}
\end{equation}
where we used $(2k-1)!!\leq (2k)^k$. Hence Carleman's condition is satisfied, and the limiting univariate distribution is uniquely determined by its moments. Therefore, $X$ converges in distribution to $\NC(0,\sigma^2)$. To conclude,  we use the fact that since its marginal distributions are determinate via Carleman's condition, then so is the joint distribution of $\{X_i\}_{i=1}^N$~\cite{kleiber2013multivariate}.

\end{proof}

        \begin{theorem}[No Gaussian Process]\label{thm:supp:no_GP}
            Under Condition~\ref{cond:full_support} and in the large-$n$ limit, $\left\lbrace X_{i} \right\rbrace_{i=1}^N$ does not form a \ac{gp} if \emph{a)} Condition~\ref{cond:gaussian_dominate} is violated, or \emph{b)}
            non-zero odd moments occur.
        \end{theorem}
        \begin{proof} 
            By the contrapositive of Isserlis's theorem \cite{isserlis1918formula}, a distribution is not a zero-mean Gaussian if its even moments deviate from the pairwise sum in Eq.~\eqref{eq:isserlis_even_moments_form}, or if its odd moments are non-zero. The theorem follows immediately:
            \begin{itemize}[nosep]
                \item[\emph{a)}] Violating Condition~\ref{cond:gaussian_dominate} implies that non-pairwise tensor matchings contribute at the leading asymptotic order. Consequently, the even moments diverge from the pairwise form of Eq.~\eqref{eq:general_asymptotic_k_moment}, violating Isserlis's theorem.
                \item[\emph{b)}] The presence of non-zero odd moments trivially violates the definition of a zero-mean Gaussian process.
            \end{itemize}
        \end{proof}

        \subsection{General implications and simplifications}\label{app:theory:4}

            We now provide simplified, equivalent formulations of Condition~\ref{cond:gaussian_dominate} for the univariate case, and establish an alternative sufficient condition. Furthermore, we derive broad implications by restricting our focus to the $\BC_1$ module and introduce verifiable equivalent conditions for the $\BC_2$ module. Note that further simplifications may be achieved by combining insights from the alternative formulations presented here.

            \subsubsection{Equivalent univariate condition}

            \begin{proposition}[Equivalent univariate formulation of Condition~\ref{cond:gaussian_dominate}] \label{prop:cond:univariate}
            In the univariate case, i.e., single state $\rho$ being considered, Condition~\ref{cond:gaussian_dominate} is equivalent to the following condition:
            \medskip

            \noindent
            For every even $k \geq 4$, 
            \begin{equation}\label{eq:cond:univariate:imp}
                \forall \sigma \in \mathfrak{B}_{mk/2},\; \neg \, \sigma \, \mathrm{pairwise} \colon\, 
                \frac{\operatorname{contr}_{\sigma^\top}\!\left( \tobs^{\otimes k} \right)}{\left\lVert O \right\rVert_{\mathcal{B}_m}^k } \frac{\operatorname{contr}_{\sigma}\!\left( \tstate^{\otimes k} \right)}{\left\lVert \rho \right\rVert_{\mathcal{B}_m}^k } \in o(1)\,,
            \end{equation}
            and 
            \begin{equation}\label{eq:cond:univariate:mix}
                \forall \sigma,\pi \in \mathfrak{B}_{mk/2}\colon\, 
                \frac{\operatorname{contr}_{\sigma}\!\left( \tobs^{\otimes k} \right)}{\left\lVert O \right\rVert_{\mathcal{B}_m}^k } \frac{\operatorname{contr}_{\pi}\!\left( \tstate^{\otimes k} \right)}{\left\lVert \rho \right\rVert_{\mathcal{B}_m}^k } \in o\left(\frac{1}{w_{\sigma, \pi}}\right)\,,
            \end{equation}
            where $1/w_{\sigma, \pi} \in \Omega(n)$ as defined in Eq.~\eqref{supp:eq:asymptotic_moment_operator_definition}.
        \end{proposition}
        \begin{proof}
            Analogous to the observable, in the univariate case, the module overlaps for the states reduce to module purities as 
            \begin{equation}
                \Sperf = \dim(\mathfrak{B}_{k/2}) \left(\left\lVert O \right\rVert_{m}^2 \left\lVert \rho \right\rVert_{_m}^2\right)^{k/2}\,
                .
            \end{equation}
            It is important to note that $\Sperf\geq 0$ so that, instead of bounding the ratios of the full sums $\Simp$ and $\Smix$, we can equivalently bound all the summands (recall that the sums iterate over the Brauer algebra basis, with dimension always independent of the number of qubits $n$). Therefore, we obtain that ${\Simp}/{\Sperf}\in o(1)$ is equivalent to Eq.~\eqref{eq:cond:univariate:imp} and ${\Smix}/{\Sperf}\in o(1)$ is equivalent to Eq.~\eqref{eq:cond:univariate:mix}, where the latter condition was additionally rearranged by dividing by $w_{\sigma, \pi}$ factors.
        \end{proof}

        \subsubsection{Sufficient conditions}

            While the non-negativity of the module purity/norm allowed for an equivalent simplified formulation of Condition~\ref{cond:gaussian_dominate} in the univariate case, for the multivariate case, such simplification can only be achieved by requiring non-negative module overlaps, resulting in the following strictly stronger sufficient condition:
            \begin{proposition}[Sufficient non-negative module overlap formulation of Condition~\ref{cond:gaussian_dominate}] \label{prop:cond:sufficient_mutivariate}
                Condition~\ref{cond:gaussian_dominate} is fulfilled if the following sufficient condition holds:
                \medskip

                \noindent 
                All states exhibit non-negative module overlaps
                \begin{equation}
                    \forall \rho_i, \rho_{i'}\colon
                    \braket{\rho_i, \rho_{i'}}_m \geq 0\,,
                \end{equation}
                and, for every even $k \geq 4$,
                \begin{equation}\label{eq:cond:pos_multi:imp}
                    \forall \sigma \in \mathfrak{B}_{mk/2}, \;\neg \, \sigma \, \mathrm{pairwise} \;\exists \bm{\sigma}\in \mathfrak{B}_{k/2} \colon \,
                    \frac{\operatorname{contr}_{\sigma^\top}\!\bigl( \tobs^{\otimes k} \bigr)}{\left\lVert O \right\rVert_{m}^k }
                    \frac{ \operatorname{contr}_{\sigma}\!\bigl( \bigotimes_{j=1}^k \tstate^{\left(i_j\right)} \bigr)}{\prod_{(j,j') \in \bm{\sigma}} \langle{ \rho_{i_{j}} , \rho_{i_{j'}} }\rangle_{m} } 
                    \in o(1)\,,
                \end{equation}
                and 
                \begin{equation}\label{eq:cond:pos_multi:mix}
                    \forall \sigma,\pi \in \mathfrak{B}_{mk/2} \; \exists \bm{\sigma}\in \mathfrak{B}_{k/2}
                    \colon \,
                    \frac{\operatorname{contr}_{\sigma}\!\bigl( \tobs^{\otimes k} \bigr)}{\left\lVert O \right\rVert_{m}^k }
                    \frac{ \operatorname{contr}_{\pi}\!\bigl( \bigotimes_{j=1}^k \tstate^{\left(i_j\right)} \bigr)}{\prod_{(j,j') \in \bm{\sigma}} \langle{ \rho_{i_{j}} , \rho_{i_{j'}} }\rangle_{m} }
                    \in o\left(\frac{1}{w_{\sigma, \pi}}\right)\,,
                \end{equation}
                where $1/w_{\sigma, \pi} \in \Omega(n)$ as defined in Eq.~\eqref{supp:eq:asymptotic_moment_operator_definition}.
                
            \end{proposition}
            \begin{proof}
                Analogous to the proof of Proposition~\ref{prop:cond:univariate}, the required non-negativity of the module overlaps ($\braket{\rho_i, \rho_{i'}}_m \geq 0$) ensures that every summand in $\Sperf$  (Eq.~\eqref{eq:Sperf}) is strictly non-negative. Consequently, the full sum $\Sperf$ is lower-bounded by any single one of its summands, corresponding to a specific pairwise matching $\bm{\sigma} \in \mathfrak{B}_{k/2}$. 
    
                Therefore, to show that $\Simp/\Sperf \in o(1)$ and $\Smix/\Sperf \in o(1)$, it suffices to show that \emph{every} individual term in $\Simp$ and $\Smix$ is asymptotically dominated by \emph{at least one} corresponding term from $\Sperf$. Because the number of elements in the Brauer algebra (and thus the number of summands) is independent of $n$, this term-by-term bounding guarantees the asymptotic dominance of the entire sum. Equations \eqref{eq:cond:pos_multi:imp} and \eqref{eq:cond:pos_multi:mix} precisely formalize this term-by-term domination by asserting the existence of such a tensor bounding $\bm{\sigma}$ for every $\sigma$ and $\pi$.
            \end{proof}

            Alternatively, by relaxing the explicit weights $w_{\sigma, \pi}$ to their asymptotic upper bound of $\mathcal{O}(1)$, we can merge $\Simp$ and $\Smix$ into a single, strictly stronger sufficient condition.

            \begin{proposition}[Sufficient unified single-sum formulation of Condition~\ref{cond:gaussian_dominate}] \label{prop:cond:unified_single_sum}
                Condition~\ref{cond:gaussian_dominate} is satisfied if $\Sperf$ asymptotically dominates the unweighted sum of all contractions involving at least one non-pairwise matching:
                \begin{equation}
                    \frac{\sum_{{\substack{\sigma, \pi \in \mathfrak{B}_{mk/2} \\ \neg(\sigma \,\mathrm{pairwise} \,\wedge \,\pi \,\mathrm{ pairwise})}}} 
                    \operatorname{contr}_{\pi}\!\bigl( \tobs^{\otimes k} \bigr) \operatorname{contr}_{\sigma}\!\bigl( \bigotimes_{j=1}^k \tstate^{\left(i_j\right)}\bigr)}{\Sperf} 
                    \in o(1)\,.
                \end{equation}
            \end{proposition}
            \begin{proof}
                $\Simp$ consists of non-pairwise terms ($\pi^\top = \sigma$, weight 1), while $\Smix$ contains cross-terms with weight $w_{\sigma, \pi} \in \mathcal{O}(1/n)$. Relaxing these weights to $\mathcal{O}(1)$ and bounding their combined sum strictly upper-bounds both $\Simp$ and $\Smix$. We can safely exclude cross-terms where both $\sigma$ and $\pi$ are pairwise since their raw contractions scale as $\mathcal{O}(\Sperf)$, so their $\mathcal{O}(1/n)$ weights inherently guarantee they vanish as $o(1)$ relative to $\Sperf$.
            \end{proof}

        \subsubsection{Module \texorpdfstring{$\BC_1$}{B1} general implications}\label{sec:general_implications:B1}

                For the smallest non-trivial module $\BC_1$ (i.e., $m=1$), Condition~\ref{cond:gaussian_dominate} always holds and we generally obtain GPs,  both in the univariate ($N=1$) and multivariate ($N>1$) cases:
                \begin{corollary}\label{cor:general_single_majo_GPs}
                    For observables $O\in\BC_1$,
                     under Condition~\ref{cond:full_support} and  in the large-$n$ limit, $\left\lbrace X_{i_j} \right\rbrace_{j=1}^N$ forms a \ac{gp}.                 \end{corollary}
                \begin{proof}
                    First, we note that Condition~\ref{cond:gaussian_dominate} (which must be fulfilled for Lemma~\ref{thm:supp:general_asym_even_moments} to apply) is trivially fulfilled since any contraction from $\mathfrak{B}_{mk/2}$ for $m = 1$ over $k$ tensors results only in pairwise tensor matchings as per the definition of the Brauer algebra $\mathfrak{B}_{mk/2}$.  This is valid for both distinct tensors and copies of the same tensor, i.e., multivariate and univariate scenarios, respectively. 
                    Specifically, when viewing these rank-1 tensors $\tobs$ and $\tstate^{(i_j)}$ as vectors,                     it is clear that any contraction leads to scalar products of pairs of such tensors.
                    Second, the odd number of Majoranas ($m=1$) makes the odd moments trivially vanish as per Prop.~\ref{prop:supp:vanish_odd_moments}. Therefore, Theorem~\ref{thm:supp:GP} generally applies and predicts \acp{gp}.
                \end{proof}

                Indeed, we  show that the asymptotic kernel function (covariance) matches the exact non-asymptotic expression for $\BC_1$. Instead of using the asymptotic moment operator in Eq.~\eqref{supp:eq:asymptotic_moment_operator_definition}, we use the exact Weingarten coefficient of $1/(2n)$ for the corresponding moment operator over the orthogonal group $M_{\mathbb{O}(2n)}^{(1)} = M_{\mathbb{SO}(2n)}^{(1)}$, as provided in Eq.~(D7) in the Supplemental Information of Ref.~\cite{garcia-martin_QuantumNeuralNetworks_2025}. The result reads
               \begin{equation}\label{eq:single_maj_cov_exact}
                    \Sigma_{t,t'} = \kappa \left( \rho_t, \rho_{t'} \right)
                    = 
                    \frac{1}{2n} 
                    \left\lVert O \right\rVert_{\mathcal{B}_1}^2 \left\langle \rho_{t} , \rho_{{t'}} \right \rangle_{\mathcal{B}_1}\,,
                \end{equation}
                and matches the covariance in Eq.~\eqref{eq:SI_theo:1-main}, which was derived from the asymptotic tools.

            \subsubsection{Module \texorpdfstring{$\BC_2$}{B2} general implications}

                For $m = 2$, the rank-$m$ coefficient tensors of the projection onto the module $\BC_2$ of the observables and states, $\tobs$ and $\tstate^{(i_j)}$ can be reshaped as $2n \times 2n$ real square matrices, denoted as $\tobsmat$ and $\tstatemat^{(i_j)}$. The anti-symmetry is directly inherited, as $A = -A^\top$.                 Here, we leverage the fact that the ratios in the univariate equivalent formulation of Condition~\ref{cond:gaussian_dominate} in Proposition~\ref{prop:cond:univariate} simplify to Schatten 1- vs 2-norm ratios (denoted as $\lVert\cdot\rVert_2$ and $\lVert\cdot\rVert_1$, respectively) in this case:
                
                \begin{proposition}[Equivalent univariate formulation of Condition~\ref{cond:gaussian_dominate} for $\BC_2$]\label{supp:cor:two_majo_simplified_contraction_condition}
                    In the univariate $\BC_2$ case, Condition~\ref{cond:gaussian_dominate} is equivalent to
                    \begin{equation}\label{supp:eq:two_majo_uni_simplified_contraction_condition}
                        \frac
                        {\left\lVert \tobsbb \right\rVert_{2} }
                        {\left\lVert \tobsbb  \right\rVert_{1} }
                        \cdot
                        \frac
                        {\left\lVert {\tstatebb} \right\rVert_{2} }
                        {\left\lVert {\tstatebb} \right\rVert_{1} }
                        \in o(1) \,,
                    \end{equation}
                    with $\tobsbb  = \tobsmat^\top \tobsmat$ and ${\tstatebb} = \tstatemat^\top \tstatemat$.
                    Since each fraction alone is at most one,
                    strict asymptotic vanishing, i.e., $o(1)$, for one fraction suffices for Eq.~\eqref{supp:eq:two_majo_uni_simplified_contraction_condition} to hold.

                \end{proposition}
                \begin{proof}

                    A contraction over rank-2 tensors forms cycles of tensors, where such a cycle can always be expressed as the trace of a product of matrices. Disjoint contractions may occur, i.e., disjoint cycles, which relate to products of such traces of matrix products.

                    First, consider $k = 4$ for Condition~\ref{cond:gaussian_dominate} in the formulation of Proposition~\ref{prop:cond:univariate}, the ratio of Eq.~\eqref{eq:cond:univariate:imp} simplifies to
                    \begin{equation}\label{eq:app:two_majo_simplified_contraction_condition_traces}
                        \left\lvert
                        \frac
                        {\Tr\left[\tobsmat^\top \tobsmat \tobsmat^\top \tobsmat \right]}
                        {\Tr\left[\tobsmat^\top \tobsmat\right]^2}
                        \cdot
                        \frac
                        {\Tr\left[\tstatemat^\top \tstatemat \tstatemat^\top \tstatemat \right]}
                        {\Tr\left[\tstatemat^\top \tstatemat \right]^2}
                        \right\rvert
                    \end{equation}
                    for any non-pairwise $\sigma \in \mathfrak{B}_{mk/2}$. 
                    Note that, without loss of generality, the transposition can be applied and fixed arbitrarily since the matrices are invariant up to a sign due to the anti-symmetry, and the sign is irrelevant for the absolute value of the products and fractions. Moreover, since any non-zero disjoint contraction over $k=4$ tensors directly corresponds to two pairwise tensor matchings, such contractions are excluded in the numerator.
                    Let 
                    $\tobsbb  = \tobsmat^\top \tobsmat$
                    and ${\tstatebb} = \tstatemat^\top \tstatemat$
                    , which are  positive semidefinite.
                    Due to this positive semidefiniteness, $\lvert \mathbb{A} \rvert = \mathbb{A}$ holds, where $\vert\cdot\rvert$ is defined with respect to the operator square root as $\vert \mathbb{A} \rvert = \sqrt{\mathbb{A}^\dagger \mathbb{A}}$. 
                    Therefore, the traces in Eq.~\eqref{eq:app:two_majo_simplified_contraction_condition_traces} match the definition of ($p$-th powers of) the Schatten $p$-norm
                    \begin{equation}
                        \left\lVert \mathbb{A} \right\rVert_p^p = \Tr\left[ \left\lvert \mathbb{A}\right\rvert^p \right]\,,
                    \end{equation}
                    for $p = 1$ (trace norm) and $p = 2$ \ac{hs} norm. In terms of Schatten norms Eq.~\eqref{eq:app:two_majo_simplified_contraction_condition_traces} reads then as
                    \begin{equation}\label{eq:app}
                        \frac
                        {\left\lVert \tobsbb  \right\rVert_2^2}
                        {\left\lVert \tobsbb  \right\rVert_1^2}
                        \cdot
                        \frac
                        {\left\lVert {\tstatebb} \right\rVert_2^2}
                        {\left\lVert {\tstatebb} \right\rVert_1^2}\,
                        ,
                    \end{equation}
                    and, therefore, directly yields Eq.~\eqref{supp:eq:two_majo_uni_simplified_contraction_condition} as an equivalent condition for the case $k=4$ in Eq.~\eqref{eq:cond:univariate:imp} in Proposition~\ref{prop:cond:univariate}.

                    It remains to show that any higher (even) $k > 4$ for Condition~\ref{cond:gaussian_dominate} and equivalently its univariate formulation in Proposition~\ref{prop:cond:univariate} automatically holds if this condition (Eq.~\eqref{supp:eq:two_majo_uni_simplified_contraction_condition}) is satisfied, i.e., for $k = 4$. 
                    The (Schatten) norm monotonicity property is relevant for the following derivation, which is formally
                    decreasing monotonicity in $p$, i.e., $\left\lVert A \right\rVert_1 \geq \left\lVert A \right\rVert_p \geq \left\lVert A \right\rVert_{p'} \geq \left\lVert A \right\rVert_\infty$ in $1 \leq p \leq p' \leq \infty$. 
                    Now, the ratios of Eq.~\eqref{eq:cond:univariate:imp} in Proposition~\ref{prop:cond:univariate} for any even $k > 4$ can again be expressed in terms of Schatten norm ratios and related to the simplified $k=4$ condition in Eq.~\eqref{supp:eq:two_majo_uni_simplified_contraction_condition}. Considering contractions in the numerators that match all $k$ tensors, we recover the $k=4$ condition as
                    \begin{equation}\label{eq:B2_contraction_proof_bound_higher_k}
                        \left\lvert
                        \frac
                        {\Tr\left[\tobsbb^{k/2}\right]}
                        {\Tr\left[\tobsbb\right]^{k/2}}
                        \frac
                        {\Tr\left[{\tstatebb}^{k/2}\right]}
                        {\Tr\left[{\tstatebb}\right]^{k/2}}
                        \right\rvert
                        = 
                        \frac
                        {\left\lVert \tobsbb \right\rVert_{k/2}^{k/2}}
                        {\left\lVert \tobsbb \right\rVert_1^{k/2}}
                        \frac
                        {\left\lVert {\tstatebb} \right\rVert_{k/2}^{k/2}}
                        {\left\lVert {\tstatebb} \right\rVert_1^{k/2}}
                        = 
                        \left(
                        \frac
                        {\left\lVert \tobsbb \right\rVert_{k/2}}
                        {\left\lVert \tobsbb \right\rVert_1}
                        \frac
                        {\left\lVert {\tstatebb} \right\rVert_{k/2}}
                        {\left\lVert {\tstatebb} \right\rVert_1}
                        \right)^{k/2}
                        \leq
                        \left(
                        \frac
                        {\left\lVert \tobsbb \right\rVert_{2}}
                        {\left\lVert \tobsbb \right\rVert_1}
                        \frac
                        {\left\lVert {\tstatebb} \right\rVert_{2}}
                        {\left\lVert {\tstatebb} \right\rVert_1}
                        \right)^{k/2}
                        \in o(1)\,
                        .
                    \end{equation}
                    Otherwise, if the contractions are still non-pairwise tensor matchings but disjoint, the ratios of $k$ tensor contractions split into the product of ratios of $k_l< k$. Hence, such contractions yield products of traces of $k_l$ powers of (anti-symmetric) matrices. If any $k_l$ is odd, the entire ratio must vanish as an odd power of an anti-symmetric matrix is anti-symmetric and, hence, traceless.
                    For only even $k_l$, apply Eq.~\eqref{eq:B2_contraction_proof_bound_higher_k} for each $k_l$ instead of $k$ individually. Overall, any even $k > 4$ ratio of Condition~\ref{cond:gaussian_dominate} is always upper bounded by the Schatten 2 vs 1-norm ratio in Eq.~\eqref{supp:eq:two_majo_uni_simplified_contraction_condition}.

                    The final note that each fraction is bounded by $1$ follows directly from the same monotonicity ($\left\lVert \mathbb{A} \right\rVert_p \leq \left\lVert \mathbb{A} \right\rVert_1$). Finally, this $\mathcal{O}(1)$ upper bound on all contraction ratios automatically satisfies the mixed term condition in Eq.~\eqref{eq:cond:univariate:mix}, as the constant left-hand side is trivially strictly dominated by the inverse weights $1/w_{\sigma, \pi} \in \Omega(n)$.

                \end{proof}

                In addition, for the special case of $m=2$, the general odd moment statement of Proposition~\ref{prop:supp:vanish_odd_moments} can be extended as follows. 
                \begin{corollary}[Vanishing odd moments]\label{cor:B2_vanishing_odd_moments} 
                    In the $\BC_2$ case and under Condition~\ref{cond:full_support}, all odd moments vanish.
                \end{corollary}
                \begin{proof}
                    The moment operator involves linear combinations of contractions of an odd number $k$ of module frame coefficient tensors $\lbrace \tobs \rbrace^k$ and $\lbrace \tstate^{(i_j)} \rbrace_{j=1}^k$, respectively. As opposed to general $m$, for $m=2$ as in the previous Corollary, any contraction of the anti-symmetric rank-$2$ tensors can be decomposed into products of traces of anti-symmetric matrices through re-shaping of the tensors. This product contains at least one trace with an odd number of such matrices. Therefore, this trace vanishes since the product of an odd number of identical (hence commuting) anti-symmetric matrices is an anti-symmetric matrix with zero trace generally. This is the case for both the univariate and multivariate cases, as the matrix associated with the observable occurs identically in the product multiple times. Hence, any such contraction vanishes, which concludes the proof.
                \end{proof}

                Concluding the case of module $\BC_2$, we provide the exact covariance expression, which resembles the (exact) kernel function to build the quantum Gaussian process model. Instead of using the asymptotic moment operator in Eq.~\eqref{supp:eq:asymptotic_moment_operator_definition} of Lemma~\ref{supp:cor:asymp_moment_operator} to calculate the (asymptotic) second moment, we use the exact Weingarten coefficients for the corresponding moment operator over the orthogonal group $M_{\SO(2n)}^{(2)} $ in Eq.~(D10) in the supplemental material of Ref.~\cite{garcia-martin_QuantumNeuralNetworks_2025}, yielding
                \begin{equation}\label{eq:two_maj_cov_exact}
                    \Sigma_{t,t'} = \kappa \left( \rho_t, \rho_{t'} \right)
                    = 
                    \frac{1}{n(2n-1)} 
                    \left\lVert O \right\rVert_{\mathcal{B}_2}^2 \left\langle \rho_{t} , \rho_{{t'}} \right \rangle_{\mathcal{B}_2}
                    \,.
                \end{equation}
                Note that, as opposed to the case of $\BC_1$, we now observe a deviation from the asymptotic expression, cf. Eq.~\eqref{eq:SI_theo:1-main}. However, as expected for the second moments and the absence of non-pairwise tensor contractions as discussed in the proof of Lemma~\ref{thm:supp:general_asym_even_moments}, this deviation manifests itself purely in a scaling factor. 
                As expected, the exact factor is simply the inverse of the module dimension $\dim(\BC_2) = \binom{2n}{2} = n(2n-1)$, while the asymptotic factor relies on the inverse of the upper bound $\dim(\BC_2) \leq 2n^2$. Overall, the relative deviation is $1-1/(2n)$.

            \subsubsection{Module isomorphism}
                We can exploit the module isomorphism $\BC_{2n-m} \simeq \BC_m$ to directly transfer all \ac{gp} statements between the modules $\BC_{2n-m}$ and $\BC_m$, i.e., between products of $2n-m$ and $m$ Majorana operators. In particular, the isomorphic map between modules, denoted as $\Gamma(\cdot)=\cdot\, \Gamma$,  is achieved through multiplication by $\Gamma = i^n C_1 C_2\cdots C_{2n}$. It is easy to see that $C_{i_1}C_{i_2}\cdots C_{i_m} \Gamma$ gives a product of Majoranas in $\BC_{2n-m}$ (up to a phase), and that  multiplication by $\Gamma$ is its own inverse, i.e., $\Gamma^2 =\id$. This readily identifies $\BC_{2n-m}$ and $\BC_m$ as vector spaces. But more importantly, $\Gamma(\cdot)$ identifies $\BC_{2n-m}$ and $\BC_m$ as $\mathbb{SPIN}(2n)$ isomorphic modules: indeed, since $[U,\Gamma]=0$ for all matchgate unitaries $U$, it follows that $\Gamma \left(UC_\mu U^\dagger\right) = U\Gamma\left(C_\mu\right)U^\dagger$, so $\Gamma(\cdot)$ is equivariant and hence compatible with the $\mathbb{SPIN}(2n)$ group action. 

                For the analysis of quantum Gaussian processes, we need to establish how the moments transform under $\Gamma$. Considering the transformed observable and state components $O\Gamma$ and $\rho_{i_j}\Gamma$, then the moments are identical,
                \begin{equation}
                    \mathbb{E}_{U\sim\mathbb{SPIN}(2n)}\left[ \prod_{j=1}^k\Tr \left[\Gamma O  U \rho_{i_j} \Gamma U^\dagger \right]\right] =  \mathbb{E}_{U\sim\mathbb{SPIN}(2n)} \left[\prod_{j=1}^k\Tr \left[ O  U \rho_{i_j}  U^\dagger \right]\right]\,,
                \end{equation}
                where we simply used that $\Gamma$ is self-adjoint ($\Gamma = \Gamma^\dagger)$ and commutes with $U^\dagger$. Consequently, all statements regarding the convergence to Gaussian processes in the module $\BC_m$ can be automatically translated to the module $\BC_{2n-m}$.
                
\clearpage
        \section{Examples and special families of matchgate Gaussian processes}\label{app:examples_families_GPs}

            In this section, we verify common scenarios under which the conditions for the matchgate \ac{qgp} apply. This is specifically Condition~\ref{cond:gaussian_dominate}. Moreover, since violations of these conditions rule out \ac{gp} behavior through Theorem \ref{thm:supp:no_GP}, we identify scenarios on a case-by-case basis for which an asymptotic deviation from \acp{gp} is provable.

            \subsection{Special states and observables preliminaries}

                Because certain families of states and observables play a central role in the subsequent analyses, we first outline their key properties here.

                \subsubsection{Gaussian states}

                    Pure (fermionic) Gaussian states are quantum states that can be prepared by a matchgate (free-fermionic) transformation $U$ of a computational basis state $\ket{i}, i \in [2^n]$, i.e.,
                    \begin{equation}
                        \ket{\psi} = U\ket{i}
                        .
                    \end{equation}
                    More narrowly defined, all Gaussian states of even (odd) parity can be prepared by applying a free-fermionic transformation to the all-zero basis state $\ket{0}^{\otimes n}$ (respectively, to the first single-excitation state $\ket{10\ldots0} = \ket{1}\ket{0}^{\otimes n-1} = X_1\ket{0}^{\otimes n} = C_1\ket{0}^{\otimes n}$). That is
                    \begin{equation}
                        \ket{\psi_{\rm even}} = U\ket{0}^{\otimes n}
                        \quad
                        \text{and}
                        \quad
                        \ket{\psi_{\rm odd}} = U\ket{10\ldots0}
                        .
                    \end{equation}
                    Gaussian states of even and odd parity are also referred to as Gaussian states in the even and odd parity sectors, respectively. The even parity sector is spanned by (the union of) all even modules, while the odd parity sector is spanned by (the union of) all odd modules \cite{diaz2023showcasing}. Gaussian states are of high interest due to their connection to classical efficient simulability, which is that matchgate (free-fermionic) transformations of Gaussian states can be classically simulated in polynomial time \cite{jozsa2008matchgates}. Hence, as a necessary condition for classical intractability, states must be non-Gaussian.

                    We provide some further properties for even parity Gaussian states through their connection to the all-zero state $\rho_0 = \left({\ket{0}\!\bra{0}}\right)^{\otimes n}$. First, $\rho_0$ decomposes over products of the Pauli-Z strings as
                    \begin{equation}\label{eq:zero_pauli_decomposition}
                        \rho_0 =\left(\frac{\id+Z}{2}\right)^{\otimes n}=\frac{1}{d}\left(\id+\sum_i Z_i+\sum_{i<j} Z_i Z_j+\sum_{i<j<k} Z_i Z_j Z_k+\ldots\right)\,,
                    \end{equation}
                    with $d = 2^n$.
                    One can now show that the frame coefficient tensor for any even module $\BC_m$ (with $m$ even) of $\rho_0$ can be expressed via direct sums of generalized (anti-symmetric) Kronecker delta tensors as
                    \begin{equation}
                        \tstate_0 = \frac{1}{\sqrt{m!d}}
                            \bigoplus^n \delta_{\alpha_1\phantom{'} \alpha_2 \dots \alpha_m}^{\alpha'_1 \alpha'_2 \dots \alpha'_m}\,,
                    \end{equation}
                    with the ordered reference indices $\alpha'_1 < \alpha'_2 < \cdots < \alpha'_m$.

                    \paragraph{Properties of Gaussian states in \texorpdfstring{$\BC_2$}{B2}.}
                    As the case of $m=2$ is of particular interest in the subsequent section, we explicitly find the following $2n \times 2n$ frame coefficient matrix in a $2 \times 2$ block-diagonal form
                    \begin{equation}\label{supp:eq:Gaussian_state_Y_matrix_form}
                        \tstatemat_0 =
                        \frac{1}{\sqrt{2d}}J\,, \qquad J= \bigoplus_{i=1}^n \begin{pmatrix}
                            0 & 1 \\
                            -1 & 0
                        \end{pmatrix}\,,
                    \end{equation}
                    for $J$ the symplectic identity.
                    Notice that each $i$-th block is associated with the Pauli-Z on the $i$-th qubit (cf. second sum in the right-hand side in Eq.~\eqref{eq:zero_pauli_decomposition}) as the result of the product of two successive Majorana operators $Z_i = -iC_{2i-1}C_{2i}$. More precisely, the left Majorana operator in the product is indicated by the columns, while the rows indicate the right Majorana operator in the product.
                    This matrix $\tstatemat_0$ obeys a specific standard form for anti-symmetric matrices, i.e., any anti-symmetric matrix of even dimension can be brought into this $2 \times 2$ block-diagonal matrix under a special orthogonal transformation. Therefore, the eigenvalues of $\tstatemat_0$ are $\pm i / \sqrt{2d}$ according to the spectral theorem and the special form of $\tstatemat_0$, which directly allows for reading off the eigenvalues from the off-diagonal values in the blocks. Notably, all eigenvalues are of magnitude $1/\sqrt{2d}$.
                    Transitioning from the all-zero state to any Gaussian state of even parity, which are related by a matchgate (free-fermionic) transformation. Through the matchgate representation $R$, the effect of this transformation is reflected in a special orthogonal transformation of the anti-symmetric coefficient matrix $\tstatemat_0$. Hence, the coefficient matrix $\tstatemat$ for any Gaussian state is equivalent to the matrix $\tstatemat_0$ of the all-zero state under a special orthogonal transformation and, hence, has the same eigenvalues.

                    This relation allows one to characterize the Hilbert-Schmidt overlap $\braket{\tstatemat_j,\tstatemat_k}_{\rm HS}$ between the $m=2$ matrices of $\tstatemat_j, \tstatemat_k$ of two even-parity Gaussian states
                    \begin{equation}
                        \ket{\psi_j} = U_j \ket{0}^{\otimes n}\,,\qquad \ket{\psi_k} = U_k \ket{0}^{\otimes n}.
                    \end{equation}
                    Let $Q_j,Q_k \in \SO(2n)$ denote the corresponding orthogonal transformations acting over the Majorana indices, namely
                    \begin{equation}
                        U_j^\dagger C_\alpha U_j = \sum_{\beta=1}^{2n} [Q_j]_{\alpha\beta} C_\beta\,,\qquad
                        U_k^\dagger C_\alpha U_k = \sum_{\beta=1}^{2n} [Q_k]_{\alpha\beta} C_\beta\,.
                    \end{equation}
                    Then their coefficient matrices satisfy
                    \begin{equation}
                        \tstatemat_j = Q_j \tstatemat_0 Q_j^\top\,,\qquad
                        \tstatemat_k = Q_k \tstatemat_0 Q_k^\top\,.
                    \end{equation}
                    Now, introducing the relative rotation
                    \begin{equation}
                        R_{jk} = Q_j^\top Q_k \in \SO(2n)\,,
                    \end{equation}
                    the Hilbert-Schmidt product can be written as
                    \begin{align}
                        \braket{\tstatemat_j,\tstatemat_k}_{\rm HS}
                        &= \Tr[\tstatemat_j^\top \tstatemat_k] \\
                        &= \frac{1}{2d}\Tr\left[J^\top R_{jk} J R_{jk}^\top\right] \\
                        &= - \frac{1}{2d}\Tr\left[J R_{jk} J R_{jk}^\top\right]\,.
                    \end{align}
                    with $J$ as defined in Eq.~\eqref{supp:eq:Gaussian_state_Y_matrix_form}.
                    Thus, the overlap depends only on the relative orthogonal transformation $R_{jk}$. Now, for completely random states $\ket{\psi_j},\ket{\psi_k}$, we can only derive the trivial upper bound via the Cauchy-Schwarz inequality
                    \begin{equation}
                        \left|\braket{\tstatemat_j,\tstatemat_k}_{\rm HS}\right|
                        \leq
                        \left\lVert 
                        \tstatemat_j\right\rVert_{\rm HS}
                        \left\lVert \tstatemat_k\right\rVert_{\rm HS}
                        =\frac{n}{d}\,.
                    \end{equation}
                    However, we can derive average bounds by leveraging the orthogonal nature of $R_{jk}$. Indeed, if $U_j$ and $U_k$ are independently Haar-random in $\SPIN(2n)$, then $R_{jk}$ is Haar-random in $\SO(2n)$, and the matrix
                    \begin{equation}
                        S = R_{jk} J R_{jk}^\top
                    \end{equation}
                    is a random orthogonal complex structure. By $\SO(2n)$-invariance and anti-symmetry, when averaging over the random $R_{jk}$ its second moments take the form
                    \begin{equation}
                        \mathbb{E}\!\left[S_{\mu\nu} S_{\alpha\beta}\right]
                        =
                        \frac{1}{2n-1}
                        \left(
                            \delta_{\mu\alpha}\delta_{\nu\beta}
                            -
                            \delta_{\mu\beta}\delta_{\nu\alpha}
                        \right)\,,
                    \end{equation}
                    where the prefactor is fixed by the identity
                    \begin{equation}
                        \braket{S,S}_{\rm HS} = \braket{J,J}_{\rm HS} = 2n\,.
                    \end{equation}
                    It then follows that
                    \begin{equation}
                        \mathbb{E}\!\left[\braket{\tstatemat_j,\tstatemat_k}_{\rm HS}\right] = 0\,,
                    \end{equation}
                    and
                    \begin{align}
                        \Var\!\left[\braket{\tstatemat_j,\tstatemat_k}_{\rm HS}\right] &=  
                        \mathbb{E}\!\left[\braket{\tstatemat_j,\tstatemat_k}_{\rm HS}^2\right] \\
                        &= \frac{1}{4d^2}\,
                           \mathbb{E}\!\left[\braket{J,S}_{\rm HS}^2\right] \\
                        &= \frac{1}{4d^2}\,
                           \frac{4n}{2n-1} \\
                        &= \frac{n}{(2n-1)d^2}
                        \in \Theta\left(\frac{1}{d^2}\right)
                    \end{align}
                    Hence, for two independent Haar-random even-parity Gaussian states, one finds the large-$n$ upper and lower bounds
                    \begin{equation}\label{eq:HS_overlap_Yjk_scaling}
                        \mathrm{Std}\left[\braket{\tstatemat_j,\tstatemat_k}_{\rm HS}\right]
                        \in \Theta\left(\frac{1}{d}\right)\,,
                    \end{equation}
                    where ${\rm Std}[\cdot]$ denotes the standard deviation, meaning that the typical size of the inner product between the coefficient matrices of two random Gaussian states is $1/d$.
                    In other words, although all Gaussian states have the same norm in $\BC_2$, two generic such states become exponentially orthogonal in the large-$n$ limit.

                \subsubsection{Pauli observables}\label{supp:sec:Pauli_obs} 
    
As previously mentioned, Pauli observables are not only fully contained in a single module $\mathcal{B}_m$ but are also equivalent to a single element $C_\mathcal{M}$ of the module basis defined in Eq.~\eqref{supp:eq:module_basis} up to a phase. Let $\mathcal{M} = \lbrace \mu_1, \ldots, \mu_m \rbrace \subseteq \lbrace 1, \ldots, 2n\rbrace^m  $ denote the index set of the $m$ distinct Majorana operators composing the product associated with the Pauli observable. Similarly, the individual ($m!$ many) elements in the module frame (Eq.~\eqref{supp:eq:module_frame}), corresponding to permutations of the Majorana products of this module basis element, are equally equivalent to the Pauli observable.
                    Therefore, the coefficient in the tensors $\tobs$ of the projection of a Pauli observable $\hat{O}$ onto the module $\BC_m$ only exhibit $m!$ many non-zero components corresponding to the $m$ distinct Majorana operators referenced by the index set $\mathcal{M}$. 
                    Formally, with fixing the order $\mu_1 < \mu_2 \ldots < \mu_m$ in the unique index vector $\bm{\mu}$ for the Pauli observable, the only non-zero elements are
                    \begin{equation}
                        o_{\pi (\bm{\mu})} = \mathrm{sgn}(\pi) \sqrt{\frac{d}{m!}}
                    \end{equation}
                    for all index permutations $\pi(\bm{\mu})$ with $\pi \in S_m$ where $\mathrm{sgn}(\pi) \in \lbrace -1, 1 \rbrace$ denotes the sign of the permutation $\pi$.
                    Importantly, this precisely recovers the (rank-$m$) generalized (anti-symmetric) Kronecker delta tensor for the non-zero part 
                    \begin{equation}
                        o_{\mu_{1} \mu_{2} \ldots \mu_{m}} = \sqrt{\frac{d}{m!}} \delta_{\mu_1\phantom{'} \mu_2 \dots \mu_m}^{\mu'_1 \mu'_2 \dots \mu'_m}\,,
                    \end{equation}
                    with the ordered reference indices $\mu'_1 < \mu'_2 < \cdots < \mu'_m$.
                    Any other choice of indices $\bm{\nu}$, which is not equivalent to $\bm{\mu}$ up to permutation, results in a zero component $o_{\bm{\nu}} = 0$. Effectively, the rank-$m$ frame coefficient tensors $\tobs$ for a Pauli observable reduce from size $2n$ to size $m$ per dimension.

            \subsection{Special GP families in \texorpdfstring{$\BC_2$}{B2}}
                In this section, we classify various families of observables, states, and their combinations based on whether they induce GPs within the $\BC_2$ module.
            
                \subsubsection{GPs from Gaussian states in \texorpdfstring{$\BC_2$}{B2}}

                    In the univariate case of a single (pure) Gaussian state $\rho$, checking Condition~\ref{cond:gaussian_dominate}, which was simplified for $\BC_2$ to a condition on fractions of norms of ${\tstatebb} = \tstatemat^\top \tstatemat$, as stated in Eq.~\eqref{supp:eq:two_majo_uni_simplified_contraction_condition} of Proposition~\ref{supp:cor:two_majo_simplified_contraction_condition}, becomes trivial. 
                    Since Gaussian states lead to the special form as defined in Eq.~\eqref{supp:eq:Gaussian_state_Y_matrix_form}, we use ${\tstatebb} = \tstatemat^\top \tstatemat = - \tstatemat^2 = 1/(d\cdot m!) \id $. Hence, the eigenvalues of ${\tstatebb}$ are all $1/(d\cdot m!)$ as the negative squares of the purely imaginary eigenvalues $\pm i /\sqrt{d\cdot m!}$ of $\tstatemat$. Because the norm ratio is scale-invariant, the norms of the identity differ by a factor of $n$:
                    \begin{equation}
                        \frac
                        {\lVert {\tstatebb} \rVert_2^2} 
                        {\lVert {\tstatebb} \rVert_1^2}
                        =
                        \frac
                        {\lVert \id \rVert_2^2} 
                        {\lVert \id \rVert_1^2}
                        =
                        \frac
                        {2n} 
                        {\left(2n\right)^2}
                        = 
                        \frac{1}{2n}
                        \in \OC\left(\frac{1}{n}\right) \subset o(1)
                        .
                    \end{equation}
                    Since further, for any observable $\O \in \mathcal{B}_2$, the corresponding norm fraction ${\lVert \tobsbb \rVert_2^2}{\lVert \tobsbb \rVert_1^2} \leq 1$, Condition~\ref{cond:gaussian_dominate} is fulfilled. The remaining Condition~\ref{cond:full_support} is trivially fulfilled by the assumption of an arbitrary observable $\O$ in the module $\mathcal{B}_2$.  Therefore, Theorem \ref{thm:supp:GP} holds, and  matchgate transformations induce univariate \acp{gp} from any Gaussian state and any observable $\O \in \mathcal{B}_2$.

                    The univariate case is clearly special, since Proposition~\ref{supp:cor:two_majo_simplified_contraction_condition} reduces Condition~\ref{cond:gaussian_dominate} to the asymptotic vanishing of a ratio of Schatten norms of $\tstatebb=\tstatemat^\top \tstatemat$, which we have easily shown to always hold. In the multivariate case, however, the relevant quantity is no longer a single positive ratio, but rather the full sum $\Sperf$ of pairwise matching contributions appearing in Eq.~\eqref{eq:Sperf}. This makes it harder in general to check for the emergence of a provable GP. In the following however, we show that for a dataset of positively correlated Gaussian states one can derive a condition on the covariances that guarantees the satisfaction of Condition~\ref{cond:gaussian_dominate}.

                    \begin{proposition}[A multivariate QGP in $\BC_2$]\label{thm:multivariate_gp_b2_gaussian}
                        Given a dataset of even-parity Gaussian states $\rho_i$, and an observable $O\in\BC_2$, the condition
                        \begin{equation}
                            \Braket{\rho_{i},\rho_{j}}_{\BC_2}
                            \in 
                            \omega\left(\frac{n^{2/3}}{d}\right)
                            \qquad \forall\, i\neq j\,\,
                        \end{equation}
                        is enough to satisfy Condition~\ref{cond:gaussian_dominate} and hence ensure the emergence of a multivariate GP.
                    \end{proposition}

                    \begin{proof}
                        Let us first recall that, for an observable $O\in \BC_2$ we are guaranteed that the ratios of observable frame-coefficient contractions are at most 1 (see Schatten norm equivalence and note in Proposition~\ref{thm:B2_GP_ratio}), and hence to prove the emergence of a multivariate GP we can simply study the asymptotics of the states' contributions.
                        Now, let $\rho_i$ be even-parity Gaussian states, and let $\tstatemat^{(i)}$ be their corresponding $2n\times 2n$ frame-coefficient matrix. One readily finds that $\tstatemat^{(i)}=Q_i\tstatemat_0 Q_i^\top$, for $Q_i\in\SO(2n)$ and $\tstatemat_0=\frac{1}{\sqrt{2d}}J$, with $J$ being the symplectic identity. Hence, all $\tstatemat^{(i)}$ share the same eigenvalues as $\tstatemat_0$, which have magnitude $1/\sqrt{2d}$. In particular, we have
                        \begin{equation}
                            \|\tstatemat^{(i)}\|_\infty = \frac{1}{\sqrt{2d}}\,.
                        \end{equation}
                        Now consider, for any even $k\geq 4$ moment, a non-pairwise contraction appearing in $\Simp$ in Condition~\ref{cond:gaussian_dominate}, associated to the Brauer pairing $\sigma$. Since the frame-coefficient tensors in $\BC_2$ are anti-symmetric matrices $\tstatemat^{(i_j)}$, the only Brauer pairings $\sigma$ yielding non-automatically vanishing contractions due to self-loops are those that can be interpreted as products of cycles between the $k$ frame-coefficient matrices involved. Let us denote by ${\rm cycles}(\sigma)$ the set of such cycles, by $\CC$ one such cycle, and by $\ell_\CC$ the length of $\CC$. The Brauer pairing $\sigma$ further contains permutations of the indices of the indices of each individual frame-coefficient matrix. However, the latter can only induce transpositions of the $\tstatemat^{(i_j)}$, that is a sign since they are anti-symmetric. On the other hand, each cycle $\CC\in {\rm cycles}(\sigma)$ contributes a trace of the product of the frame-coefficient matrices picked up by $\CC$.
                        For later convenience, let us split the cycles of $\sigma$ into those of length $2$, whose set we denote by $\mathrm{cycles}_{2}(\sigma)\subset {\rm cycles}(\sigma)$, and those of length strictly larger than $2$, whose set we similarly denote by $\mathrm{cycles}_{>2}(\sigma)\subseteq {\rm cycles}(\sigma)$. 
                        Let us denote by $|\mathrm{cycles}_{>2}(\sigma)|$ the number of the latter appearing in ${\rm cycles}(\sigma)$. Note that since by definition $\sigma$ is non-pairwise it must be that $|\mathrm{cycles}_{>2}(\sigma)|\geq1$.
                        The full state contribution magnitude coming from $\sigma$ then factorizes as
                        \begin{equation}
                            \left\lvert\operatorname{contr}_{\sigma}\!\left( \bigotimes_{j=1}^k \tstatemat^{(i_j)}\right) \right\rvert
                            =
                            \left\lvert
                            \prod_{(j, j')\in\mathrm{cycles}_{2}(\sigma)}
                            \Braket{\tstatemat^{(i_j)},\tstatemat^{(i_{j'})}}_{\rm HS}
                            \cdot
                            \prod_{\CC \in \mathrm{cycles}_{>2}(\sigma)}
                            \left(
                            \tr\!\left[
                                \prod_{i_j\in \CC} \tstatemat^{(i_j)}
                            \right]
                            \right)
                            \right\rvert
                            \,.
                        \end{equation}
                        For each cycle $\CC\in \mathrm{cycles}_{>2}(\sigma)$, the corresponding state contribution satisfies
                        \begin{equation}\label{eq:new_b2_upperbound}
                            \left|
                            \tr\!\left[
                                \prod_{i_j\in \CC} \tstatemat^{(i_j)}
                            \right]
                            \right|
                            \leq
                            \operatorname{rank}\!\left(
                                \prod_{i_j\in \CC} \tstatemat^{(i_j)}
                            \right)
                            \left\|
                                \prod_{i_j\in \CC} \tstatemat^{(i_j)}
                            \right\|_\infty
                            \leq
                            2n \prod_{i_j\in \CC} \|\tstatemat^{(i_j)}\|_\infty
                            =
                            \frac{2n}{(2d)^{\ell_\CC/2}}\,.
                        \end{equation}
                        Using this we can bound the full state contribution associated to $\sigma$ by
                        \begin{equation}\label{eq:b2_nonpair_upperbound_refined}
                            \left|
                            \operatorname{contr}_{\sigma}\!\left( \bigotimes_{j=1}^k \tstatemat^{(i_j)}\right)
                            \right|
                            \leq
                            \prod_{(j, j')\in\mathrm{cycles}_{2}(\sigma)}
                            \Braket{\tstatemat^{(i_j)},\tstatemat^{(i_{j'})}}_{\rm HS}
                            \frac{(2n)^{|\mathrm{cycles}_{>2}(\sigma)|}}{(2d)^{\sum_{\CC\in\mathrm{cycles}_{>2}(\sigma)}\ell_\CC/2}}\,.
                        \end{equation}
    
                        On the other hand, each state contribution term in $\Sperf$ associated to a pairwise matching $\bm{\sigma}$ is of the form
                        \begin{equation}\label{eq:b2_parirings_denominator}
                            \prod_{(j,j')\in \bm{\sigma}}
                            \Braket{\rho_{i_j},\rho_{i_{j'}}}_{\BC_2}
                            =
                            \prod_{(j,j')\in \bm{\sigma}}
                            \Braket{\tstatemat^{(i_j)},\tstatemat^{(i_{j'})}}_{\rm HS}\,.
                        \end{equation}
                        
                        We now use the bound from Eq.~\eqref{eq:b2_nonpair_upperbound_refined} to study the asymptotics of the ratio
                        \begin{equation}\label{eq:b2_ratio_gaussian}
                            \frac{\left|
                            \operatorname{contr}_{\sigma}\!\left( \bigotimes_{j=1}^k \tstatemat^{(i_j)}\right)
                            \right|}{\prod_{(j,j')\in \bm{\sigma}}
                            \Braket{\tstatemat^{(i_j)},\tstatemat^{(i_{j'})}}_{\rm HS}}\,,
                        \end{equation}
                        Let us start by assuming that all pairwise module overlaps are non-negative, so as to avoid cancellations in $\Sperf$ and allows us to analyze the ratios in Eq.~\eqref{eq:b2_ratio_gaussian} independently as per Proposition~\ref{prop:cond:sufficient_mutivariate}, sufficient to fulfill Condition~\ref{cond:gaussian_dominate}. We then compare a non-pairwise contribution with a pairwise matching contribution associated to some $\bm{\sigma}$ chosen so as to contain all the $2$-cycles already appearing in $\sigma$ as Proposition~\ref{prop:cond:sufficient_mutivariate} requires only the existence of such a pairwise matching $\bm{\sigma}$. In this way, the pairwise factors already present in Eq.~\eqref{eq:b2_nonpair_upperbound_refined} also appear in the chosen denominator in Eq.~\eqref{eq:b2_ratio_gaussian} and can be canceled out.
                        Now, writing
                        \begin{equation}
                            \Braket{\rho_i,\rho_j}_{\BC_2}
                            =
                            \frac{n}{d}\,r_{ij},
                        \end{equation}
                        where $n/d$ is the largest possible module overlap inside $\BC_2$ \cite{deneris2025analyzing} and $0\leq r_{ij}\leq 1$ denotes the corresponding normalized overlap, the remaining pairwise matching contribution reads
                        \begin{equation}
                            \prod_{(j,j')\in \bm{\sigma}\setminus\mathrm{cycles}_{2}(\sigma)}
                            \Braket{\rho_{i_j},\rho_{i_{j'}}}_{\BC_2}
                            =
                            \left(\frac{n}{d}\right)^{k/2 - |\mathrm{cycles}_{2}(\sigma)|}
                            \prod_{(j,j')\in \bm{\sigma}\setminus\mathrm{cycles}_{2}(\sigma)}
                            r_{i_j i_{j'}}\,.
                        \end{equation}
                        Note that $k/2 - |\mathrm{cycles}_{2}(\sigma)|=\sum_{\CC\in\mathrm{cycles}_{>2}(\sigma)}\ell_\CC/2$.
                        For the corresponding pairwise matching term to dominate the non-pairwise contribution, namely, the ratio in Eq.~\eqref{eq:b2_ratio_gaussian} being $o(1)$, the strict asymptotic lower bound
                        \begin{equation}
                            \prod_{(j,j')\in \bm{\sigma}\setminus\mathrm{cycles}_{2}(\sigma)}
                            r_{i_j i_{j'}}
                            \in \omega \left(
                            (2n)^{|\mathrm{cycles}_{>2}(\sigma)|-\sum_{\CC\in\mathrm{cycles}_{>2}(\sigma)}\ell_\CC/2} \right)\,
                        \end{equation}
                        must hold.
                        If we now further impose a uniform lower bound on the normalized overlaps, namely
                        \begin{equation}
                            r_{ij}\geq r_*(n)
                            \qquad \forall\, i\neq j,
                        \end{equation}
                        then a sufficient condition is obtained by requiring
                        \begin{equation}
                            r_*(n)^{\sum_{\CC\in\mathrm{cycles}_{>2}(\sigma)}\ell_\CC/2}
                            \in \omega\left(
                            n^{|\mathrm{cycles}_{>2}(\sigma)|-\sum_{\CC\in\mathrm{cycles}_{>2}(\sigma)}\ell_\CC/2}\right)\,.
                        \end{equation}
                        Equivalently, this reads
                        \begin{equation}
                            r_*(n)
                            \in \omega\left(
                            n^{\frac{2|\mathrm{cycles}_{>2}(\sigma)|}{\sum_{\CC\in\mathrm{cycles}_{>2}(\sigma)}\ell_\CC}-1}\right)\,.
                        \end{equation}
                        Now, since every cycle in $\mathrm{cycles}_{>2}(\sigma)$ has length at least $3$, one has
                        \begin{equation}
                            \sum_{\CC\in\mathrm{cycles}_{>2}(\sigma)}\ell_\CC
                            \geq
                            3|\mathrm{cycles}_{>2}(\sigma)|\,,
                        \end{equation}
                        and thus
                        \begin{equation}
                            \frac{2|\mathrm{cycles}_{>2}(\sigma)|}{\sum_{\CC\in\mathrm{cycles}_{>2}(\sigma)}\ell_\CC}-1
                            \leq
                            -\frac{1}{3}\,.
                        \end{equation}
                        Hence, a sufficient condition on the individual normalized overlaps is the following strict asymptotic lower bound
                        \begin{equation}\label{eq:overlap_bound_multi_B2}
                            r_{ij}\in \omega\left(n^{-1/3}\right)
                            \qquad \forall\, i\neq j\,.
                        \end{equation}
                        Thus, assuming non-negative module overlaps, it is sufficient that all normalized overlaps $r_{ij}$ decay more slowly than $n^{-1/3}$ in order for the pairwise matching contributions to dominate the non-pairwise ones. Since as we discussed in the beginning, we are guaranteed that observable contributions are at most 1, this overall fulfills the sufficient formulation of Condition~\ref{cond:gaussian_dominate} via Proposition~\ref{prop:cond:sufficient_mutivariate} by satisfying Eq.~\eqref{eq:cond:pos_multi:imp} and trivially Eq.~\eqref{eq:cond:pos_multi:mix}. Consequently, with ensured zero odd moments in $\BC_2$ (Corollary~\ref{cor:B2_vanishing_odd_moments}), a set of Gaussian states with bounded module overlaps as in Eq.~\eqref{eq:overlap_bound_multi_B2}, which is the statement of Proposition~\ref{thm:multivariate_gp_b2_gaussian} once the normalization factor $n/d$ is reinserted, is guaranteed to form multivariate QGPs for any observable $O \in \BC_2$.
                    \end{proof}

                \subsubsection{Pauli observables (in \texorpdfstring{$\BC_2$}{B2}) alone cannot lead to GPs}

                    Since any Pauli observable is fully contained in only a single module (see Sec.~\ref{supp:sec:Pauli_obs}), Condition~\ref{cond:full_support} is trivially fulfilled.
                    As we are furthermore considering the $m=2$ Majorana case, Proposition~\ref{supp:cor:two_majo_simplified_contraction_condition} provides an equivalent simplification for Condition~\ref{cond:gaussian_dominate} such that it suffices to check the scaling of the Schatten norm fractions 
                    $
                        {\left\lVert \tobsbb \right\rVert_{2} }
                        /
                        {\left\lVert \tobsbb \right\rVert_{1} }  
                    $ associated with the observable.

                    Here, $\tobsbb = O^\top O$ where $O$ denotes the matrix $2n \times 2n$ reflecting the rank-$2$ tensors $\tobs$ of the frame coefficients when projecting the observable $\hat{O}$ onto the module $\mathcal{B}_2$.
                    
                    To demonstrate that Pauli observables alone cannot lead to GPs, we show that these fractions remain constant in scaling $n$, which leaves the strict asymptotic vanishing (Condition~\ref{cond:gaussian_dominate}) to the specific choice of states. 
                    As highlighted in Sec.~\ref{supp:sec:Pauli_obs}, considering Pauli observables reduces to the study of a single rank-$m$ generalized (anti-symmetric) Kronecker delta tensor up to a factor.
                    Hence, reshaped into a $2\times 2$ matrix, we uncover the $2\times 2$ identity as
                    \begin{equation}
                        \tobsbb 
                        =
                        O^\top O
                        =
                        \sqrt{\frac{d}{m!}}
                        \begin{pmatrix}
                            0 & 1 \\
                            -1 & 0
                        \end{pmatrix}
                        \sqrt{\frac{d}{m!}}
                        \begin{pmatrix}
                            0 & -1 \\
                            1 & 0
                        \end{pmatrix}
                        = {\frac{d}{m!}} \id
                        .
                    \end{equation}
                    Then, we simply find the two norms as
                    \begin{align}
                        \left\lVert \tobsbb \right\rVert_{1}^2 
                        &= \left({\frac{d}{m!}}\right)^2 \left\lVert \id \right\rVert_{1}^2 
                        = 
                        \left({\frac{d}{m!}}\right)^2 (m!)^2
                        =
                        {d}^2, \\
                        \left\lVert \tobsbb \right\rVert_{2}^2 
                        &= {\frac{d}{2}} \left\lVert \id \right\rVert_{2}^2 
                        = \left({\frac{d}{m!}}\right)^2 m!
                        = 
                        {\frac{d^2}{m!}}
                        .
                    \end{align}
                    Their ratio is constant
                    \begin{equation}
                        \frac{\left\lVert \tobsbb \right\rVert_{2} }
                        {\left\lVert \tobsbb \right\rVert_{1} }   
                        = \frac{1}{\sqrt{m!}}
                        \in \Theta(1)
                        \nsubseteq
                        o(1)
                        ,
                    \end{equation}
                    i.e., does not strictly vanish asymptotically as functions in $o(1)$,
                    since $m=2$ does not scale with $n$.

                    The scaling of the Pauli observable norm fractions cannot meet Eq.~\eqref{supp:eq:two_majo_uni_simplified_contraction_condition} in Proposition~\ref{supp:cor:two_majo_simplified_contraction_condition} (equivalent simplification of Condition~\ref{cond:gaussian_dominate}).

                    Hence, the statement transfers that a Pauli observable alone is not sufficient to obtain a GP, and a state family must be chosen to satisfy Condition~\ref{cond:gaussian_dominate}. Proposition~\ref{supp:cor:two_majo_simplified_contraction_condition} again provides a simplified condition in terms of Schatten norm ratios of the state module coefficients for the univariate case.

                \subsubsection{GPs from random two Majorana product observables                }

                    We study how inside the module $\BC_2$, a random observable can yield a Gaussian process without imposing strong conditions on the dataset of states. In Supp.~Fig.~\ref{fig:random_o_2majos} we show the numerical analysis that led us to prove that the ratios $(\vert\vert \tobsbb \vert\vert_\ell / \vert\vert \tobsbb \vert\vert_1)^\ell$ scale as $\OC(1/(2n)^{\ell-1})$ in the large-$n$ limit.
                    \begin{figure}[tbp]
                        \centering
                        \includegraphics[width=\linewidth]{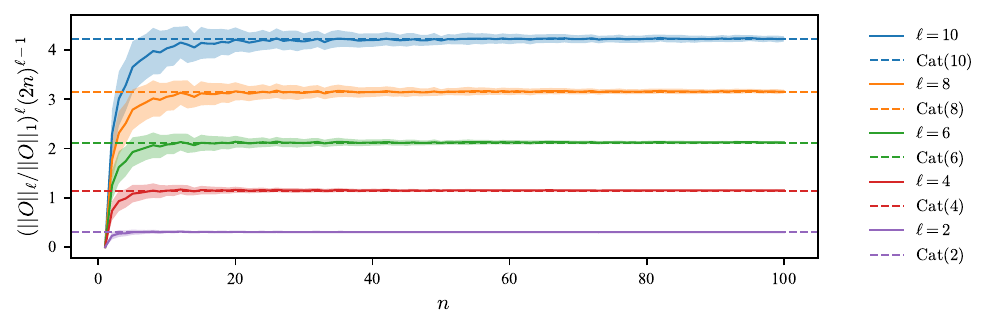}
                        \caption{\textbf{Numerical verification of asymptotic GP conditions satisfaction.} Per Theorem~\ref{thm:B2_GP_ratio}, an observable $\hat{O} \in \BC_2$ with random i.i.d. $\NC(0, 1)$ coefficients shows vanishing coefficient tensor norm ratios $(\vert\vert O \vert\vert_\ell / \vert\vert O \vert\vert_1)^\ell$ towards $\mathrm{Cat}(\ell) / (2n)^{\ell-1}$, guaranteeing GPs in the large-$n$ limit. 
                        The solid lines and associated shaded regions represent the mean and standard deviation of $(\vert\vert O \vert\vert_\ell / \vert\vert O \vert\vert_1)^\ell (2n)^{\ell-1}$ in log-space over $100$ samples of $\hat{O}$. The dashed lines instead represent the theoretically predicted convergence value ${\rm Cat}(\ell)$. The colors correspond to different values of $\ell$, as stated in the legend.}
                        \label{fig:random_o_2majos}
                    \end{figure}
                    We consider observables $\hat{O}=-i\sum_{\mu,\nu} \tobselem_{\mu\nu} \, C_\mu C_\nu$ where the coefficients $\tobselem_{\mu,\nu}\sim\NC(0,1)$ are sampled i.i.d. from a Gaussian distribution of zero mean and unit variance. For increasing system sizes $n=1,\dots,100$ we sample $100$ independent observables $\hat{O}$ and use their $(2n \times 2n)$ coefficient matrices $\tobsmat$ and construct $\tobsbb=O^\top O$. We then compute the random variable $(\vert\vert \tobsbb \vert\vert_\ell / \vert\vert \tobsbb \vert\vert_1)^\ell$ for $\ell=2,4,\dots,10$ and finally compute its mean and standard deviation. Supp.~Fig.~\ref{fig:random_o_2majos} clearly shows that after an initial decaying phase at small $n$, these ratios converge to an exact multiple of $(2n)^{\ell-1}$, and in the following we prove why.

                    \begin{proposition}\label{thm:B2_GP_ratio}
                        Let $O \in \BC_2$ be a random observable with i.i.d. $\NC(0, 1)$ coefficients in the $\BC_2$ basis, defined in Eq.~\eqref{supp:eq:module_basis}.
                        In the large-$n$ limit, the ratio $(\vert\vert \tobsbb \vert\vert_\ell / \vert\vert \tobsbb \vert\vert_1)^\ell$ converges almost surely\footnote{That is, with overwhelming probability over the sampling of the random operator $O$.} to
                        \begin{equation}\label{eq:theorem:random_B2_obs_vanish}
                            (\vert\vert \tobsbb \vert\vert_\ell / \vert\vert \tobsbb \vert\vert_1)^\ell \longrightarrow \frac{\rm{Cat}(\ell)}{(2n)^{\ell-1}}\,,
                        \end{equation}
                        where $\rm{Cat}(\ell)=\frac{1}{\ell +1}\binom{2\ell}{\ell}$ is the $\ell$-th Catalan number and constant in $n$.
                    \end{proposition}

                    \begin{proof}

                        Let us recall that $\tobsbb=O^\top O$ for $O$ the random anti-symmetric $2n\times 2n$ matrix whose independent entries are sampled i.i.d. from $\NC(0,1)$.
                        Let us now define the Hermitian matrix
                        \begin{equation}
                            H = \frac{i}{\sqrt{2n}} O\,.
                        \end{equation}
                        Notice that by construction the diagonal entries of $H$ vanish, while its off-diagonal entries have mean zero and variance $1/2n$, so $H$ is a Wigner-type matrix.
                        
                        In terms of $H$, the ratios can be readily found to read
                        \begin{equation}
                        \label{ap-eq:ratio_from_wigner}
                            (\vert\vert \tobsbb \vert\vert_\ell / \vert\vert \tobsbb \vert\vert_1)^\ell
                            =
                            \frac{1}{(2n)^{\ell-1}}\frac{\frac{1}{2n}\mathrm{Tr}(H^{2\ell})}{\left(\frac{1}{2n}\mathrm{Tr}(H^2)\right)^\ell}\,.
                        \end{equation}
                        We explicitly factorized out the $(2n)^{\ell-1}$ factor, this way the whole problem has reduced to understanding the normalized traces $\frac{1}{2n}\mathrm{Tr}(H^{2\ell}) $ as $n \to \infty$.
                        
                        Now, let $\lambda_1^{(n)},\dots,\lambda_{2n}^{(n)}$ be the eigenvalues of $H$:
                        \begin{equation}
                            \frac{1}{2n}\mathrm{Tr}(H^{2\ell})
                            =
                            \frac{1}{2n}\sum_{j=1}^{2n} \left(\lambda_j^{(n)}\right)^{2\ell}\,.
                        \end{equation}
                        By the Wigner semicircle law, the eigenvalue distribution of $H$ converges, in the large-$n$ limit, almost surely to the semicircle density
                        \begin{equation}
                            \rho_{\mathrm{sc}}(x) =
                            \frac{1}{2\pi}\sqrt{4-x^2}\,\mathbf{1}_{[-2,2]}(x)\,,
                        \end{equation}
                        for $\mathbf{1}_{[-2,2]}(x)$ the indicator function on the interval $[-2,2]$, namely the function
                        \begin{equation}
                            \mathbf{1}_{[-2,2]}(x) 
                            =
                            \begin{cases}
                                1\quad &{\rm if}\,x\in[-2,2]\\
                                0&{\rm else}
                            \end{cases}\,.
                        \end{equation}
                        Hence, for any integer $\ell \geq 1$,
                        \begin{equation}
                            \frac{1}{2n}\mathrm{Tr}(H^{2\ell})
                            \underset{\rm a.s.}{\longrightarrow}
                            \int_{-2}^{2} x^{2\ell}\rho_{\mathrm{sc}}(x)\,dx\,.
                        \end{equation}
                        Using this, we see that all we need are the even moments of the semicircle law.
                        
                        Let us denote the latter as
                        \begin{equation}
                            \mu_{2\ell}
                            =
                            \int_{-2}^{2} x^{2\ell}\rho_{\mathrm{sc}}(x)\,dx
                            =
                            \frac{1}{2\pi}\int_{-2}^{2} x^{2\ell}\sqrt{4-x^2}\,dx\,.
                        \end{equation}
                        We can use the fact that the integrand is even to write
                        \begin{equation}
                            \mu_{2\ell}
                            =
                            \frac{1}{\pi}\int_{0}^{2} x^{2\ell}\sqrt{4-x^2}\,dx\,.
                        \end{equation}
                        Now let us perform the change of variable $x=2\sqrt{t}$, this leads to
                        \begin{equation}
                            \mu_{2\ell}
                            =
                            \frac{2^{2\ell+1}}{\pi}
                            \int_0^1 t^{\ell-\frac12}(1-t)^{\frac12}\,dt\,.
                        \end{equation}
                        One can then recognize in the expression above the Beta function,
                        \begin{equation}
                            \mu_{2\ell}
                            =
                            \frac{2^{2\ell+1}}{\pi}
                            B\left(\ell+\frac12,\frac32\right)\,.
                        \end{equation}
                        Then, using $B(a,b)=\frac{\Gamma(a)\Gamma(b)}{\Gamma(a+b)}$ and $\Gamma(3/2)=\sqrt{\pi}/2$, we obtain
                        \begin{equation}
                            \mu_{2\ell}
                            =
                            \frac{2^{2\ell}}{\sqrt{\pi}}
                            \frac{\Gamma\left(\ell+\frac12\right)}{\Gamma(\ell+2)}\,.
                        \end{equation}
                        Lastly, by using
                        \begin{equation}
                            \Gamma\left(\ell+\frac12\right)
                            =
                            \frac{(2\ell)!}{4^\ell \ell!}\sqrt{\pi}
                            \quad\text{and}\quad
                            \Gamma\left(\ell+2\right) = \left(\ell + 1\right)!,
                        \end{equation}
                        we find
                        \begin{equation}
                            \mu_{2\ell}
                            =
                            \frac{(2\ell)!}{\ell!(\ell+1)!}
                            =
                            \frac{1}{\ell+1}\binom{2\ell}{\ell}
                            =
                            \mathrm{Cat}(\ell)\,.
                        \end{equation}
    
                        Plugging this back in Eq.~\eqref{ap-eq:ratio_from_wigner}, we conclude the proof.

                        Finally, we note that fluctuations over the convergence scale as~\cite{anderson2010introduction}
                        \begin{equation}
                            (\vert\vert \tobsbb \vert\vert_\ell / \vert\vert \tobsbb \vert\vert_1)^\ell \sim \frac{\rm{Cat}(\ell)}{(2n)^{\ell-1}}+\mathcal{O}\left(\frac{1}{n^l}\right)\,,
                        \end{equation}
                    which implies through a standard upper tail bound that probability of deviation decays exponentially with $n$.
        
                    \end{proof}

                    In conclusion, the vanishing Schatten norm ratio in Eq.~\eqref{eq:theorem:random_B2_obs_vanish} for the case $\ell = 2$ suffices already to guarantee that Condition~\ref{cond:gaussian_dominate} is met via Proposition~\ref{supp:cor:two_majo_simplified_contraction_condition} in the univariate case. Therefore, GPs emerge in the large-$n$ limit with overwhelming probability when sampling such random $\BC_2$ observables $\hat{O}$

                \subsubsection{GPs from magic states in \texorpdfstring{$\BC_2$}{B2}}

                    Consider now the family of magic states, known as the extent state \cite{diaz2023showcasing},
                    \begin{equation}
                        \ket{\psi_\theta}
                        =
                        \left(
                        \cos\left(\frac{\theta}{4}\right)\ket{0000}
                        +
                        \sin\left(\frac{\theta}{4}\right)\ket{1111}
                        \right)^{\otimes n/4}\,,
                    \end{equation}
                    for $n$ divisible by $4$, and let $\rho_\theta = \ket{\psi_\theta}\!\bra{\psi_\theta}$. In the univariate case of a single such state, checking Condition~\ref{cond:gaussian_dominate} for the module $\BC_2$ is again immediate via Proposition~\ref{supp:cor:two_majo_simplified_contraction_condition}. Indeed, by the module frame conventions, the entries of the two-Majorana coefficient matrix $\tstatemat$ are
                    \begin{equation}
                        \tstateelem_{\alpha\beta}
                        =
                        \tr\left[C_{\alpha\beta}\rho_\theta\right]
                        =
                        \frac{-i}{\sqrt{2d}}\tr\left[C_\alpha C_\beta\rho_\theta\right]\,.
                    \end{equation}
                    Hence, it suffices to characterize the two-point Majorana correlators of $\rho_\theta$.

                    Let us first consider a single four-qubit block
                    \begin{equation}
                        \ket{\phi_\theta}
                        =
                        \cos\left(\frac{\theta}{4}\right)\ket{0000}
                        +
                        \sin\left(\frac{\theta}{4}\right)\ket{1111}\,.
                    \end{equation}
                    In the Jordan-Wigner representation, for each qubit $q$ in the block one has
                    \begin{equation}
                        C_{2q-1}C_{2q} = i Z_q\,,
                    \end{equation}
                    and therefore
                    \begin{equation}
                        \bra{\phi_\theta} C_{2q-1}C_{2q}\ket{\phi_\theta}
                        =
                        i\bra{\phi_\theta} Z_q \ket{\phi_\theta}
                        =
                        i\left(
                        \cos^2\left(\frac{\theta}{4}\right)
                        -
                        \sin^2\left(\frac{\theta}{4}\right)
                        \right)
                        =
                        i\cos\left(\frac{\theta}{2}\right)\,.
                    \end{equation}
                    On the other hand, if the two Majorana operators act on different qubits, then the corresponding Pauli string contains two $X/Y$ terms and therefore it maps both $\ket{0000}$ and $\ket{1111}$ to computational basis states that are orthogonal to each other. Thus, all such matrix elements vanish. Since the full state $\ket{\psi_\theta}$ is a tensor product of identical four-qubit blocks, the same argument also implies that all mixed two-point functions between Majoranas belonging to different blocks vanish. It follows that the full two-Majorana coefficient matrix is of the form
                    \begin{equation}
                        \tstatemat
                        =
                        \frac{\cos\left(\frac{\theta}{2}\right)}{\sqrt{2d}}\,J\,,
                    \end{equation}
                    where $J$ is the standard symplectic form on $2n$ Majorana modes as defined in Eq.~\eqref{supp:eq:Gaussian_state_Y_matrix_form}. Consequently,
                    \begin{equation}\label{eq:magic_state_pbb}
                        \tstatebb
                        =
                        \tstatemat^\top \tstatemat
                        =
                        \frac{\cos^2\left(\frac{\theta}{2}\right)}{2d}\,\id
                        =
                        \frac{\cos^2\left(\frac{\theta}{2}\right)}{d\cdot 2!}\,\id\,.
                    \end{equation}

                    Therefore, for all $\theta$ such that $\cos(\theta/2)\neq 0$, the eigenvalues of $\tstatebb$ are all equal and we obtain
                    \begin{equation}
                        \frac
                        {\lVert {\tstatebb} \rVert_2^2} 
                        {\lVert {\tstatebb} \rVert_1^2}
                        =
                        \frac
                        {\lVert \id \rVert_2^2} 
                        {\lVert \id \rVert_1^2}
                        =
                        \frac
                        {2n} 
                        {\left(2n\right)^2}
                        = 
                        \frac{1}{2n}
                        \in \OC\left(\frac{1}{n}\right) \subset o(1)
                        \,.
                    \end{equation}
                    Since further, for any observable $\O \in \mathcal{B}_2$, the corresponding norm fraction ${\lVert \tobsbb \rVert_2^2}/{\lVert \tobsbb \rVert_1^2} \leq 1$, Condition~\ref{cond:gaussian_dominate} is fulfilled. The remaining Condition~\ref{cond:full_support} is trivially fulfilled by the assumption of an arbitrary observable $\O$ in the module $\mathcal{B}_2$. Therefore, Theorem \ref{thm:supp:GP} holds also for the family of product states $\rho_\theta$, for all $\theta$ such that $\cos(\theta/2)\neq 0$, and Haar-random matchgate unitaries induce univariate \acp{gp} from these states and any observable $\O \in \mathcal{B}_2$.

                    In Supp.~Fig.~\ref{fig:magic} we show the empirical GP distribution for magic states $\rho_\theta$ for different values of $\theta$ and a random Pauli measurement in $\BC_2$, together with the expected Gaussian curve, which using Eq.~\eqref{eq:two_maj_cov_exact} and Eq.~\eqref{eq:magic_state_pbb} can be readily found to be $\NC(0,|\cos(\theta/2)|/\sqrt{2n-1})$.

                    \begin{figure}
                        \centering
                        \includegraphics[width=0.6667\linewidth]{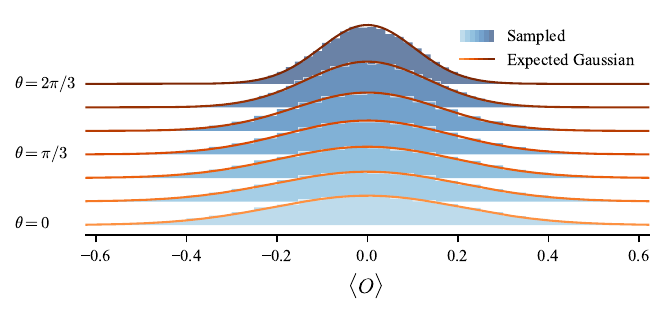}
                        \caption{\textbf{Numerical verification of a univariate GP for magic states and Pauli observable in $\BC_2$.} We consider a system of $n=12$ qubits. For each of the seven equispaced values of $\theta\in[0,2\pi/3]$, we take $5\cdot10^4$ samples of the random variable given by the expectation value of $O=Z_1$ over the state $\sigma_\theta= U\rho_\theta U^\dagger$, with $U\sim \SPIN(2n)$. The blue histograms show the numerical distribution of the samples, while the orange curves represent the expected Gaussian distribution $\NC(0,|\cos(\theta/2)|/\sqrt{2n-1})$. Different shades of the colors correspond to different values of $\theta$, increasing from lighter to darker. Note that while a $\BC_2$ Pauli observable alone is insufficient to induce a GP, the magic state $\rho_\theta$ successfully drives the asymptotic convergence.}
                        \label{fig:magic}
                    \end{figure}

            \subsection{Special GP families in \texorpdfstring{$\BC_4$}{B4}}

                In this section, we analyze various families of states and observables for the module $\BC_4$ with respect to their suitability for forming GPs. A particular focus is on Pauli observables, which alone cannot lead to GPs in $\BC_4$. More interestingly, we identify that, when combined with Gaussian states, they can never lead to GPs, while providing numerical evidence that a class of random states can help to achieve GPs with $\BC_4$ Pauli observables.

                \subsubsection{Gaussian states alone cannot lead to GPs in \texorpdfstring{$\BC_4$}{B4}}

                    The following analysis of the behavior of Gaussian states when $\O \in \BC_4$ relies on a particular subset of ratios in Eq.~\eqref{eq:cond:univariate:imp} of the equivalent univariate formulation of Proposition~\ref{prop:cond:univariate} that are analogous to the simplified $\BC_2$ version of Condition~\ref{cond:gaussian_dominate} in Eq.~\eqref{supp:eq:two_majo_uni_simplified_contraction_condition} of Proposition~\ref{supp:cor:two_majo_simplified_contraction_condition}. Specifically, these are ratios that correspond to Schatten norm comparisons. Since these ratios naturally occur as a subset of the full condition for $\BC_4$ (Proposition~\ref{prop:cond:univariate}), we prove that Condition~\ref{cond:gaussian_dominate} cannot be fulfilled, provided that the ratio associated with the observable does not vanish sufficiently fast.
                    Particularly, we study the cases of Proposition~\ref{prop:cond:univariate} (simplified $\BC_2$ condition) where the contraction over four replications of the (rank-4) frame coefficient tensor $\tstate$ of a Gaussian state arranges them on a cycle where each tensor is connected to two other tensors via two (index) legs each. As in Proposition~\ref{supp:cor:two_majo_simplified_contraction_condition}, the state coefficient tensor $\tstate$ can then be reshaped into square matrices $\tstatemat$ and their contractions (magnitudes) are expressed as the (absolute) trace of matrix powers $\left\lvert \tr\left[\tstatemat^4\right] \right\rvert$ in the numerator and  $\left\lvert \tr\left[\tstatemat^2\right] \right\rvert$ in the denominator. Furthermore, the relation to Schatten 2- and 1-norms then still holds as $\left\lvert \tr\left[\tstatemat^4\right] \right\rvert = \left\lVert {\tstatebb} \right\rVert^2_2$ and $\left\lvert \tr\left[\tstatemat^2\right] \right\rvert = \left\lVert {\tstatebb} \right\rVert_1$, respectively. Any Gaussian state with non-zero support in $\BC_4$ (which implies even parity) can be prepared by a matchgate transformation applied to the all-zero state $\ket{0}^{\otimes n}$, which is related to a special orthogonal transformation of the all-zero state $\BC_4$ frame coefficient matrix $\tstatemat_0$. As the traces/norms are invariant under such an orthogonal basis change, the ratio for any Gaussian state matches the ratio associated with the all-zero state, i.e., 
                    $
                        \left\lVert {\tstatebb} \right\rVert_2/\left\lVert {\tstatebb} \right\rVert_1 = \left\lVert {\tstatebb}_0 \right\rVert_2/\left\lVert {\tstatebb}_0 \right\rVert_1
                    $
                    with ${\tstatebb}_0 = \tstatemat^{\top}_0 \tstatemat_0$.
                    Therefore, studying the asymptotic behavior of this ratio for the all-zero state allows for drawing conclusions about all Gaussian states. Hence, we provide a partial violation of Condition~\ref{cond:gaussian_dominate} for any Gaussian state in the $\BC_4$ case. As long as this scaling is not counteracted by the ratio with respect to the observable $\O \in \BC_4$ to yield an asymptotic vanishing as required by Condition~\ref{cond:gaussian_dominate}, a Gaussian process cannot be formed as per Theorem~\ref{thm:supp:no_GP}.
                    
                    The following Theorem shows that when Gaussian states are considered, the ratio $\frac{\|{\tstatebb}_0\|_2}{\|{\tstatebb}_0\|_1}$ converges to a constant in the large-$n$ limit. 
                    \begin{proposition}\label{thm:gaussian_ratio_b4}
                        The ratio of the Schatten 2 and 1-norms of the squared, reshaped tensor of coefficients $\tstatemat_0$ of the reference Gaussian state $\ket{0}^{\otimes n}$, i.e., ${\tstatebb}_0 = \tstatemat^{\top}_0 \tstatemat_0$ is
                        \begin{equation}
                            \left(\frac{\|{\tstatebb}_0\|_2}{\|{\tstatebb}_0\|_1}\right)^2
                            =
                            \frac{n^2-3n+5}{9n(n-1)}\,.
                        \end{equation}
                        In the large-$n$ limit, it hence holds
                        \begin{equation}
                            \left(\frac{\|{\tstatebb}_0\|_2}{\|{\tstatebb}_0\|_1}\right)^2 \longrightarrow \frac{1}{9}\,.
                        \end{equation}
                    \end{proposition}
                        
                    \begin{proof}
                        By direct computation, one readily finds
                        \begin{equation}
                            [\tstatemat_0]_{(i j),(k l)} = J_{ij}J_{kl} - J_{ik}J_{jl} + J_{il}J_{jk}\,,
                        \end{equation}
                        from elements of the symplectic form $J \in \mathbb{R}^{2n \times 2n}$ as defined in Eq.~\eqref{supp:eq:Gaussian_state_Y_matrix_form}.
                        By grouping arbitrary index pairs $(i,j)$ and $(k,l)$, we have constructed $\tstatemat_0$ as a matrix of size $(2n)^2\times(2n)^2$. However, since $[\tstatemat_0]_{(ij),(kl)}$ is anti-symmetric under $i\leftrightarrow j$ and under $k\leftrightarrow l$, it annihilates the symmetric sector of $(\mathbb{R}^{2n})^2$ and acts non-trivially only on the anti-symmetric one. Therefore, without loss of generality, we can restrict $\tstatemat_0$ to the space of $2$-forms\footnote{We recall that in general, a $k$-form is a completely anti-symmetric $k$-linear map acting on $k$ copies of a vector space.} $\Lambda^2 \mathbb{R}^{2n}$, whose dimension is
                        \begin{equation}
                            \dim(\Lambda^2 \mathbb{R}^{2n})=\binom{2n}{2}=n(2n-1)\,.
                        \end{equation}
                        Then, by the symmetry of $\tstatemat_0$
                        \begin{equation}
                            {\tstatebb}_0 = \tstatemat_0^2\,.
                        \end{equation}
                        The elements $A, B\in \Lambda^2 \mathbb{R}^{2n}$ can be identified with anti-symmetric matrices $A=[A_{ij}]$ and $B=[B_{ij}]$, and we can equip them with the inner product
                        \begin{equation}
                            \langle A,B\rangle = \sum_{i<j} A_{ij}B_{ij}\,.
                        \end{equation}
                        The symplectic identity $J$ is also an element of $\Lambda^2 \mathbb{R}^{2n}$, and its norm induced by the previous inner product reads
                        \begin{equation}
                            \langle J,J\rangle = n\,.
                        \end{equation}

                        Now, we want to characterize the spectrum of $\tstatemat_0$, and hence of ${\tstatebb}_0$, so that we compute the needed Schatten norms.
                        To do so, let us understand how $\tstatemat_0$ acts on an element $A$ of $\Lambda^2 \mathbb{R}^{2n}$. Expanding the left product, we have
                        \begin{equation}
                            [\tstatemat_0 A]_{ij}
                            =
                            \sum_{k<l} [\tstatemat_0]_{(ij),(kl)} A_{kl}
                            =
                            \sum_{k<l} \left(J_{ij}J_{kl} - J_{ik}J_{jl} + J_{il}J_{jk}\right) A_{kl}\,.
                        \end{equation}
                        The first term is immediately
                        \begin{equation}
                            \sum_{k<l} J_{ij}J_{kl}A_{kl}
                            =
                            J_{ij}\langle J,A\rangle\,.
                        \end{equation}
                        For the remaining two terms, using anti-symmetry of $A$ and $J$, one can check that
                        \begin{equation}
                            \sum_{k<l} \bigl(-J_{ik}J_{jl}+J_{il}J_{jk}\bigr)A_{kl}
                            =
                            -[JAJ^\top]_{ij}\,.
                        \end{equation}
                        Hence
                        \begin{equation}
                            \tstatemat_0 A = \langle J,A\rangle J - T(A)
                            \quad \text{with} \quad
                            T(A)=JAJ^\top\,.
                        \end{equation}
                        The operator $T$ is an involution, i.e., it is its own inverse. This can be readily seen from the fact that $J^\top=-J$ and $J^2=-\id$. This implies that $T$ has eigenvalues $\pm 1$.
                        Now, we can decompose $\Lambda^2 \mathbb{R}^{2n}$ into the eigenspaces of $T$. Using again the properties of $J$ we can immediately see that $J$ is an eigenvector of eigenvalue $1$ of $T$. Hence the space $\mathbb{R}J$ is $T$-invariant.
                        Let us then decompose $\Lambda^2 \mathbb{R}^{2n}$ as 
                        \begin{equation}
                            \Lambda^2 \mathbb{R}^{2n} = \mathbb{R}J \oplus W_+ \oplus W_-\,,
                        \end{equation}
                        where $\mathbb{R}J \oplus W_+$ is the eigenspace of eigenvalue $+1$ of $T$, with $W_+ \perp J$, and $W_-$ is the eigenspace of eigenvalue $-1$. We can notice that the positive eigenspace of $T$ is the space of anti-symmetric matrices $A$ that commute with $J$, since $JAJ^\top=A\to JA=AJ$ using $J^2=-I$. Similarly, the negative eigenspace is given by the anti-symmetric matrices that anticommute with $J$. We can use this to compute the dimensions of $W_+$ and $W_-$. Indeed, write $A$ as
                        \begin{equation}
                            A = \begin{pmatrix}
                                P & Q \\
                                -Q^\top & S
                            \end{pmatrix}\,,
                        \end{equation}
                        for $P,S$ anti-symmetric $n\times n$ real matrices and $Q$ an $n\times n$ real matrix. 
                        Notice that in this block decomposition 
                        \begin{equation}
                            J=\begin{pmatrix}
                                0 & \openone \\
                                -\openone & 0
                            \end{pmatrix}\,.
                        \end{equation}
                        If we now impose $[J,A]=0$ we get the condition
                        \begin{equation}
                            JA = \begin{pmatrix}
                                -Q^\top & S \\
                                -P & -Q
                            \end{pmatrix}
                            =
                            \begin{pmatrix}
                                -Q & P \\
                                -S & -Q^\top
                            \end{pmatrix}
                            = AJ\,,
                        \end{equation}
                        which implies $P=S$ and $Q=Q^\top$. Hence the elements of the positive eigenspace of $J$ look like
                        \begin{equation}
                            A = \begin{pmatrix}
                                P & Q \\
                                -Q & P
                            \end{pmatrix}\,,
                        \end{equation}
                        for $P$ anti-symmetric and $Q$ symmetric. Since the space of $n\times n$ anti-symmetric matrices has dimension $\dim(P)=\frac{n(n-1)}{2}$, while that of symmetric ones is $\dim(Q)=\frac{n(n+1)}{2}$, we readily get
                        \begin{equation}
                            \dim(W_+)+1=\frac{n(n-1)}{2} + \frac{n(n+1)}{2} = n^2\,,
                        \end{equation}
                        where we remember that we already separated the positive eigenspace of $J$ as $\mathbb{R}J \oplus W_+$.
                        Since the space $\Lambda^2 \mathbb{R}^{2n}$ has dimension $\dim(\Lambda^2 \mathbb{R}^{2n})=n(2n-1)$, we get
                        \begin{equation}
                            \dim(W_-)=n(n-1)\,.
                        \end{equation}
                        
                        Now let us study the action of $\tstatemat_0$ on each of $\mathbb{R}J$, $W_+$ and $W_-$.
                        First, on $J$ we quickly see
                        \begin{equation}
                            \tstatemat_0 J=\langle J,J\rangle J - T(J)= nJ - J = (n-1)J\,.
                        \end{equation}
                        Next, let $A\in W_+$. Then, since $W_+\perp J$, $\langle J,A\rangle=0$ and $T(A)=A$, hence
                        \begin{equation}
                            \tstatemat_0 A = 0\cdot J - A = -A\,.
                        \end{equation}
                        Lastly, let $A\in W_-$. Then we again have $\langle J,A\rangle=0$ and now $T(A)=-A$, so
                        \begin{equation}
                            \tstatemat_0 A = 0\cdot J - (-A) = A\,.
                        \end{equation}
                        
                        From this, we get that the spectrum of $\tstatemat_0$ is
                        \begin{equation}
                            \operatorname{spec}(\tstatemat_0)
                            =
                            \{n-1\}^{[1]}
                            \cup
                            \{1\}^{[n(n-1)]}
                            \cup
                            \{-1\}^{[n^2-1]}\,.
                        \end{equation}
                        Squaring, we obtain the spectrum of ${\tstatebb}_0=\tstatemat_0^2$
                        \begin{equation}
                            \operatorname{spec}({\tstatebb}_0)
                            =
                            \{(n-1)^2\}^{[1]}
                            \cup
                            \{1\}^{[n(n-1)]}
                            \cup
                            \{1\}^{[n^2-1]}\,.
                        \end{equation}
                        Since ${\tstatebb}_0$ is positive semidefinite, its Schatten $1$-norm and Schatten $2$-norm satisfy
                        \begin{equation}
                            \|{\tstatebb}_0\|_1 = \Tr[{\tstatebb}_0]\,,
                            \qquad
                            \|{\tstatebb}_0\|_2^2 = \Tr[{\tstatebb}_0^2]\,.
                        \end{equation}
                        Using the spectrum above,
                        \begin{equation}
                            \|{\tstatebb}_0\|_1
                            =
                            (n-1)^2 + \bigl(n(2n-1)-1\bigr)
                            =
                            3n(n-1)\,,
                        \end{equation}
                        and
                        \begin{equation}
                            \|{\tstatebb}_0\|_2^2
                            =
                            (n-1)^4 + \bigl(n(2n-1)-1\bigr)\,.
                        \end{equation}
                        Putting all together and simplifying we arrive at
                        \begin{equation}
                            \left(\frac{\|{\tstatebb}_0\|_2}{\|{\tstatebb}_0\|_1}\right)^2
                            =
                            \frac{n^2-3n+5}{9n(n-1)}\,.
                        \end{equation}
                        Lastly, taking $n\to\infty$ gives
                        \begin{equation}
                            \left(\frac{\|{\tstatebb}_0\|_2}{\|{\tstatebb}_0\|_1}\right)^2 \longrightarrow \frac{1}{9}\,,
                        \end{equation}
                        which concludes the proof.
                    \end{proof}

                \subsubsection{Pauli observables (in \texorpdfstring{$\BC_4$}{B4}) alone cannot lead to GPs}

                    Analogously to the case of Gaussian states, we here show that Pauli observables in $\BC_4$ cannot fulfill the equivalent univariate formulation of Condition~\ref{cond:gaussian_dominate} in Proposition~\ref{prop:cond:univariate}. Hence, this partial violation of Condition~\ref{cond:gaussian_dominate} can only be resolved by the state choice and either leads to GPs or prevents their formation as per Theorems~\ref{thm:supp:GP} and~\ref{thm:supp:no_GP}, respectively. Specifically, we show that some contractions yield ratios, namely the ones that form Schatten 2 vs 1-norm ratios under the observable coefficient tensor $\tobs$ being reshaped into a matrix $\tobsmat$, that do not vanish asymptotically. The following theorem proves that these kinds of ratios are constant. 

                    \begin{proposition}\label{thm:pauli_ratio_b4}
                        Given a Pauli observable $\hat{O}=-C_aC_bC_cC_d\in\BC_4$, let $\tobs$ be the completely anti-symmetric tensor of its frame coefficients over $\BC_4$, and $\tobsmat = \tobs_{(ij),(kl)}$ the corresponding symmetric matrix obtained by grouping the first two indices of $\tobs$ as the row index and the last two as the column one.
                        Then, for $\tobsbb=\tobsmat^2$ we have
                        \begin{equation}
                            \left(\frac{\|\tobsbb\|_2}{\|\tobsbb\|_1}\right)^2=\frac{1}{6}\,.
                        \end{equation}
                    \end{proposition}
                    
                    \begin{proof}
                        Let us denote by $\{e_1,\dots,e_{2n}\}$ the canonical basis of $\mathbb{R}^{2n}$. For any four distinct indices $a<b<c<d\in[2n]$, one has
                        \begin{equation}
                            \tobs = e_a\wedge e_b\wedge e_c\wedge e_d\,.
                        \end{equation}
                        Indeed, by definition the exterior product $\wedge$ returns the completely anti-symmetric combination of the four delta-like tensors associated to each Majorana index.
                        Then, the matrix entries of $\tobsmat$ read as
                        \begin{equation}
                            \tobselem_{(ij),(kl)} = [\tobs]_{ijkl}\,.
                        \end{equation}
                        Notice that, grouping arbitrary pairs $(i,j)$ and $(k,l)$, then $\tobsmat$ is a matrix of size $(2n)^2\times(2n)^2$. However, since $\tobs$ is anti-symmetric under $i\leftrightarrow j$ and under $k\leftrightarrow l$, the corresponding operator vanishes on the symmetric sector of $(\mathbb{R}^{2n})^2$ and acts non-trivially only on the anti-symmetric one. Therefore, without loss of generality, we may restrict $\tobsmat$ to the space of $2$-forms $\Lambda^2(\mathbb{R}^{2n})$, whose natural basis is given by the vectors
                        \begin{equation}
                            E_{ij}=e_i\wedge e_j\,,\qquad i<j\,.
                        \end{equation}
                        Now, we can make the components of $\tobsmat$ explicit by introducing the generalized (anti-symmetric) Kronecker delta
                        \begin{equation}
                            \delta_{ij}^{pq} = \delta_i^p\delta_j^q-\delta_i^q\delta_j^p = \left(e_i\wedge e_j\right)_{p,q} \,.
                        \end{equation}
                        From here, we have
                        \begin{equation}
                            \tobselem_{(ij),(kl)}
                            =
                            \delta_{ij}^{ab}\delta_{kl}^{cd}
                            -\delta_{ij}^{ac}\delta_{kl}^{bd}
                            +\delta_{ij}^{ad}\delta_{kl}^{bc}
                            +\delta_{ij}^{cd}\delta_{kl}^{ab}
                            -\delta_{ij}^{bd}\delta_{kl}^{ac}
                            +\delta_{ij}^{bc}\delta_{kl}^{ad}.
                        \end{equation}
                        Equivalently, introducing the six basis vectors
                        \begin{equation}
                            u_1 = E_{ab}\,,\qquad
                            u_2 = E_{ac}\,,\qquad
                            u_3 = E_{ad}\,,\qquad
                            u_4 = E_{bc}\,,\qquad
                            u_5 = E_{bd}\,,\qquad
                            u_6 = E_{cd}\,,
                        \end{equation}
                        one has
                        \begin{equation}
                            \tobsmat
                            =
                            u_1u_6^\top+u_6u_1^\top
                            -u_2u_5^\top-u_5u_2^\top
                            +u_3u_4^\top+u_4u_3^\top.
                        \end{equation}
                        Here, the $E_{ij}$ are orthonormal with respect to the standard inner product on 2-forms induced by the standard inner product on $\mathbb{R}^{2n}$, namely
                        \begin{equation}
                            \langle u\wedge v,\; x\wedge y\rangle
                            =
                            \langle u,x\rangle \langle v,y\rangle
                            -
                            \langle u,y\rangle \langle v,x\rangle\,.
                        \end{equation}
                        
                        Now, since $u_1,\dots,u_6$ are orthonormal, we can directly compute
                        \begin{equation}
                            \tobsbb = \tobsmat^2
                            =
                            (u_1u_6^\top+u_6u_1^\top)^2
                            +(-u_2u_5^\top-u_5u_2^\top)^2
                            +(u_3u_4^\top+u_4u_3^\top)^2\,.
                        \end{equation}
                        Each square is the projector onto its $2$-dimensional support
                        \begin{align}
                            (u_1u_6^\top+u_6u_1^\top)^2 &= u_1u_1^\top+u_6u_6^\top\,, \\
                            (-u_2u_5^\top-u_5u_2^\top)^2 &= u_2u_2^\top+u_5u_5^\top\,, \\
                            (u_3u_4^\top+u_4u_3^\top)^2 &= u_3u_3^\top+u_4u_4^\top \,.
                        \end{align}
                        Namely
                        \begin{equation}
                            \tobsbb=\sum_{r=1}^6 u_ru_r^\top \,.
                        \end{equation}
                        Thus, $\tobsbb$ is the orthogonal projector onto the $6$-dimensional subspace spanned by $E_{ab}, E_{ac},E_{ad},E_{bc},E_{bd},E_{cd}$. This readily implies that its spectrum is
                        \begin{equation}
                            \operatorname{spec}(\tobsbb)=\{1\}^{[6]}\cup \{0\}^{[\binom{2n}{2}-6]}.
                        \end{equation}
                        Hence
                        \begin{equation}
                        \|\tobsbb\|_1=\operatorname{Tr}(\tobsbb)=6,
                        \qquad
                        \|\tobsbb\|_2^2=\operatorname{Tr}(\tobsbb^2)=6,
                        \end{equation}
                        and so
                        \begin{equation}
                        \left(\frac{\|\tobsbb\|_2}{\|\tobsbb\|_1}\right)^2
                        =
                        \frac{6}{6^2}
                        =
                        \frac{1}{6}\,,
                        \end{equation}
                        completing the proof.
                    \end{proof}

            \subsubsection{Pauli observables (in \texorpdfstring{$\BC_4$}{B4}) and Gaussian states do not form GPs, but random states may help}

                    As Theorems~\ref{thm:gaussian_ratio_b4} and~\ref{thm:pauli_ratio_b4} in the previous two sections identify, neither Gaussian states nor Pauli observables in $\BC_4$ alone can lead to GPs. Therefore, combining both definitely fails to meet Condition~\ref{cond:gaussian_dominate}. This violation implies, as per Theorem~\ref{thm:supp:no_GP}, that Gaussian states combined with Pauli observables in $\BC_4$ cannot lead to GPs. Although Theorems~\ref{thm:gaussian_ratio_b4} and~\ref{thm:pauli_ratio_b4} explicitly demonstrate this violation for the univariate case, it naturally follows that multivariate GPs are likewise impossible. Because the set of multivariate moments strictly contains all univariate moments, the univariate failure inherently precludes any multivariate generalization. Supplemental Fig.~\ref{fig:rand_state_single_Pauli_gp_b4}b illustrates this violation from a $12$-qubit simulation.

                    Theorem~\ref{thm:pauli_ratio_b4} implies that if the observable $O$ is a Pauli operator in $\BC_4$, the emergence of a GP can only be obtained through the state choice, which Gaussian states fail to achieve. Therefore, GPs from a Pauli observable in $\BC_4$ require strong conditions on the family of states. While analytically characterizing such states is generally challenging, we provide numerical evidence in Supp.~Fig.~\ref{fig:rand_state_single_Pauli_gp_b4}a. Specifically, we show numerically that a random fermionic state, i.e., a random positive eigenstate of the parity operator $Z^{\otimes n}$, could yield a univariate GP. 
                    It is worth noting that Proposition~\ref{thm:special_obs_odd_moments} and the subsequent remark rigorously establish that all odd moments are zero for any single Pauli observable in $\BC_4$. 
                    
                    \begin{figure}[htbp]
                        \centering
                        \includegraphics[width=1\linewidth]{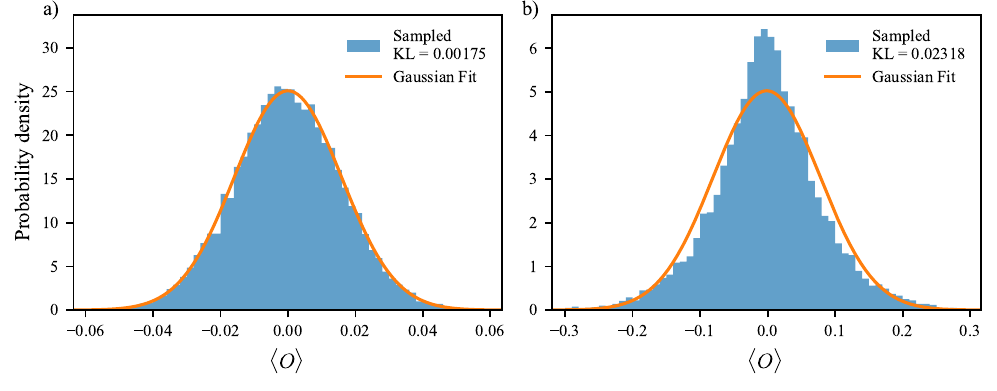}
                        \caption{\textbf{A plausible univariate GP vs a provable violation for a Pauli observable in $\BC_4$.} We compare the empirical statistics of the expectation values of a random Pauli drawn from $\BC_4$ when: a) the input state is a random fermionic state, b) the input state is Gaussian. In both cases, the state is then evolved under a Haar-random matchgate (free-fermionic) transformation before being measured.
                        The system size considered in this analysis is $n=12$, and the histograms were generated from $10^4$ samples for both panels. On top of the blue histograms, we plot the fitted Gaussian curve (in orange) obtained from the empirical mean and variance of the samples. Furthermore, in the legend we report the Kullback-Leibler divergence (KL) between the fitted Gaussian distribution and the samples' one. Both the KL values and the visual agreement between the histogram and the fitted curve suggest that random fermionic states could yield a GP, while Gaussian states do not.}
                        \label{fig:rand_state_single_Pauli_gp_b4}
                    \end{figure}

            \subsection{GPs from special observables in any module\texorpdfstring{ $\BC_m$}{}} \label{sec:special_obs}
            
            We here consider a special class of observables $\O$ that can be constructed in any module $\BC_m$, which provably leads to moments required to obey GPs. These observables form linear combinations of $\BC_m$ basis operators associated with mutually disjoint Majorana indices, which is formally defined in the following. This observable is based on a (partial) partition of the $2n$ Majorana operator indices into mutually disjoint index sets (of size $m$), i.e.,
            \begin{equation}
                \bm{\VC} = \left\lbrace \VC^{(j)} \subseteq [2n] \mid j \le 1 \le \lfloor 2n/m \rfloor, \, \lvert\VC^{(j)}\rvert = m \right\rbrace
                \quad \text{such that} \quad
                \forall j \neq j'\colon \VC^{(j)} \cap \VC^{(j')} = \varnothing \,.
            \end{equation}
            Hence, each of these index subsets $\VC^{(j)}$ labels one module basis operator $C_{\VC^{(j)}}$ of the module $\BC_m$ such that these operators, which are products of $m$ distinct Majorana operators, never share the same Majorana operator. 
            Clearly, $|\bm{\VC}|=\lfloor 2n/m \rfloor$. Then, the partition $\bm{\VC}$ completely covers all $2n$ Majorana indices if $2n$ is a multiple of $m$, and, otherwise, is a partial partition, which leaves $2n \bmod m$ indices uncovered.
            
            Therefore, the form of such module partition observables $O$ reads as 
            \begin{equation}\label{eq:special_obs}
                \O = \sum_{j=1}^{\lfloor 2n/m \rfloor} \omega_{j} C_{\VC^{(j)}} \,,
            \end{equation}
            where $\omega_j$ are non-zero real coefficients that we assume to be constant $\Theta(1)$ in the number of qubits $n$.
            Crucially, this is in contrast to summing over all possible indices and corresponding module basis operators $C_{\VC}$, which would lead to any element of $\BC_m$.
          
            These observables have the nice property that their associated tensors $\tobs$ can be expressed as a sum of mutually orthogonal generalized (anti-symmetric) Kronecker deltas
            \begin{equation}\label{eq:special_obs_tensor}
                \tobselem_{\alpha_1, \ldots, \alpha_m} =\Tr\left[\hat{O}C_{\alpha_1, \ldots, \alpha_m}\right] = 
                \sum_{j=1}^{\lfloor 2n/m \rfloor} \omega_j \delta^{(j)}_{\alpha_1\dots \alpha_m}\,.
            \end{equation}
            Here, the index $j$ associates the Kronecker delta with the index subset $\VC^{(i)}$. Let $\nu_1 < \dots < \nu_m$ be the strictly ordered indices of $\VC^{(j)}$, we can then define $\delta^{(j)}$ via the standard notation for the generalized Kronecker delta as
            \begin{equation}
                \delta^{(j)}_{\alpha_1\dots \alpha_m} = \delta^{\nu_1 \dots \nu_m}_{\alpha_1\dots \alpha_m} \,,
            \end{equation}
            and further defined in Eq.~\eqref{eq:generalized_antisym_kronecker},
            which automatically enforces total anti-symmetry and evaluates to zero unless the lower indices $\alpha_1 \dots \alpha_m$ are a permutation of the subset $\VC^{(j)}$. 
            
            By the generalized Kronecker deltas being \emph{mutually orthogonal} we mean that any possible contraction involving multiple tensors $\delta^{(j)}$ vanishes identically unless every tensor shares the exact same subset index $j$. This follows directly from the fact that the subsets $\VC^{(j)}$ are mutually disjoint ($\VC^{(j)} \cap \VC^{(j')} = \varnothing$ for $j \neq j'$). Because a component of $\delta^{(j)}$ is non-zero only when its indices are exclusively from $\VC^{(j)}$, any contraction between tensors with different subset indices forces a mismatch, evaluating to zero.
            Consequently, when evaluating the tensor contractions required for the moments according to Eq.~\eqref{eq:supp:moment_via_Brauer_moment_op}, there are no cross-terms and the contributions from each index subset $\VC^{(j)}$ decouple completely. Intuitively, if one simply orders the indices so that the elements of each subset $\VC^{(j)}$ are grouped together, the observable tensor $\tobs$ reveals a strictly block-diagonal structure consisting of isolated blocks of size $m$.
            
            With these useful properties established, we now study the odd and even moments.

            We prove that for the class of observables here considered, the odd moments always vanish, non-trivially because any possible contraction of an odd number $k$ of module frame coefficient tensors $\tobs$ vanishes:
            \begin{proposition}[Odd moments for special observable]\label{thm:special_obs_odd_moments}
                For the special case of observables $O$ defined in Eq.~\eqref{eq:special_obs},
                all odd moments vanish.
            \end{proposition}
            \begin{proof} 
                As per Lemma~\ref{lemma:vanish_vs_moment_op}, the odd moments trivially vanish for odd $m$ (also cf. Lemma~\ref{prop:supp:vanish_odd_moments}), and, hence, the main part of the proof focuses on modules $\BC_m$ with even $m$ and shows that vanishing odd $k$ moments, expressed via contractions determined by the Brauer algebra $\mathfrak{B}_{mk/2}$.
                
                Recall that if two tensors $\delta^{(j)}_{\alpha_1\dots \alpha_m},\delta^{(j')}_{\beta_1\dots \beta_m},$ with $j\neq j'$ are involved in a contraction the latter will vanish. Hence we can assume without loss of generality that $\tobselem_{\alpha_1\dots \alpha_m}=\omega_j\delta^{(j)}_{\alpha_1\dots \alpha_m}$ for one choice of index $j$ and, hence, index subset $\VC_j$.
                Now, effectively the indices $\alpha_1,\ldots,\alpha_m$ can be restricted to the $m$ possible values $\{\nu_1,\dots,\nu_m\}$, since any other value in $[2n]$ yields zero.
                Hence, we recover arbitrary contractions of $k$ copies of $\delta_{\alpha_1\dots \alpha_m}^{\nu_1\dots\nu_m}$, which corresponds to matching $mk/2$ indices and summing over all their possible values.
                However, the complete anti-symmetry of these tensors implies that only the combinations where the indices of every single tensor take different values survive.
                We can think of this as an instance of an edge-coloring graph problem, for an $m$-regular multigraph of $k$ vertices/nodes. Indeed, the $m$ possible values that the indices $\alpha_1,\ldots,\alpha_m$ can take can be associated with $m$ possible colors, while the edges of the multigraph are provided by the contraction at hand, i.e., the Brauer algebra element, and the vertices are the tensors.
                
                We highlight that if it were possible to edge-color an $m$-regular graph with exactly $m$ colors, each color would have to visit every vertex exactly once. In terms of the tensors, this means any specific index value (e.g., $\nu_1$) must appear exactly once within each of the $k$ copies of $\delta^{(j)}$.
                Consequently, this specific index value must appear a total of $k$ times across the entire expression. However, every contraction inherently links exactly two indices together. Since $k$ is odd, it is algebraically impossible to group $k$ occurrences into pairs. This contradiction dictates that no valid index assignment can simultaneously satisfy the contraction lines and distinct index values for the Kronecker delta tensors, meaning every term evaluates to zero, completing the proof.
            \end{proof}

            Notice that for a general observable $O$, the graph coloring argument for its tensor $\tobs$ corresponds to $2n$ available colors rather than $m$. This implies that the vanishing of odd moments cannot be deduced directly from anti-symmetry properties alone, but must instead stem from explicit cancellations between the tensor components.
            
            Furthermore, note that the proof for the vanishing of odd moments relies neither on the non-zero requirement of the weights in Eq.~\eqref{eq:special_obs}, nor on the large-$n$ limit or any associated scaling arguments. Instead, it strictly requires only that the observable is a linear combination of basis operators $C_{\VC^{(j)}}$ supported on mutually disjoint Majorana index subsets $\VC^{(j)} \subseteq [2n]$.

            All that is left to show is that the even moments coming from this class of observables converge to those of a Gaussian in the large-$n$ limit, that is, that Condition~\ref{cond:gaussian_dominate} is satisfied. 

            \begin{proposition}[Even moments for special observable]\label{thm:special_obs_even_moments}
                For the special class of observables $\O$ defined in Eq.~\eqref{eq:special_obs}
                \begin{equation}\label{eq:special_obs_even_moments}
                    \frac{\operatorname{contr}_{\sigma}\!\left( \tobs^{\otimes k} \right)} {\lVert O \rVert_{\BC_m}^k 
                    }\in 
                    \OC\left(\frac{1}{n}\right)\,,
                \end{equation}
                for any even $k$ and contraction $\sigma \in \mathfrak{B}_{mk/2}$ that does not correspond to pairwise matchings of $k$ tensor copies $\tobs$.
            \end{proposition}
            \begin{proof}
                Let us start by computing the result of a full contraction between two copies of the tensor $\tobs$ governed by an arbitrary index permutation $\pi \in S_m$, denoted $\operatorname{contr}_{\pi}\!\left(\tobs^{\otimes 2}\right)$. This reads 
                \begin{equation}
                    \operatorname{contr}_{\pi}\!\left(\tobs^{\otimes 2}\right) = \sum_{j, j'}\omega_j \omega_{j'}\delta^{(j)}_{\alpha_1\dots \alpha_m}\delta^{(j')}_{\alpha_{\pi(1)}\dots \alpha_{\pi(m)}}\,,
                \end{equation}
                where Einstein summation over the indices $\{\alpha_1\dots \alpha_m\}$ is intended. The magnitude of these contractions arises as factors $\lVert O \rVert_{\BC_m}^2$ in the denominator in Eq.~\eqref{eq:special_obs_even_moments}. Using the mutual orthogonality of the anti-symmetric Kronecker delta tensors for $j\neq j'$, and their total anti-symmetry which yields $\delta^{(j)}_{\alpha_{\pi(1)}\dots \alpha_{\pi(m)}} = \operatorname{sgn}(\pi)\delta^{(j)}_{\alpha_1\dots \alpha_m}$, we find
                \begin{equation}
                    \operatorname{contr}_{\pi}\!\left(\tobs^{\otimes 2}\right) = \operatorname{sgn}(\pi)\sum_{j=1}^{\lfloor 2n/m \rfloor}\omega_j^2 \delta^{(j)}_{\alpha_1\dots \alpha_m}\delta^{(j)}_{\alpha_1\dots \alpha_m} = \operatorname{sgn}(\pi) m!\sum_{j=1}^{\lfloor 2n/m \rfloor} \omega_j^2 \in \Theta(n)\,,
                \end{equation}
                where we used the unpermuted full contraction $\delta^{(j)}_{\alpha_1\dots \alpha_m}\delta^{(j)}_{\alpha_1\dots \alpha_m}=m!$, 
                and the fact that the weights are asymptotically bounded as $\omega_j \in \Theta(1)$, ensuring the sum scales within linear bounds with qubit number $n$.
                Consequently, the denominator in Eq.~\eqref{eq:special_obs_even_moments}, corresponding to the desired contractions in the even $k$ moments achieved via pairwise tensor matchings, implies that each of the $k/2$ tensor pairs yields a factor $\lVert O \rVert_{m}^2 = \lvert\operatorname{contr}_{\pi}(\tobs^{\otimes 2})\rvert$ for any $\pi \in S_m$, resulting in an overall denominator scaling of $\Theta(n^{k/2})$.
                
                Now, the spurious contractions $\operatorname{contr}_{\sigma}( \tobs^{\otimes k} )$ with $\sigma \in \mathfrak{B}_{mk/2}$ restricted to a non-pairwise tensor matching necessarily involve tensor matchings of three or more copies of $\tobs$. While the previously proven Proposition~\ref{thm:special_obs_odd_moments} allows us to immediately discard any contraction involving an odd number of copies, we still need to address those contractions matching an even number, greater than 2, of them.
                However, whenever such a non-pairwise contraction arises and does not vanish, the mutual orthogonality of the generalized Kronecker deltas forces that any two contracted tensors must share the exact same subset index $j$. By transitivity, every tensor within a fully connected component is locked to the identical subset index. This rigidly collapses their independent multiple index summations into a single, shared summation over the index subsets in $\bm{\VC}$. Because neither the weights $\omega_j$ nor the generalized Kronecker delta tensors carry any dependence on $n$ (yielding only $\Theta(1)$ combinatorial factors in $m$), this single summation over the index subsets in $\bm{\VC}$ evaluates strictly to $\Theta(n)$.
                We thus understand that any generalized cycle\footnote{We use the term \emph{generalized cycle} to denote a fully connected, inseparable component of tensors within the overall contraction $\sigma$, which, in contrast to the case of tensors with rank $m=2$ (matrices), no longer form visual cycles in a tensor network diagram.} $\mathcal{C}$, no matter its length $|\mathcal{C}|$, between the $k$ copies of tensor $\tobs$ contributes a factor of $\Theta(n)$ to the total contraction $\operatorname{contr}_{\sigma}( \tobs^{\otimes k})$ which hence scales as
                \begin{equation}
                    \operatorname{contr}_{\sigma}\left( \tobs^{\otimes k} \right) = \prod_{\mathcal{C}\in \operatorname{cycles}(\sigma)} \Theta(n) = \Theta\left(n^{\lvert\operatorname{cycles}(\sigma)\rvert}\right)\,,
                \end{equation}
                where we defined $\operatorname{cycles}(\sigma)$ as the generalized cycles of $\sigma$ and, hence, $\lvert\operatorname{cycles}(\sigma)\rvert$ as the number of such cycles.
                Since this is exactly maximized by the contractions corresponding to the pairings, i.e., cycles of length 2, it follows that the spurious contractions vanish at least as $\mathcal{O}(1/n)$ relative to the pairwise denominator $\lvert\operatorname{contr}_{\pi}(\tobs^{\otimes 2})\rvert^{k/2}$, which concludes the proof.
            \end{proof}

            Given the vanishing of the odd moments and the partial fulfillment of Condition~\ref{cond:gaussian_dominate} for the even moments, this proves that for this special class of observables defined in Eq.~\eqref{eq:special_obs}, GPs emerge via Theorem~\ref{thm:supp:GP}. 
            Generally, the complete fulfillment of Condition~\ref{cond:gaussian_dominate} is subject to an appropriate scaling behavior of the states, which can be checked equivalently via Proposition~\ref{prop:cond:univariate} in the univariate case, or sufficiently via Proposition~\ref{prop:cond:sufficient_mutivariate} (or directly by probing Condition~\ref{cond:gaussian_dominate} or Proposition~\ref{prop:cond:unified_single_sum}).

            Furthermore, we note that the partitioning requirement can be relaxed such that the subsets do not need to cover all $2n$ indices. As long as the subsets $\VC^{(j)}$ remain mutually disjoint, which guarantees the exact vanishing of the odd moments, and the total number of partitions $|\bm{\VC}|$ scales linearly with $n$ (i.e., $|\bm{\VC}| = \Theta(n)$), the asymptotic domination of the Gaussian pairings in the even moments established in Proposition~\ref{thm:special_obs_even_moments} remains entirely valid. For instance, single Pauli observables in $\BC_m$ do not meet such $\lvert \bm{\VC}\rvert$ scaling. While the odd moments vanish since Proposition~\ref{thm:special_obs_odd_moments} even applies to this relaxed observable form, such single Pauli observables do not suffice to achieve the required even-moment scaling (Proposition~\ref{thm:special_obs_even_moments} does not apply), e.g., consistent with Proposition~\ref{thm:pauli_ratio_b4} for $\BC_4$.

\clearpage
\section{Efficient measurement protocol for matchgate kernel estimation}\label{app:kernel_estimation}

    We introduce a measurement protocol for the matchgate quantum kernel functions $\kappa(\rho, \rho')$ that avoids an indirect and naive estimation of the module projection coefficients (Eq.~\eqref{supp:eq:projection_via_basis}) for both states $\rho, \rho'$ and classically computing the kernel value based on them according to the module overlap definition in Eq.~\eqref{supp:eq:module_overlap_coeffs}. Instead, a direct estimation by simultaneously preparing $\rho \otimes \rho'$, gathering Bell basis measurement statistics, and classical (efficient) post-processing provides kernel estimates.
    The relevant quantity that has to be estimated from the quantum input states $\rho, \rho'$ is the module overlap $\left\langle \rho, \rho' \right\rangle_m$, while the remaining factors ${m!}/{(2n)^m}\lVert{{O}}\rVert^2_m$ are known exactly classically. Hence, we denote the matchgate kernel function explicitly as a function of $\rho, \rho'$ as $\kappa(\rho, \rho')$ here, as opposed to a function of state labels $t, t'$.

    It is more convenient to consider an equivalent form, introducing the natural scaling factor $d = 2^n$, as
    \begin{equation}\label{eq:matchgate_kernel_factors}
        \kappa(\rho, \rho')
        =
        \frac{m!}{(2n)^m}
        \left( d^{-1} \lVert{{O}}\rVert^2_m \right)\left( d \langle \rho, \rho' \rangle_m \right).
    \end{equation}
    We can think of these rescaled module overlaps and norms as the effective module quantities, which typically arise as follows.
    The (squared) effective module norm $d^{-1} \lVert{{O}}\rVert^2_m$, is computed classically, and typically corresponds to a quantity that, even after dividing by $d$, does not vanish exponentially with $n$ for typical observables $\O$, such as any Pauli strings in $\BC_m$. Furthermore, when expressing $\O$ directly as a linear combination of distinct Majorana products (or their corresponding Pauli strings) in $\BC_m$, i.e., without the normalization factor in Eq.~\eqref{supp:eq:module_basis}, the corresponding sum of squared coefficients yields $d^{-1} \lVert{{O}}\rVert^2_m$. The second term, $d \langle \rho, \rho' \rangle_m$ resembles an effectively measured variant of the module overlaps, whose estimation will be discussed in detail below.

\subsection{Direct matchgate kernel estimation}\label{app:ff_kernel_estimation_protocol}

    The direct estimation protocol for a matchgate kernel value $\kappa(\rho, \rho')$ consists of the following steps, which are detailed subsequently:
    \begin{enumerate}
        \item Prepare states $\rho$ and $\rho'$ coherently
        \item Perform measurements in the Bell basis
        \item Classically compute (in polynomial time) an unbiased estimate of $d\Braket{\rho, \rho'}_m$ and, hence, $\kappa(\rho, \rho')$ via Eq.~\eqref{eq:matchgate_kernel_factors}
    \end{enumerate}
    The module overlap can be written as
    \begin{equation}\label{eq:module_overlap_expectation}
        \Braket{\rho, \rho'}_{m}
        =
        \Tr\left[ 
        (\rho \otimes \rho') 
        \sum_{\idxset } C_{\idxset}^\dagger \otimes C_{\idxset} 
        \right]\,,
    \end{equation}
    where $\VC \subseteq \lbrace 1, \ldots, 2n\rbrace$ with $\lvert \VC \rvert = m$ denotes all subsets of $\lbrace 1, \ldots, 2n\rbrace$ of cardinality $m$, i.e., $\lbrace C_{\idxset}\rbrace_{\idxset}$ forms a basis for module $\BC_m$ as defined in Eq.~\eqref{supp:eq:module_basis}.
    Analogously, the rescaled variant, which is the quantity to be measured, reads as
    \begin{equation}\label{eq:rescaled_module_overlap_expectation}
        d\Braket{\rho, \rho'}_{m}
        =
        \Tr\left[ 
        (\rho \otimes \rho') 
        \sum_{\idxset } \sqrt{d}C_{\idxset}^\dagger \otimes \sqrt{d}C_{\idxset} 
        \right]
        =
        \Tr\left[
        (\rho \otimes \rho') 
        \sum_{\idxset } P_{\idxset}\otimes P_{\idxset} 
        \right]        ,
    \end{equation}
    where we denote $P_{\idxset}$ as the Majorana product that corresponds to the module basis operator $C_{\idxset}$ without the normalization factor\footnote{The factor $h_m$ in Eq.~\eqref{supp:eq:module_basis} cancels out due to its unit modulus.}. Hence, $P_{\idxset}$ is a Pauli string.
    For any states $\rho, \rho'$, the effective module overlaps can be bounded\footnote{Note that for certain modules, bounds are known to be tighter than the module dimension $d_m$, such as $n$ in the $\BC_2$ case \cite{deneris2025analyzing}.} as
    \begin{equation}\label{eq:rescaled_module_overlap_bound}
        \left\lvert d\Braket{\rho, \rho'}_{m}\right\rvert \leq d_m \quad \text{with} \quad d_m = \dim(\BC_m) = \binom{2n}{m}
        .
    \end{equation}
    
    The Bell basis $\lbrace \ket{P} \rbrace_{P\in \PC}$, with $P$ an $n$-qubit Pauli string and $\PC = \lbrace I, X, Y, Z \rbrace^{n}$ denoting the set of all $4^n$ such strings, is defined as
    \begin{equation}\label{eq:bell_basis_states}
        \ket{P} = \frac{1}{\sqrt{d}} P \otimes I_n \ket{\Omega}
    \end{equation}
    with $\ket{\Omega}$ denoting the un-normalized maximally entangled state over two $n$-qubit systems $\ket{\Omega} = \sum_{i=1}^{2^n} \ket{i} \otimes \ket{i}$.
    An important property is that all terms in the operator sums in Eqs.~\eqref{eq:module_overlap_expectation} and~\eqref{eq:rescaled_module_overlap_expectation} diagonalize in the Bell basis:
    \begin{proposition}
        For any module basis operator $C_{\idxset}$, the tensor product $C_{\idxset}^\dagger \otimes C_{\idxset}$ is diagonal in the Bell basis $\lbrace \ket{P} \rbrace_{P\in \PC}$.
    \end{proposition}
    \begin{proof}
        Via Jordan-Wigner mapping and module basis definition in Eq.~\eqref{supp:eq:module_basis}, any tensor product of a product of $m$ Majorana operators and its adjoint $\sqrt{d} C_{\idxset}^\dagger \otimes \sqrt{d} C_{\idxset} = C_{n_1} C_{n_2} \cdots C_{n_m} \otimes C_{n_1} C_{n_2} \cdots C_{n_m}$ is represented by $P \otimes P$ with $P\in \PC$. 
        Note that $I, X, Z$ are symmetric while $Y$ is anti-symmetric, i.e., $I^\top = I, X^\top = X, Y^\top = -Y, Z^\top = Z$ and that any Pauli strings $P, Q$ either commute or anti-commute, leading to either $PQ = QP$ or $PQ = -QP$. As a special case of the latter, Pauli conjugation preserves the Pauli as either $PQP = Q$ or $PQP = -Q$.
        Together with the identity $A \otimes I \ket{\Omega} = I \otimes A^\top \ket{\Omega}$ \cite{nielsen2000quantum_refs} it follows that
        \begin{align}
            P \otimes P \ket{Q} 
            &= \tfrac{1}{\sqrt{d}} (P \otimes P) (Q \otimes I) \ket{\Omega} \\
            &= \tfrac{1}{\sqrt{d}}(PQ \otimes I)(I \otimes P) \ket{\Omega} \\
            &= \tfrac{1}{\sqrt{d}}PQP^\top \otimes I \ket{\Omega}\\
            &= \tfrac{1}{\sqrt{d}}\varphi PQP \otimes I \ket{\Omega}\\
            &= \tfrac{1}{\sqrt{d}}\varphi\varphi' Q\otimes I \ket{\Omega}\\
            &= \varphi\varphi' \ket{Q} 
        \end{align} 
        for any ($n$-qubit) Pauli strings $P, Q \in \PC$, where $\varphi = \pm 1$ and $\varphi' = \pm 1$ depending on (transposition) symmetry of $P$ and commutation relation of $P$ and $Q$, respectively.
        Thus, $\ket{Q}$ is an eigenstate of $P \otimes P$, concluding the proof.
    \end{proof}
    Formally, the diagonalization in the Bell basis can be written as $\forall P, Q \in \PC:$
        \begin{equation}
            \sum_{\idxset }  
            \Braket{ P |  P_{\idxset}\otimes P_{\idxset}  | Q} = \delta_{P, Q} f_m(P)\,,
        \end{equation}
        with
        \begin{equation}
            f_m(P) = \sum_{\idxset }  \Braket{ P |  P_{\idxset}\otimes P_{\idxset}  | P} \,,
        \end{equation}
        and $\delta$ denoting the Kronecker delta.
    For a constant (in $n$) module order $m \in \mathcal{O}(1)$, given an observation for $P$, $f_m(P)$ can be evaluated in polynomial\footnote{The order of the polynomial upper bound is $m$ since $\binom{2n}{m} \leq (2n)^m/m!$}
    time classically since $\binom{2n}{m} \in \mathcal{O}(\poly(n))$, which is the number of terms in $f_m(P)$ as it becomes clear from this equivalent form:
    \begin{equation}
        f_m(P) = \sum_{\idxset} (-1)^{\mathbf{1}_{\{0\}}([P, P_{\idxset}])} (-1)^{\#Y(P_{\idxset})}\,,
    \end{equation}
    where $\mathbf{1}_{\{0\}}([P,P_{\idxset}])$ denotes the indicator function that evaluates to one if $[P,P_{\idxset}]\in \{0\}$, i.e., iff $P$ and $P_{\idxset}$ commute, otherwise to zero. $\#Y(P_{\idxset})$ counts the number of Pauli-$Y$ operators in Pauli string $P_{\idxset}$.

    Hence, we can estimate the module overlap through $M$ Bell basis measurements on copies of $\rho \otimes \rho'$.
    Measuring $\rho \otimes \rho'$ in the Bell basis can be achieved via standard computational basis measurements after a change of basis, e.g., applying CNOT and Hadamard gates on each pair of qubits between $\rho$ and $\rho'$ \cite{hangleiter2024bell}.
    Given the $M$ Bell basis measurement outcomes $\lbrace P^{(j)} \rbrace_{j=1}^M$, an unbiased estimator is obtained  
    \begin{align}
        d\braket{\rho, \rho'}_m
        &=
        \sum_{P\in\PC} \underbrace{\Braket{P| \rho \otimes \rho' |P}}_{\Pr_{\rho \otimes \rho'}[P]} f_m(P) \\
        &\approx
        \frac{1}{M}\sum_{j=1}^{M} f_m(P^{(j)})\\
        &= \widetilde{d\braket{\rho, \rho'}}_m^{(M)}.
    \end{align}
    Here, Bell measurement outcomes $P^{(j)}$ were obtained according to sampling probabilities $\Pr_{\rho \otimes \rho'}\left[P^{(j)}\right]$ and used to estimate the expectation of $f_m$.
    Note that this scheme does not conceptually differ from a module purity estimation scheme, which is merely the special case of $\rho = \rho'$.
    Note that for any $P$, the module dimension $d_m$ bounds
    \begin{equation}
        \lvert f_m(P)\rvert \leq d_m
        ,
    \end{equation}
    which is a result of $f_m(P)$ being the sum over $d_m$ Pauli expectation values with respect to the (normalized) states $\ket{P}$. Such Pauli expectation values are at most 1 in magnitude. This bound aligns with the effective module overlap bound of Eq.~\eqref{eq:rescaled_module_overlap_bound}.
    Eventually, inserting this estimate for $d\braket{\rho, \rho'}_m$ in Eq.~\eqref{eq:matchgate_kernel_factors} provides an unbiased estimate of the kernel value $\kappa(\rho, \rho')$ based on $M$ samples.

    \subsection{Naive indirect matchgate kernel estimation from module projection coefficient estimates}

        A naive implementation of a measurement protocol for the kernel would be to estimate the module projection coefficients for each quantum state to indirectly obtain kernel value estimates, as outlined in the following steps and detailed subsequently:
        \begin{enumerate}
            \item Prepare states $\rho$ and estimate expectation values of all Pauli strings in the module $\BC_m$
            \item Prepare states $\rho'$ and estimate expectation values of all Pauli strings in the module $\BC_m$
            \item Classically compute an unbiased estimate of $d\Braket{\rho, \rho'}_m$ as the inner product of these coefficient estimate vectors and, hence, $\kappa(\rho, \rho')$ via Eq.~\eqref{eq:matchgate_kernel_factors}
        \end{enumerate}
        This not only defeats the conceptual purpose of using a kernel function as an implicit approach to evaluating inner products in the underlying (module) feature space, but it also introduces an estimator with significantly higher error.

        As in the introduced measurement protocol above, rescaled coefficients are considered to enable simple Pauli string measurements representing the distinct Majorana products in the module $\BC_m$ through the Jordan-Wigner mapping, i.e., basis operators up to factor $\sqrt{d}$ in Eq.~\eqref{supp:eq:module_basis}. 
        Hence, the rescaled coefficients, as Pauli string expectation values, have a range of 
        \begin{equation}
            -1 \leq c_i(\rho) \leq 1\,,
        \end{equation}
        and provide the rescaled module overlaps via Eq.~\eqref{eq:rescaled_module_overlap_expectation} as
        \begin{equation}\label{eq:rescaled_module_overlap_via_coeffs}
            d\Braket{\rho, \rho'}_m = \sum_{i=1}^{d_m}c_i(\rho)c_i(\rho') = \bm{c}(\rho)^\top\bm{c}(\rho')
        \end{equation}
        where $d_m = \binom{2n}{m}$ and $c_i(\rho), c_i(\rho')$ denote the $i$-th module coefficients of the states $\rho$ and $\rho'$, respectively.
        Let the module basis Pauli strings provide $d_m$ estimates of the module projection coefficients using $t$ shots each, denoted as
        \begin{equation}
            \left\lbrace \tr\left[\rho P_{\idxset}\right] \right\rbrace_{\idxset} \approx 
            \left\lbrace \widetilde{c_{i}(\rho)}^{(t)}\right\rbrace_{i = 1}^{d_m}\,.
        \end{equation}
        Inserting the estimates of $c_i(\rho)$ and $c_i(\rho')$ in Eq.~\eqref{eq:rescaled_module_overlap_via_coeffs} provides an estimation of the effective module overlap $d\braket{\rho, \rho'}_m$.
        Finally, from the estimate of $d\braket{\rho, \rho'}_m$, an unbiased estimate of the kernel value $\kappa(\rho, \rho')$ based on $T = 2d_mt$ samples is obtained through substitution in Eq.~\eqref{eq:matchgate_kernel_factors}.

    \subsection{Comparison of direct and indirect matchgate kernel protocol}

        To assess the estimation efficiency of the two protocols, we first derive upper bounds on their respective errors using standard techniques, such as Hoeffding's concentration bounds.
        Beyond merely demonstrating that the direct estimator has a tighter upper bound, we establish a clear performance separation between the two protocols.
        While the direct estimator exhibits a significantly more favorable upper bound, comparing upper bounds alone is insufficient to show a separation. To rigorously establish a strict separation in estimator efficiency, we achieve this by deriving a probabilistic lower bound for the indirect estimator. This definitively rules out the possibility that tighter mathematical bounding techniques could close the performance gap, proving that the naive indirect protocol is fundamentally sub-optimal.

        \subsubsection{Error bounds of direct estimator}

            We start with error bounds on the effective module overlap estimates $\widetilde{d\braket{\rho, \rho'}}_m^{(M)}$. 
            By Hoeffding's inequality, we obtain the following bound
            \begin{equation}
                \Pr\left(\left\lvert \widetilde{d\braket{\rho, \rho'}}_m^{(M)} - {d\braket{\rho, \rho'}_m} \right\rvert \geq \epsilon \right) \leq 2\exp\left( -\frac{1}{2} M \epsilon^2d_m^{-2}\right) = \delta
            \end{equation}
            for all $\epsilon > 0$.
            Therefore, to guarantee an error of less than $\epsilon$ with probability at least $1-\delta$,
            \begin{equation}
            M \geq 2\ln(2/\delta)\epsilon^{-2}d_m^2
            \end{equation}
            measurements suffice. Alternatively, for a fixed $M$ measurements, we can guarantee with probability at least $1-\delta$ that the error bound is
            \begin{equation}\label{eq:direct_total_error}
            \epsilon = d_m \sqrt{2\ln(2/\delta) / M} \in \OC(n^m),
            \end{equation}
            which is a polynomial error bound in the number of qubits $n$.
            
            To bound the final matchgate kernel value estimate, we directly scale the intermediate error of the effective module overlap estimation by substituting in Eq.~\eqref{eq:matchgate_kernel_factors}. Using $M$ measurements, we guarantee with a probability of at least $1-\delta$ that the final estimate satisfies the error bound
            \begin{equation}
            \left\lvert \widetilde{\kappa(\rho, \rho')}^{(M)} - \kappa(\rho, \rho') \right\rvert < \sqrt{\frac{2\ln(2/\delta)}{M}} d_m \frac{m!}{(2n)^m} d^{-1} \lVert O \rVert^2 = \varepsilon.
            \end{equation}
            Consequently, the number of measurements required to achieve this guarantee is
            \begin{equation}\label{eq:direct_protocol_measurement_bound}
            M \geq 2\ln(2/\delta) \frac{d_m^2}{\varepsilon^2} \left( \frac{m!}{(2n)^m} d^{-1} \lVert O \rVert^2 \right)^2.
            \end{equation}
            Even though a factor $d$ is present in this measurement number requirement, recall that typical observables yield a norm that introduces a canceling factor of $d$. Hence, a polynomially small error can be guaranteed with a polynomial number of measurements with probability at least $1-\delta'$.
        
            The efficiency of this kernel estimation protocol can be significantly improved by reducing the variance of the unbiased estimator. This high variance stems from the heavy-tailed distribution of the samples, in which rare, large-magnitude outcomes dominate the error. To address this, we can employ established variance-reduction techniques such as the median-of-means estimator or local Clifford twirling prior to the Bell measurements \cite{huang2020predicting}. More sophisticated approaches could instead leverage fermionic classical shadows~\cite{wan2022matchgate,zhao2021fermionic}.

    \subsubsection{Error bounds of indirect estimator}

        Let $T = 2d_mt$ denote the total budget given $t$ shots per module projection coefficient per state with $d_m = \binom{2n}{m}$. 
        The probabilistic error bound through Hoeffding's inequality on each coefficient reads as 
        \begin{equation}\label{supp:eq:coeff_error}
            \Pr \left[
            \left \lvert \widetilde{c_{i}(\rho)}^{(t)} - c_{i}(\rho) \right\rvert \geq \epsilon_i
            \right] 
            \leq 
            2 \exp \left( - \epsilon_i^2 t  / 2\right)
            = \delta_i
        \end{equation}
        since $c_i \in [-1, 1]$.
        Consequently, with probability at least $1-\delta_i$, the coefficient $c_i(\rho)$ is estimated with an (upper) error bound
        \begin{equation} \label{eq:app:coeff_measurement_error_bound}
            \epsilon_i 
            = \sqrt{\frac{2\ln(2/\delta_i)}{t}}
            .
        \end{equation}
        Furthermore, let $E_i$ and $E_i'$ denote the events of Eq.~\eqref{supp:eq:coeff_error}, which occurs with probability $\delta_i$ and $\delta_i'$, for the two states $\rho$ and $\rho'$, respectively. 
        As per the union bound, the probability that at least one event occurs is at most
        \begin{equation}\label{eq:app:coeff_est_union_bound}
        \delta = \mathrm{Pr}\left[\left(\bigcup_{i=1}^{d_m} E_i \right) \cup \left( \bigcup_{i=1}^{d_m} E_i'\right)\right]\leq \sum_{i=1}^{d_m} \left(\mathrm{Pr}(E_i) + \mathrm{Pr}(E_i') \right)= \sum_{i=1}^{d_m} (\delta_i + \delta_i')
        .
        \end{equation}
        Hence, the probability that the estimation errors of all $\widetilde{c_i(\rho)}^{(t)}$ and $\widetilde{c_i(\rho')}^{(t)}$ simultaneously do not exceed $\epsilon_i$ and $\epsilon_i'$, respectively, is at least $1-\delta$.

        Under this probability bound, the error propagates in the classical inner product calculation 
        \begin{equation}
            \xi^{(t)}(\rho, {\rho'}) =
            \left\lvert 
            \sum_{i=1}^{{d_m} } \widetilde{c_i(\rho)}^{(t)} \widetilde{c_i({\rho'})}^{(t)}
            -
            d\braket{\rho, \rho'}_m
            \right\rvert
            =
            \left\lvert 
            \sum_{i=1}^{{d_m} } \widetilde{c_i(\rho)}^{(t)} \widetilde{c_i({\rho'})}^{(t)}
            -
            \sum_{i=1}^{{d_m} } c_i(\rho) c_i({\rho'}) 
            \right\rvert
            \leq 
            \sum_{i=1}^{{d_m} } (\epsilon_i + \epsilon_i')\,,
        \end{equation}
        where $\epsilon_i$ and $\epsilon_i'$ denote the coefficient error thresholds for $\rho$ and $\rho'$, respectively. 
        This bound on the propagated error is derived via the triangle inequality first over the sum
        \begin{equation}
            \left\lvert 
            \sum_{i=1}^{{d_m} } 
            \left(
            \widetilde{c_i(\rho)}^{(t)} \widetilde{c_i({\rho'})}^{(t)}
            -
             c_i(\rho) c_i({\rho'}) 
             \right)
            \right\rvert
            \leq 
            \sum_{i=1}^{{d_m} } 
            \left\lvert 
            \left(
            \widetilde{c_i(\rho)}^{(t)} \widetilde{c_i({\rho'})}^{(t)}
            -
             c_i(\rho) c_i({\rho'}) 
             \right)
            \right\rvert\,,
        \end{equation}
        and then over the individual product error terms
        \begin{align}
            \left\lvert 
            \left(
            \widetilde{c_i(\rho)}^{(t)} \widetilde{c_i({\rho'})}^{(t)}
            \! -
             c_i(\rho) c_i({\rho'}) 
             \right)
            \right\rvert
            &=
            \left\lvert
            \frac{
            (\widetilde{c_i(\rho)}^{(t)} \! - c_i(\rho))
            (\widetilde{c_i({\rho'})}^{(t)} \! + c_i({\rho'}))
            +
            (\widetilde{c_i(\rho)}^{(t)} \! + c_i(\rho))
            (\widetilde{c_i({\rho'})}^{(t)} \! - c_i({\rho'}))
            }{2}
            \right\rvert \\
            &\leq 
            \frac{1}{2}\left\lvert\widetilde{c_i(\rho)}^{(t)} \! - c_i(\rho)\right\rvert
            \left\lvert\widetilde{c_i({\rho'})}^{(t)} \! + c_i({\rho'})\right\rvert
            +
            \frac{1}{2}
            \left\lvert\widetilde{c_i(\rho)}^{(t)} \! + c_i(\rho)\right\rvert
            \left\lvert\widetilde{c_i({\rho'})}^{(t)} \! - c_i({\rho'})\right\rvert
            \\
            &\leq 
            \epsilon_i(\rho) + \epsilon_i({\rho'})
            \label{eq:app:clipping_ineq}
            .
        \end{align}
        Equation~\eqref{eq:app:clipping_ineq} 
        bounds the sums by $2$ and substitutes the error bounds as defined in Eq.~\eqref{eq:app:coeff_measurement_error_bound}.

        By satisfying the union bound of Eq.~\eqref{eq:app:coeff_est_union_bound} via $\delta_i = \delta_i' = \delta / (2d_m)$, we obtain the uniform threshold $\epsilon' = \sqrt{2\ln(4d_m/\delta) / t}$. The total error $\xi^{(t)}(\rho, \rho')$ in the effective module overlaps from estimating $2d_m$ coefficients with $t$ measurements each, yielding a total shot budget of $T = 2d_mt$, is bounded by
        \begin{equation}
        \xi^{(t)}(\rho, \rho') \leq 2d_m\epsilon' = 2d_m \sqrt{ \frac{2\ln(4d_m/\delta)}{t} }\,.
        \end{equation}
        Substituting $t = T / (2d_m)$, we achieve the final error bound $\epsilon$ with respect to the total measurement budget $T$
        \begin{equation}\label{eq:indirect_total_error}
        \xi^{(t)}(\rho, \rho') \leq 2d_m \sqrt{ \frac{4d_m\ln(4d_m/\delta)}{T} } = 4d_m^{3/2} \sqrt{ \frac{\ln(4d_m/\delta)}{T} } = \epsilon \in \tilde{\mathcal{O}}(n^{3m/2})
        \end{equation}
        with probability at least $1-\delta$.

        We restrict the present error analysis to estimating a single kernel value. In practice, when constructing a full kernel matrix, one could theoretically reuse the module coefficient estimates of individual quantum states across multiple pairwise module overlaps. While this data reuse might justify allocating more measurement shots per state to further suppress variance, it inevitably introduces biased estimators. Most notably in the kernel matrix diagonal, i.e., module purities, where the finite-sampling variance positively biases the squared coefficient estimates. Because our comparison relies entirely on unbiased estimators, examining the trade-offs of such biased estimation strategies is beyond the scope of this work.

    \subsubsection{Strict polynomial separation favoring the direct estimator}
        
        A direct comparison of the error bounds in Eqs.~\eqref{eq:direct_total_error} and~\eqref{eq:indirect_total_error} reveals a clear difference in the scaling of the two estimators. Specifically, the error of the direct protocol scales strictly linearly with the module dimension as $\mathcal{O}(d_m) = \mathcal{O}(n^m)$, whereas the propagated error of the naive indirect protocol scales as $\tilde{\mathcal{O}}(d_m^{3/2}) = \tilde{\mathcal{O}}(n^{3m/2})$, which is a significant polynomial divergence of $\tilde{\mathcal{O}}(n^{m/2})$ with respect to the number of qubits $n$. However, to show a separation, we must rule out the possibility that the upper error bound of the indirect protocol is merely loose and could be improved. We show that it is, in fact, asymptotically tight in the module dimension $d_m$ (up to logarithmic factors) by establishing a lower bound on a specific (hard) instance. The error of the indirect protocol for this instance strictly exceeds the general upper bound of the direct estimator with respect to $d_m$.

        For the hard instance, we assume all projection coefficients vanish, i.e., $c_i(\rho) = c_i(\rho') = 0$ for all $i = 1,\dots, d_m$. We briefly provide examples to illustrate that such zero-coefficient scenarios arise naturally.
        For odd $m$, setting $\rho, \rho'$ to any computational basis states yields strictly zero projection coefficients for $\BC_m$. This lack of support occurs because computational basis states decompose exclusively into Pauli strings over $\{I, Z\}$. These strings belong solely to even-parity modules, as any product of an odd number of distinct Majorana operators necessarily introduces a Pauli-$X$ or $Y$ operator.
        
        For even $m$\footnote{Note that we neglect the trivial cases of $\BC_0$ and $\BC_{2n}$ and focus on intermediate modules $\BC_m$ of $1 < m < 2n$ with even $m$.}, such a scenario can be trivially realized using mixed states, such as $\rho, \rho'$ being the maximally mixed states of the even and odd fermionic parity sectors. For pure states, one can readily select $\rho, \rho'$ to be pure stabilizer states whose stabilizer group is constructed to contain no products of $m$ distinct Majorana operators to guarantee strictly zero support on $\BC_m$.

        To apply the Paley-Zygmund inequality for a lower bound on the indirect kernel estimation error for this hard instance of zero coefficients, we derive the estimator's mean, variance, and fourth moment. 
        Because all $t$ measurements per coefficient are independent, the coefficient estimation has a zero mean and a variance of
        \begin{equation}
            \forall i = 1,\dots, d_m: \Var[\widetilde{c_i(\rho)}] = \Var[\widetilde{c_i(\rho')}] = 1 / t,
        \end{equation}
        following an appropriately shifted and scaled binomial distribution\footnote{$\widetilde{c_i(\rho)} = \frac{2}{t}(B - t/2)$ where $B \sim \mathrm{Bin}(t, 1/2)$ for zero coefficient estimations.}.
        Furthermore, the products of the two $i$-th coefficient estimates have zero mean and variance
        \begin{equation}\label{eq:indirect_est_product_var}
            \Var\left[\widetilde{c_i(\rho)} \widetilde{c_i(\rho')}\right] = \Var[\widetilde{c_i(\rho)} ]
            \Var[\widetilde{c_i(\rho')} ] 
            = 1/t^2\,,
        \end{equation}
        such that the final module overlap estimate (as the scalar product of coefficient estimates)
        has a zero mean again and a variance of 
        \begin{equation}\label{eq:indirect_var} 
            \sigma^2 = \Var\left[\sum_{i=1}^{d_m} \widetilde{c_i(\rho)} \widetilde{c_i(\rho')} \right] = 
            \sum_{i=1}^{d_m}\Var\left[ \widetilde{c_i(\rho)} \widetilde{c_i(\rho')} \right]
            = d_m/t^2,
        \end{equation}
        all due to independence, leading to any cross-terms vanishing.
        For the fourth moments of a single coefficient estimate $\widetilde{c_i(\rho)}$, using the known $4$-th central moment of the Binomial distribution, we directly obtain 
        \begin{equation}
        \mathbb{E}[\widetilde{c_i(\rho)}^4] = \left(\frac{2}{t}\right)^4 \frac{t(3t-2)}{16} \le \frac{3}{t^2}.
        \end{equation}
        The same derivations apply for $\rho'$.
        Consequently, for the product of two independent coefficient estimates, 
        \begin{equation}\label{eq:indirect_est_product_4th}
            \E\left[\left(\widetilde{c_i(\rho)}\widetilde{c_i(\rho')}\right)^4\right] = \mathbb{E}[\widetilde{c_i(\rho)}^4]\mathbb{E}[\widetilde{c_i(\rho')}^4] \le \frac{9}{t^4}.
        \end{equation}
        Finally, the $4$-th moment for the indirect estimator obeys the decomposition
        \begin{equation}
            \E\left[\left(\sum_{i=1}^{d_m}\widetilde{c_i(\rho)}\widetilde{c_i(\rho')}\right)^4\right] = \sum_{i=1}^{d_m}\E\left[\left(\widetilde{c_i(\rho)}\widetilde{c_i(\rho')}\right)^4\right] + 3 \sum_{\substack{i,j=1\\i\neq j}}^{d_m} \E\left[\left(\widetilde{c_i(\rho)}\widetilde{c_i(\rho')}\right)^2\right]\E\left[\left(\widetilde{c_j(\rho)}\widetilde{c_j(\rho')}\right)^2\right]
            ,
        \end{equation}
        where all the other (odd power) cross-terms vanished due to the mean zero estimates $\E[\widetilde{c_i(\rho)}] = 0$. Substituting the coefficient product variance of Eq.~\eqref{eq:indirect_est_product_var} and $4$-th moment bound of Eq.~\eqref{eq:indirect_est_product_4th} yields a fourth moment bound of the indirect estimator as 
        \begin{equation}
            \E\left[\left(\sum_{i=1}^{d_m}\widetilde{c_i(\rho)}\widetilde{c_i(\rho')}\right)^4\right] \leq \frac{9d_m}{t^4} + 3\frac{d_m(d_m - 1)}{t^4} = \frac{3d_m^2 + 6d_m}{t^4},
        \end{equation}
        which is further expressed as multiples of the estimator variance squared $\sigma^4 = d_m^2/t^4$ as
        \begin{equation}\label{eq:indirect_4th}
            \E\left[\left(\sum_{i=1}^{d_m}\widetilde{c_i(\rho)}\widetilde{c_i(\rho')}\right)^4\right] \leq \left(3 + \frac{6}{d_m}\right)\sigma^4
            \leq 9 \sigma^4
            .
        \end{equation}
        Note that the last inequality is obtained via $d_m \geq 1$.

        To lower bound the estimation error $\xi^{(t)}(\rho, {\rho'})$ of the indirect estimator from the effective module overlap $d\braket{\rho, \rho'}_m = 0$, which vanishes in the cases where states lack support on the module $\BC_m$, we apply the Paley–Zygmund inequality for the square estimation error $\xi^{(t)}(\rho, {\rho'})^2$ first, which reads as
        \begin{equation}
            \Pr\left( 
            \xi^{(t)}(\rho, {\rho'})^2
            > \gamma \E\left[\xi^{(t)}(\rho, {\rho'})^2\right]
            \right)
            \geq (1-\gamma)^2 \frac{\E\left[\xi^{(t)}(\rho, {\rho'})^2\right]^2}{\E\left[\xi^{(t)}(\rho, {\rho'})^4\right]}
        \end{equation}
        for any $0 < \gamma < 1$.
        Considering that $\E\left[\xi^{(t)}(\rho, {\rho'})^2\right] = \sigma^2$ being the variance and $\E\left[\xi^{(t)}(\rho, {\rho'})^4\right]$ the $4$-th moment of the estimation error and substituting from Eqs.~\eqref{eq:indirect_var} and~\eqref{eq:indirect_4th}, respectively, yields
        \begin{equation}
            \Pr\left( 
            \xi^{(t)}(\rho, {\rho'})^2
            > \gamma \frac{d_m}{t^2}
            \right)
            \geq (1-\gamma)^2 \frac{1}{9}\,.
        \end{equation}
        Setting $\gamma = 0.5$, inserting the total shot budget $T = 2d_mt$, and applying the square root, we obtain the final lower error bound for the indirect estimator as
        \begin{equation}\label{eq:indirect_lower_bound}
            \Pr\left(\xi^{(t)}(\rho, {\rho'}) > \frac{\sqrt{2d_m^{3}}}{T}\right) \ge \frac{1}{36}
            .
        \end{equation}
        Note that this probabilistic lower bound holds with a constant probability of $1/36$.

        This proves that the (Hoeffding) upper bound $\tilde{\OC}(d_m^{3/2}/\sqrt{T})$ of the indirect estimator is asymptotically tight up to logarithmic factors with respect to the module dimension $d_m$ and fixed measurement budget $T$, as Eq.~\eqref{eq:indirect_lower_bound} yields the matching lower bound of $\Omega(d_m^{3/2}/T)$.
        The full physical consequence of this separation becomes clear when considering the estimation error (for fixed $T$) relative to the true scale of the effective module overlap $d\braket{\rho, \rho'}_m$. This quantity is proportional to the module dimension $d_m$, e.g., bounded by $d_m$. Hence, the estimation error relative to the range of the effective overlap to be estimated becomes independent of the module dimension $d_m$ for the direct estimator as $\OC(1/\sqrt{T})$. In stark contrast, the relative error of the indirect approach diverges with the module dimension $d_m$ as $\tilde{\Omega}(\sqrt{d_m}/T)$. 
        Consequently, to maintain a constant relative error as the system grows, the indirect protocol requires an increase in measurement budget by a factor of $\tilde{\mathcal{O}}(n^{m/2})$. By avoiding this dimensional scaling penalty entirely, the direct estimator establishes a strict polynomial separation in sample complexity.

\FloatBarrier
\clearpage
    \section{Numerical benchmarking}\label{app:bench}
        Here, we complement the numerical demonstrations of QGPs across various applications and tasks with systematic benchmarking in supervised learning (regression) and optimization settings. Moreover, the introduced kernel correction via Wigner's semicircle law is numerically validated against standard approaches.

\paragraph{Recap of matchgate QGP equations.}
For self-containment, we summarize the key equations of the matchgate QGP framework below, as introduced in the main text. The prior is defined by a zero mean, $\mu(t) = 0$, and the matchgate covariance kernel is
\begin{equation}\label{eq:SI_theo:1-main}
    \kappa(t,t')=\frac{m!}{(2n)^m} \lVert O \rVert^2_m \left\langle \rho(t) , \rho(t') \right \rangle_{m}.
\end{equation}
Evaluated over a finite set of points $\TC$, the random vector $\bm{f}_{\TC}$ follows a joint Gaussian distribution,
\begin{equation}\label{eq:SI_GP1}
    \bm{f}_{\TC} \sim \mathcal{N}\!\bigl( \bm{\mu}_\TC,\, \bm{\Sigma}_{\TC,\TC} \bigr)\,,
\end{equation}
where $\bm{\Sigma}_{\TC,\TC}$ is the covariance matrix with entries $\kappa(t,t')$, $t, t'\in \TC$. Given noisy training observations
\begin{equation}\label{eq:SI_GP2}
    \bm{y}_{\TC} \;=\; \bm{f}_{\TC} + \bm{\varepsilon},\qquad
    \bm{\varepsilon}\sim\mathcal{N}(\mathbf{0},\,R),
\end{equation}
with the (typically diagonal $R=\sigma^2\openone$) noise covariance matrix $R$, the joint Gaussian distribution for an unseen $t$ is
\begin{equation}\label{eq:SI_joint_GP}
    \begin{bmatrix}
    f(t) \\[2pt]
    \bm{y}_{\TC}
    \end{bmatrix}
    \sim
    \mathcal{N}\!\left(
    \begin{bmatrix}
    \mu(t) \\[2pt]
    \bm{\mu}_{\TC}
    \end{bmatrix},
    \begin{bmatrix}
    \kappa(t,t) & \bm{k}_{t\TC}^{\!\top} \\
    \bm{k}_{t\TC} & K_{\TC\TC}+R
    \end{bmatrix}
    \right)\,,
\end{equation}
with $\bm{k}_{t\TC} = \bigl[\kappa(t,s)\bigr]_{s\in\TC}$ and $K_{\TC\TC}=\bigl[\kappa(s,s')\bigr]_{s,s'\in\TC}$. Conditioning on the observations then yields the posterior predictive distribution for $f(t)$,
\begin{equation}\label{eq:SI_GP-posterior}
    \mathrm{Pr}\left(f(t) \middle| \bm{y}_{\TC} \right) = \mathcal{N}\left( 
    \mu_{\text{P}}(t),
    \sigma_{\text{P}}^{2}(t)
    \right),
\end{equation}
where the posterior mean and variance are given by
\begin{align}
    \mu_{\text{P}}(t) &= \mu(t) + \bm{k}_{t\TC}^{\!\top}\bigl(K_{\TC\TC}+R\bigr)^{-1} \bigl(\bm{y}_{\TC}-\bm{\mu}_{\TC}\bigr),\label{eq:SI_post-mean}\\
    \sigma_{\text{P}}^{2}(t) &= \kappa(t,t) - \bm{k}_{t\TC}^{\!\top}\bigl(K_{\TC\TC}+R\bigr)^{-1}\bm{k}_{t\TC}.\label{eq:SI_post-var}
\end{align}

    \subsection{Regression}\label{app:bench:regression}

        In this section we benchmark the  matchgate QGP regression application against standard baselines, validating the QGP's superior interpolation capabilities to achieve higher accuracy in regimes with fewer measurements than standard approaches. We compare QGPs against three counterparts:
        
        \paragraph{Classical simple baselines.} By interpolating the observations over the input state parameter $t$ without access to the actual state $\rho(t)$. The most basic baseline predicts the constant mean of the training data, while more refined yet simple baselines ``connect the training dots'' by linearly interpolating the training points or smoothly using cubic splines. Sufficiently accurate results from these classical baselines render the regression task trivial and mark a regime in which quantum methods are obsolete. Note that these simple baselines would be guaranteed to fail in extrapolation tasks. Constructing robust classical extrapolation models solely on input $t$, without observing any properties of the corresponding state $\rho(t)$, would generically be  exceedingly challenging.

        \paragraph{Parametric quantum learning method.}
        Due to the duality of kernel methods, which implicitly evaluate inner products in a feature space \cite{hofmann_KernelMethodsMachine_2008}, we construct a parametric model based on features extracted from the input states $\rho(t)$. For the matchgate kernels, these features correspond to the expectation values of the module $\BC_m$ basis operators, expressed as Pauli strings. Hence, the parametric model is a linear model
        \begin{equation}\label{eq:param_pred}
            y(t) = \bm{w}^\top\bm{c}(t)
        \end{equation}
        in the $\binom{2n}{m}$-dimensional feature vector $\bm{c}_t$ of state $\rho_t$. Here,  the least-squares solution on the training labels $\bm{y}_\TC$ and inputs $\TC$ with features arranged in a matrix $C_{\TC}$ yields the weight parameters 
        \begin{equation}
            \bm{w} = (C_{\TC}^\top C_{\TC} + \sigma_{\rm obs}^2 I)^{-1} C_{\TC}^\top \bm{y}_{\TC}.
        \end{equation}
        More precisely, this parametric method is ridge regression. By matching kernel and explicit feature space as well as the ridge regularizer being set to the observation variance $\sigma_{\rm obs}^2$, matching the Gaussian observation prior in Eq.~\eqref{eq:SI_GP2}\footnote{For heteroscedastic observation noise, where $R$ in Eq.~\eqref{eq:SI_GP2} is a diagonal matrix with distinct entries, the equivalent ridge regression formulation uses a per-point regularizer $ D = R^{-1} = \operatorname{diag}\left( {1}/{\sigma_t^2} \right)_{t \in \TC} $ to determine the weights $ \bm{w} = (C_{\TC}^\top D C_{\TC} + I)^{-1} C_{\TC}^\top D \bm{y}_{\TC} $ \cite{bishop2006patternCORRECT}.}, the predictions of the parametric model in Eq.~\eqref{eq:param_pred} match the predictive mean of the QGP in Eq.~\eqref{eq:SI_post-mean} \cite{rasmussen_GaussianProcessesMachine_2006}. While this holds for exact values, differences may arise with finite measurements for the kernel and feature estimates.

        \paragraph{Hamiltonian learning (HL).}
        As an alternative quantum learning method for this regression task, we also consider HL. Specifically, we recall that HL is given query access to the unknown (matchgate) transformation, enabling it to send arbitrary states and perform measurements in arbitrary Pauli bases. 
        Since we still retain the inductive bias of knowing that the unknown evolution $U$ belongs to the matchgate group, and assuming that an observable $O\in\BC_m$ is measured, the HL task can be reduced to learning the decomposition of $O_U=UOU^\dagger$ over an orthonormal basis of $\BC_m$. Namely, we want to learn $\left\langle O_U \right\rangle_k=\Tr[O_U P^{(k)}]/2^n$, for $\{P^{(k)}\}_{k=1}^{\dim(\BC_m)}$ a Pauli basis of $\BC_m$ with size $\dim(\BC_m) = \binom{2n}{m}$. 
        Notice that this matches the model in the parametric approach. Hence, we aim to determine the module projection coefficients for the measurement operator $O_U$. Once the latter is learned one can predict the outcomes of measuring $O_U$ over any quantum state $\rho$ by analogously determining the module projection coefficients of $\rho$, which does not require queries to $U$.
        In contrast to the parametric quantum learning method, we do not assume a dataset but query access to the unknown transformation $U$.
        In practice, we adopt the following learning scheme, which is based on learning a single column of the transfer matrix associated with the adjoint action on $O$. Let $P^{(k)}$ be a Pauli operator in $\BC_m$, we want to prepare the quantum state $\sigma=\frac{1}{2^n}(\openone+P^{(k)})$, send it through $U$ and measure the expectation value of $O$. This can be accomplished with the following steps, to be executed at each run of the experiment (where we will use $j$ to denote the index of the run and initialize the sign factor as $\eta_j = 1$), and to be followed for all qubits $i=1,\ldots, n$:
        \begin{itemize}
            \item If $P_i^{(k)}=\openone$ (i.e., the Pauli $P^{(k)}$ acts trivially on the $i$-th qubit), prepare the qubit in the state $\ket{0}$ or $\ket{1}$, with probability $1/2$.
            \item If $P_i^{(k)}\neq \openone$ (i.e., the Pauli $P^{(k)}$ acts non-trivially on the $i$-th qubit), prepare the qubit in the $\ket{+_i^{(k)}}$ or $\ket{-_i^{(k)}}$ eigenstates of $P_i^{(k)}$, with probability $1/2$. Then, update $\eta_j$ by multiplying it with $1$ if $\ket{+_i}^{(k)}$ was prepared or with $-1$ if $\ket{-_i^{(k)}}$ was prepared. 
        \end{itemize}
        Once the previous (pure) state is prepared on the $j$-th experiment, it is sent through $U$ and $O$ is measured, resulting in an estimate $\langle\widetilde{O}\rangle_k$ of $\langle O_U\rangle_k$.

        \begin{figure}
            \centering
            \includegraphics[width=\linewidth]{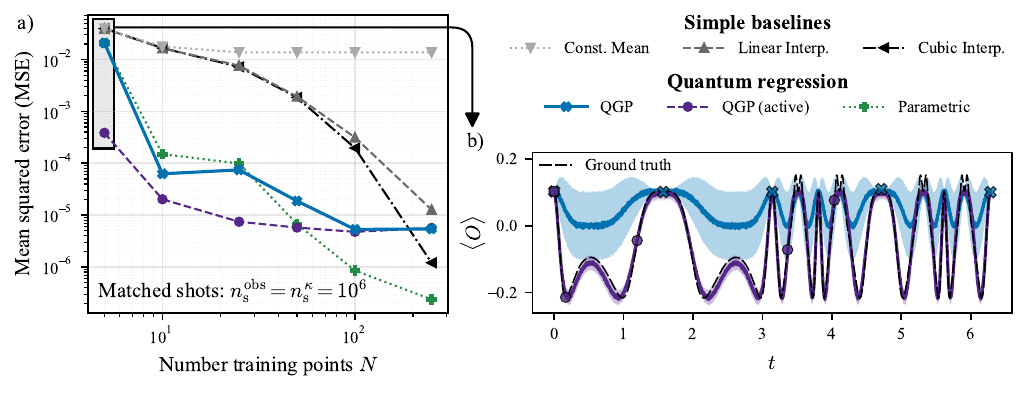}
            \caption{%
            \textbf{QGP regression benchmark against baselines and the advantage of active learning.} 
    \textbf{(a)} Mean squared error (MSE) evaluated across varying training set sizes $N$. To guarantee a fair comparison, the total shot budget is matched to the realistic counts $\shotsobs = \shotsker = 10^6$ across all models at a given $N$. The total shot budget scales with $N$. Unlike quantum regression models, simple baselines neglect observation noise, causing interpolations to fit the noisy observations exactly.
    \textbf{(b)}
    Regression results for a limited dataset of $N=5$ points as highlighted by the bounding box in (a). The plot contrasts the predictive distributions of the fixed-dataset QGP (blue) and the active learning QGP (purple) against the ground truth (dashed black line). Markers denote the selected training points. By iteratively querying points (purple markers) of maximum uncertainty, active learning captures the underlying dynamics entirely missed by the fixed equi-spaced dataset (blue markers).
            }
            \label{fig:bench:regression:over_ntp}
        \end{figure}

\medskip

Having defined the three methods against which we will compare QGPs, we now provide additional details of the implementation. We define quantum data input states $\rho_i$ as product states of $n=50$ qubits in which each qubit is rotated by $2\pi(\sin(t) +\cos(2t)/2)$ about both the x-axis and then about the y-axis, starting from an initialization in $\ket{0}$.         A random matchgate circuit is applied with a measurement in a single $\BC_2$ Pauli observable. These assumptions inform the quantum kernel to be $\BC_2$ matchgate for the QGP (and equivalently the quantum feature to be the expectation values of the $\BC_2$ basis operators in Pauli representation in the parametric formulation).
        The dataset comprises quantum states $\rho_i$ corresponding to $t_i$ equidistantly sampled from the interval $[0, 2\pi]$ (including samples on the boundaries).
        Note that the performance metric, mean squared error (MSE), is evaluated on an evenly spaced test set of size $9900 - N$ after removing the $N$ training points.

        Beyond standard QGP regression on the fixed, equidistant training dataset, we leverage a unique strength of uncertainty quantification to use an active learning QGP protocol that picks a training dataset (of matching size) dynamically throughout learning. Hence, in the parametric view, an active learning protocol cannot be directly devised due to the lack of inherent uncertainty quantification. More details on this method are provided in Sec.~\ref{sec:methods:active_learning}.

        We first establish the boundary between trivial and non-trivial performance by evaluating simple classical baselines under a realistic shot noise setting ($\shotsobs = \shotsker = 10^6$), as presented in Supp.~Fig.~\ref{fig:bench:regression:over_ntp}a. 
        The constant mean baseline reveals that achieving a mean squared error (MSE) within the order of $10^{-2}$ is trivial, regardless of the training set size $N$. Conversely, many closely-spaced data points can be accurately interpolated with classical linear and cubic splines, meaning that the non-trivial performance regime lies at MSEs below $10^{-3}$ with a training set of fewer than 100 points. 

        As expected, the active learning QGP consistently establishes a lower bound on the MSE compared to the fixed-dataset quantum regressors, highlighting a critical advantage in the extreme low-data regime. At the smallest evaluated dataset size ($N=5$), the fixed-dataset models fail to meaningfully surpass simple classical baselines. However, by strategically selecting training points to maximize information gain, the active learning QGP achieves a highly non-trivial MSE on the order of $10^{-4}$. This stark contrast is visualized in Supp.~Fig.~\ref{fig:bench:regression:over_ntp}b. Guided by the inherent QGP uncertainty quantification, the active learning strategy successfully aligns its predictive mean with the true signal and minimizes predictive variance, whereas the standard fixed-dataset approach fails to recover any overall dynamics. Consequently, compared to non-probabilistic models (such as the parametric formulation), QGPs offer a distinct algorithmic advantage, especially when quantum data acquisition is prohibitively expensive.
        
        Beyond the active learning paradigm, the standard fixed-dataset QGP and the parametric formulation demonstrate nearly identical performance across most dataset sizes, as expected from their theoretical equivalence. We suspect that the minor discrepancy in MSE observed in the trivially large-dataset regime arises purely from finite-precision numerical stability issues between the primal and dual implementations.

        \begin{figure}
            \centering
            \includegraphics[width=1\linewidth]{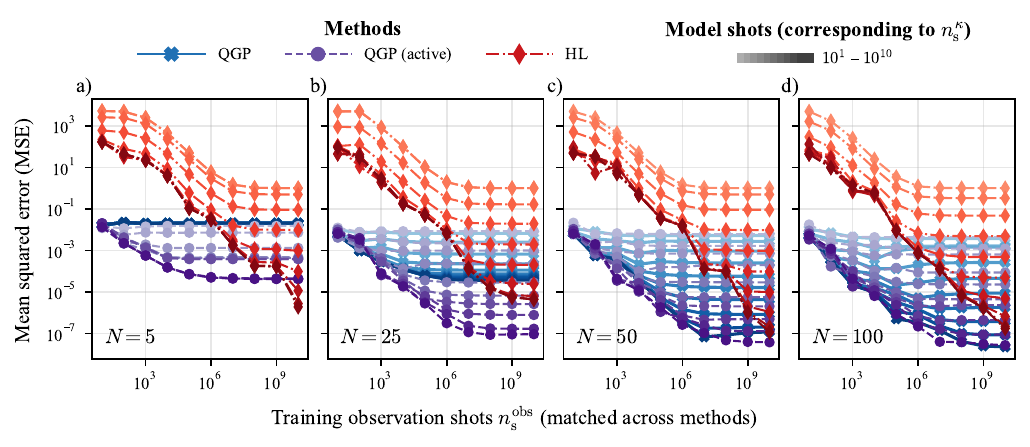}
            \caption{
            \textbf{Benchmarking QGP regression for varying training set sizes $N$ against Hamiltonian learning (HL).} 
            QGPs are evaluated using both a standard fixed-dataset approach and active learning, where the training dataset is built iteratively based on QGP uncertainty quantification. While the mean squared error (MSE) is plotted against the training observation shot budget $\shotsobs$, varying model shot counts (corresponding to the kernel entry shots $\shotsker$) are indicated by color brightness, ranging from $\shotsker = 10$ (bright) to $\shotsker = 10^{10}$ (dark). To guarantee a fair comparison, the total observation shot budget $N\shotsobs$ is evenly distributed to estimate the query outcomes of $U$ in HL. Similarly, the QGP kernel entry shots $\shotsker$ correspond to the shots used to represent the input states in the HL model and are evenly distributed across all states.
            }
            \label{fig:bench:regression:multiples}
        \end{figure}

        After setting the stage for a non-trivial regime, i.e., low MSE with few training data points, for QGP and, in particular, QGP with active learning, we benchmark against HL, an alternative, non-data-driven quantum method for regression (Supp.~Fig. \ref{fig:bench:regression:multiples}).
        As established previously, an MSE on the order of $10^{-2}$ or higher represents trivial performance across all training set sizes. Our results indicate that HL fails to surpass this trivial threshold unless both shot budgets are particularly large ($\shotsobs \geq 10^7$, varying slightly with $N$, and $\shotsker \geq 10^5$). 
        Note that for this comparison, the HL total shot budget is matched to $\shotsobs$ and $\shotsker$ for query outcome estimation of $U$ and state measurements for the model prediction, respectively, and rounded to the next higher integer in favor of HL.

        The only scenario in which HL surpasses QGPs occurs in an unrealistically high shot regime ($\shotsobs = 10^{10}$) when QGP is restricted to the smallest training set size benchmarked ($N=5$). In this extreme scenario, the limited data naturally caps the achievable regression accuracy, causing the active learning QGP to saturate at an MSE on the order of $10^{-4}$. Across all other evaluated scenarios, QGP regression typically outperforms HL substantially. For small training sets, i.e., $N=5$ and $N=25$ in Figs.~\ref{fig:bench:regression:multiples}a and ~\ref{fig:bench:regression:multiples}b, the active learning QGP provides significantly improved results over fixed-dataset QGP regression. This advantage stems from the inherent shot efficiency of QGPs, which is evident not only in low-data regimes but also in their robust performance despite low shot counts ($\shotsobs, \shotsker$) and correspondingly noisy observation and kernel estimates. 
        Furthermore, a fundamental limitation of HL is its strict requirement for query access of $U$. Hence, it can neither learn from a fixed dataset when query access is no longer available, nor is it suitable in settings where the cost of querying $U$ significantly outweighs the cost of measuring the input state. In the latter scenario, QGPs, as a data-driven approach, offer a viable trade-off that limits the queries to $U$ by reducing the number of training points.

        \begin{figure}[htp]
            \centering
            \includegraphics[width=1\linewidth]{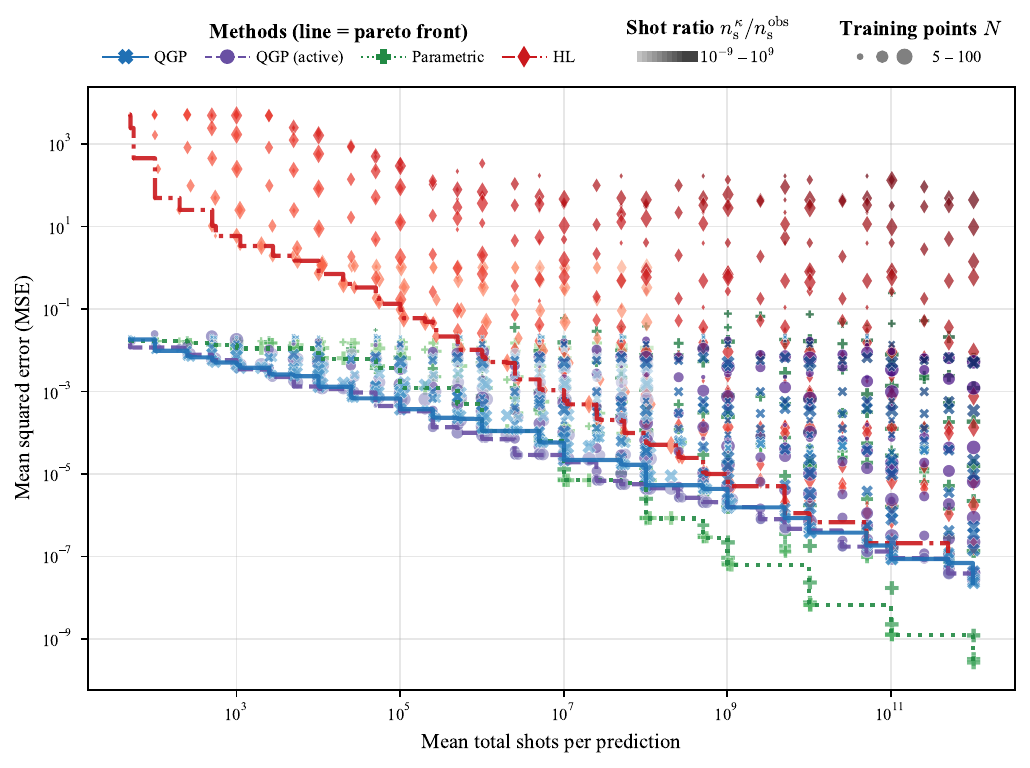}
            \caption{
            \textbf{Pareto front comparison of methods in regression benchmark.} The regression benchmark requires the simultaneous minimization of two quantities: MSE and total measurement shot budget. Consequently, the Pareto front for each method indicates the lowest achievable MSE for a given total shot budget. Therefore, any segment of a Pareto front that lies above and to the right of another represents a strictly sub-optimal regime.
            The total shot budget is divided by the number of test set predictions, which reflects the mean shot budget per prediction (including the training).
            Note that each experiment in this benchmark is presented as an individual marker, where the color brightness indicates the shot ratio $\shotsker/\shotsobs$, ranging from $10^{-9}$ (bright) to $10^{9}$  (dark), and the marker size indicates the number of training points $N$, ranging from 5 (small) to 100 (large). We exclude experiments with $N > 100$ training points, which were previously found trivial.
            }
            \label{fig:bench:regression:pareto}
        \end{figure}

        We conclude the QGP regression benchmark study with a joint comparison of the quantum regression methods in Supp.~Fig.~\ref{fig:bench:regression:pareto}.
        Crucially, we observe a strict separation between the Pareto fronts of HL and the Pareto fronts of QGP (analogously for QGP active learning and the parametric formulation). Therefore, for any given total shot budget, there exists a setting (i.e., shot ratio $\shotsker/\shotsobs$ and number of training points) in which QGP regression (as well as the active learning and parametric variants) outperforms HL in terms of MSE. This consistent Pareto dominance underscores the practical viability of data-driven QGPs over generic HL approaches.

    \subsection{Noisy kernel correction}\label{app:bench:kernel_correction}

        \begin{figure}
            \includegraphics[width=\linewidth]{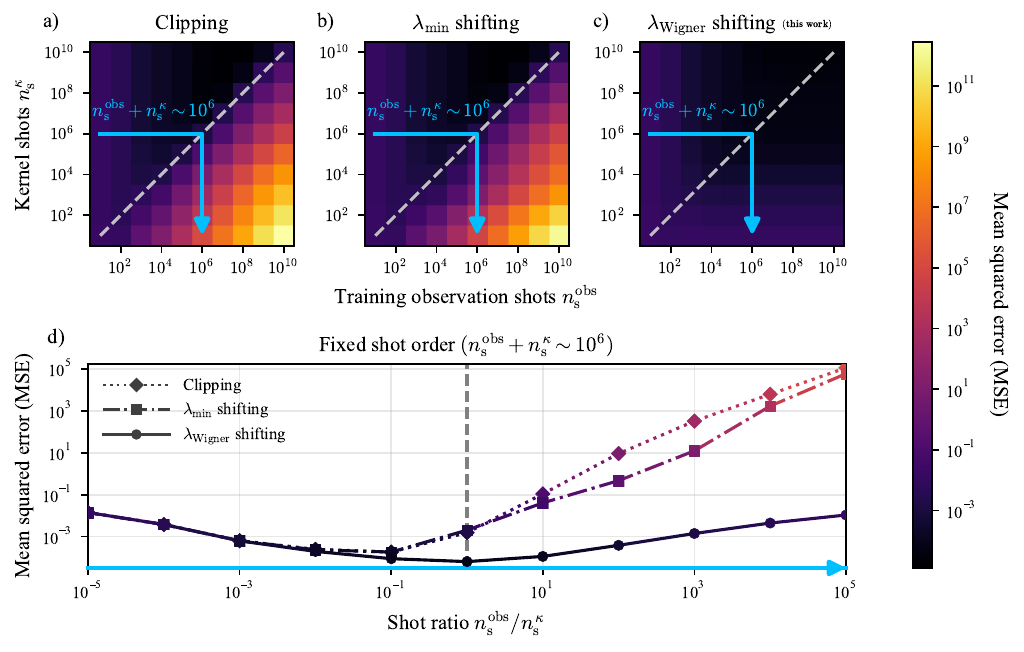}
            \caption{
            \textbf{Robustness of noisy kernel correction methods across measurement shot allocations.} Mean squared error (MSE) for a QGP using different noisy kernel correction methods on regression benchmark task ($N=10$ training points, $n=50$ qubits, $\BC_2$ kernel, as introduced in Supp.~Info.~\ref{app:bench:regression}) evaluated over a logarithmically spaced grid of shots per kernel entry ($\shotsker$) and shots per training observation ($\shotsobs$). 
            (a)--(c) Heatmaps comparing (a) standard negative eigenvalue clipping, (b) $\lambda_{\min}$ shifting, and (c) the proposed $\lambda_{\mathrm{Wigner}}$ shifting. The dashed diagonal line indicates a balanced shot allocation ($\shotsker = \shotsobs$). The solid blue path denotes a fixed total shot budget order $\shotsker + \shotsobs$ of $10^6$. 
            (d) MSE along the fixed-budget slice as a function of the shot ratio. While standard corrections cause an exploding MSE when high observation certainty is paired with noisy kernels (high ratio $\shotsobs / \shotsker \gg 1$), the $\lambda_{\mathrm{Wigner}}$ shift maintains robustness across all regimes and successfully recovers the global minimum error at the balanced allocation (dashed line).
            }
            \label{fig:bench:psd}
        \end{figure}

        Before discussing numerical benchmarking results for the three PSD correction approaches for noisy kernels, we state their definitions and supplement further details. The noisy kernel $\widetilde{K}$ is assumed to exhibit independent and identically distributed (i.i.d.) noise in its (symmetric) entries as $\tilde{k}_{ij} = \tilde{k}_{ji} = k_{ij} + \varepsilon_{ij}$ with $\varepsilon_{ij} \sim \NC(0, \sigma^2_{\kappa})$.

        First, the geometrically optimal correction via negative eigenvalue clipping reads as
        \begin{equation}\label{eq:psd_clipping}
            \bar{K} = \sum_{\lambda_i \geq 0} \lambda_i \bm{u}_i \bm{u}_i^\top,
        \end{equation}
        where $\lambda_i$ and $\bm{u}_i$ are the eigenvalues and eigenvectors of the eigendecomposition $\widetilde{K} = U \Lambda U^\top$.
        
        Second, the minimal uncertainty propagation is achieved through a diagonal shift by the lowest eigenvalue as
        \begin{equation}\label{eq:psd_min_shift}
            \bar{K} = \widetilde{K} - \lambda_{\min} \id\,,
        \end{equation}
        where $\lambda_{\min}$ denotes the lowest eigenvalue $\lambda_{\min} = \min_i \lambda_i$, which must be negative for a non-PSD noisy kernel $\widetilde{K}$.
        
        Finally, the proposed method in this work aims for robust uncertainty propagation by a diagonal shift according to Wigner's semicircle law and is defined as 
        \begin{equation}\label{eq:psd_wigner_shift}
            \bar{K} = \widetilde{K} - \lambda_{\rm Wigner} \id 
            \qquad \text{with} \qquad
            \lambda_{\mathrm{Wigner}} = 2\sqrt{N}\sigma_{\kappa}
            ,
        \end{equation}
        where $N$ denotes the number of training points, i.e., the dimension of the kernel matrices $K$, $\widetilde{K}$ and $\bar{K}$.
        In contrast to the standard correction methods, since the $\lambda_{\rm Wigner}$-shift, can be formulated before measuring the kernel matrix expressed as a diagonal part of the observation noise covariance matrix $R$ in Eq.~\eqref{eq:SI_GP-posterior}, it can be interpreted as part of the data generative prior defined in Eq.~\eqref{eq:SI_GP2}. Although the model should not affect noise in the data generation process, modeling the kernel uncertainty with this generative prior remains a sensible choice. Intuitively, this effective increase in observation noise prior compensates for the model's degraded access to the true relational geometry expressed by the kernel of the training labels.

        We benchmark our proposed noisy kernel correction via $\lambda_{\rm Wigner}$ shifting against standard approaches to identify precise regimes where the choice of method becomes critical. Most importantly, we aim to highlight the conditions under which standard approaches fail, and only $\lambda_{\rm Wigner}$ shifting prevents the QGP prediction from collapsing.
        Supplemental Fig.~\ref{fig:bench:psd} presents these benchmark results for the regression task introduced in Supp.~Info.~\ref{app:bench:regression}, utilizing $\BC_2$ matchgate kernels on $n=50$ qubits and a QGP trained on $N=10$ observations. We systematically vary the measurement shots allocated per kernel entry ($\shotsker$) and per training observation ($\shotsobs$) across a logarithmically spaced equidistant grid from $10^1$ to $10^{10}$.

        The mean squared error (MSE) heatmaps (top) reveal two distinct performance regimes for the standard clipping and $\lambda_{\min}$ shifting corrections. These methods succeed when kernel noise is negligible relative to observation noise, but fail entirely in the opposite regime (bottom-right), where $\shotsker$ is multiple orders of magnitude smaller than $\shotsobs$, causing the latent kernel noise to dominate the observation uncertainty. Beyond the presentation in the figure (for $N = 10$ training points), the exact transition boundary between these phases is found to depend on the number of training points $N$. This dependence is because $\shotsobs$ (and $\shotsker$) are shot numbers per training point (pair), and evaluating the full $N \times N$ kernel matrix requires $N$ times more shots than estimating the $N$ training observations. This scaling directly motivates the $N$-dependent definition of $\lambda_{\rm Wigner}$ in Eq.~\eqref{eq:psd_wigner_shift}.

        This failure manifests as a counter-intuitive behavior: For standard corrections, strictly increasing observation shots $\shotsobs$ while holding $\shotsker$ constant (moving horizontally rightward in the heatmaps) actually increases the MSE. A high $\shotsobs$ implies near-perfect training labels, causing the GP to inherently trust a highly corrupted kernel to fit the ``perfect'' data, leading to fatal overfitting and an exploding MSE. In this regime, the nominal observation uncertainty on the diagonal of $R$ drops orders of magnitude below the kernel noise variance. Because the standard clipping and $\lambda_{\min}$ shifting methods fail to effectively raise $R$ to compensate for this discrepancy, they cannot mitigate the GP failure. Conversely, $\lambda_{\rm Wigner}$ shifting properly injects this uncertainty into $R$, ensuring stable regression and the expected decrease in MSE as observation quality improves.

        This robustness is explicitly demonstrated by the fixed-budget slice (blue line and bottom plot in Supp.~Fig.~\ref{fig:bench:psd}), where shot allocations $\shotsker$ and $\shotsobs$ are constrained to a total order $\shotsker + \shotsobs$ of $10^6$. As the allocation shifts towards high observation certainty ($\shotsobs / \shotsker \gg 1$), clipping (no uncertainty propagation) and $\lambda_{\min}$ shifting (minimal propagation) severely diverge, with the former more than the latter as expected. In contrast, the $\lambda_{\rm Wigner}$ shift maintains a symmetric, robust MSE profile across the entire domain $\shotsobs / \shotsker$, successfully recovering the global minimum error at the balanced allocation ($\shotsker = \shotsobs$), marking the highest possible fixed-order allocation.

        Before concluding, we acknowledge two limitations of $\lambda_{\rm Wigner}$-shift PSD correction. Both arise from the Gaussian orthogonal ensemble (GOE) noise assumption, which invokes Wigner's semicircle law. While it matches the numerical simulation, the assumption is not satisfied in practice. 
        First, the GOE model assumes that the matrix elements (up to symmetry) are i.i.d., which fails to account for the correlations among kernel value estimates arising from the overlaps between input state pairs $\rho_i, \rho_j$.
        Second, the kernel noise does not obey a normal distribution but is rather multinomial and, in particular, heavy-tailed for the matchgate kernels, cf.~Supp.~Info.~\ref{app:ff_kernel_estimation_protocol}.
        However, the violation of the latter assumption behind Wigner's semicircle law may be remedied by the concept of universality in random matrix theory \cite{tao_RandomMatricesUniversality_2010}, which shows convergence to the semicircle distribution for other distributions.
        Ultimately, regardless of the true noise distribution, the specific value of the $\lambda_{\rm Wigner}$-shift should not be viewed as absolute. The broader takeaway is our methodological approach to deriving PSD corrections and a corresponding robust Bayesian prior directly from the physics of the kernel measurement process, rather than relying on a post-hoc kernel correction.

        In the classical machine learning literature, noisy kernels do not typically appear because kernels are deterministically computed up to numerical precision. 
        Conceptually, however, noisy kernel matrices relate closely to GPs with uncertain inputs \cite[ch.~9.5]{rasmussen_GaussianProcessesMachine_2006}. Algebraically, diagonal shifts mirror Tikhonov regularization (kernel ridge regression), acting as an additional regularization term that penalizes model complexity and improves numerical conditioning. Furthermore, this physics-driven approach naturally connects to robust Bayesian analysis by deriving the shift from an expected worst-case noise bound, which effectively acts as a minimax prior \cite{berger_StatisticalDecisionTheory_1985}.

    \subsection{Optimization}\label{app:bench:bo}

        \begin{figure}[tb]
            \centering
            \includegraphics[trim=0 0 0 40, clip,width=1\linewidth]{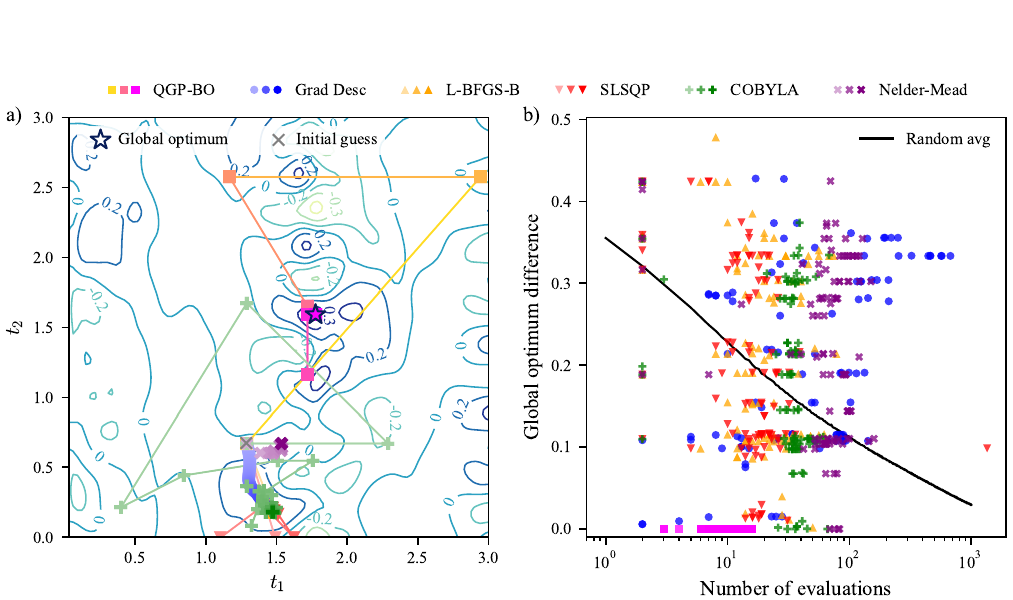}
            \caption{
            \textbf{Validating QGP-BO effectiveness in benchmark against standard optimizers on a synthetic two-dimensional problem.}
            The objective function and QGP model for 30 qubits match those presented in the main text. 
            (a) QGP-BO optimization trajectory visualized in the main text contrasted with example trajectories of the other optimizers, which converge only to local optima.
            (b) Optimization results over 100 independent runs. Regardless of the random initial guess, QGP-BO reaches the global optimum in fewer than 15 evaluations. In contrast, other optimizers cannot outperform a random sampling baseline (black line) on average for an equivalent evaluation budget.
            }
            \label{fig:bo:bench}
        \end{figure}

       Here we consider the same synthetic 2D optimization task used in the main text to benchmark QGP-BO against widely adopted optimization algorithms, including both gradient-based and gradient-free approaches. The gradient-based methods are standard gradient descent, Broyden-Fletcher-Goldfarb-Shanno (BFGS) \cite{broyden1970convergence,fletcher1970new,goldfarb1970family,shanno1970conditioning} in the standard limited-memory and bounded-domain variant, L-BFGS-B \cite{byrd1995limited}, and Sequential Least Squares Programming (SLSQP) \cite{kraft1988software}. The gradient-free methods are Constrained Optimization BY Linear Approximation (COBYLA) \cite{powell1994direct} and the Nelder-Mead simplex method \cite{nelder1965simplex}.
       For the QGP-BO benchmark, the compared optimizers are implemented via \texttt{SciPy} in Python~\cite{virtanen2020scipy} with default hyperparameter settings, except for gradient descent, which was implemented manually. Note that while these implementations intend to minimize, maximization is implicitly achieved by providing negated objective function values instead. Gradient descent uses the optimal fixed step size $1/{L_{\rm min}}$, i.e., the fastest guaranteed convergence, based on the minimal Lipschitz constant $L_{\rm min}$, which was numerically estimated as $45.866$. In practice, the minimal Lipschitz constant is typically not accessible. All gradient-based optimizers rely on numerical gradients via finite differences.

        A direct visual comparison of the trajectories of all employed optimizers from the same initial guess is provided in Supp.~Fig.~\ref{fig:bo:bench}(a). It is striking that all optimizers explore only a small region and settle in a nearby local optimum. BO, in contrast, jumps radically across the feasible region because it relies on the QGP surrogate, which it can analyze globally. For instance, gradient descent is highly local. Hence, gradient descent can only locate the global optimum if initialized in its close vicinity -- about $4 \%$ of the feasible area in the example studied. The ``curse of dimensionality'' renders local optimizers unsuitable for global optimization of such complex objective landscapes in higher dimensions.
        
        For a quantitative comparison, all optimizers are initialized uniformly at random over the feasible region $t_1, t_2 \sim \UC[0, 3]$ over 100 runs. Due to the complex, unstructured nature of the objective landscape, the benchmarking optimizers perform rather inefficiently. Supplemental Fig.~\ref{fig:bo:bench}b reveals this clearly in the bottom panel through the relation between the number of evaluations and deviation from the true maximum $\lvert f(t_{\rm opt}) - f^*\rvert$. Put into perspective, these optimizers cannot beat a naive baseline (cf.~black line) as the average deviation from $f*$ when using the same number of objective evaluations for mere uniform samples. 
        The relative frequencies that these optimizers outperform the naive sampling baseline are $39\%$ (gradient descent), $50\%$ (L-BFGS-B), $57\%$ (SLSQP), $51\%$ (COBYLA), and $31\%$ (Nelder-Mead).
        In contrast, BO always reaches the global optimum in less than 15 evaluations. The global optimum is rarely found by some of the other optimizers and, if so, typically by almost an order of magnitude more evaluations than BO.

        We attribute the superiority of QGP-BO to two key advantages of the QGP surrogate: First, it efficiently retains and reuses information, unlike the black-box optimizers compared here, which exhibit limited information reuse, if any. Second, and most critically, the inherent inductive bias enables the QGP to recover the underlying structure from scarce observations, which is difficult to replicate in non-Bayesian optimization methods.

\clearpage    
    \section{Supplemental numerical experiments and evaluation}

    \subsection{Extrapolation}\label{app:B1_extrapolation}

        Building on the $\BC_1$ free-fermionic QGP extrapolation presented in the Regression section of the main text, we here present a more detailed comparison of the two shot noise sources. Specifically, Supplemental Fig.~\ref{fig:B1_extrapolation:noise_impact} isolates the individual contributions of observation noise and kernel noise, determined by the observation ($\shotsobs$) and kernel shot budget ($\shotsker$), respectively.
        While kernel noise induces microscopic (per-point) jitter in the prediction, both noise sources have a fundamentally similar impact on the overall macroscopic extrapolation behavior: They broaden the predictive uncertainty and reduce the model's reliance on the training data. Consequently, kernel noise acts as an effective observation noise, as the underlying training information cannot be perfectly resolved when the kernel is noisy. We leveraged this interpretation for the kernel PSD correction by extending the data generation prior via the Wigner semicircle law.

        \begin{figure}
            \centering
            \includegraphics[width=1\linewidth]{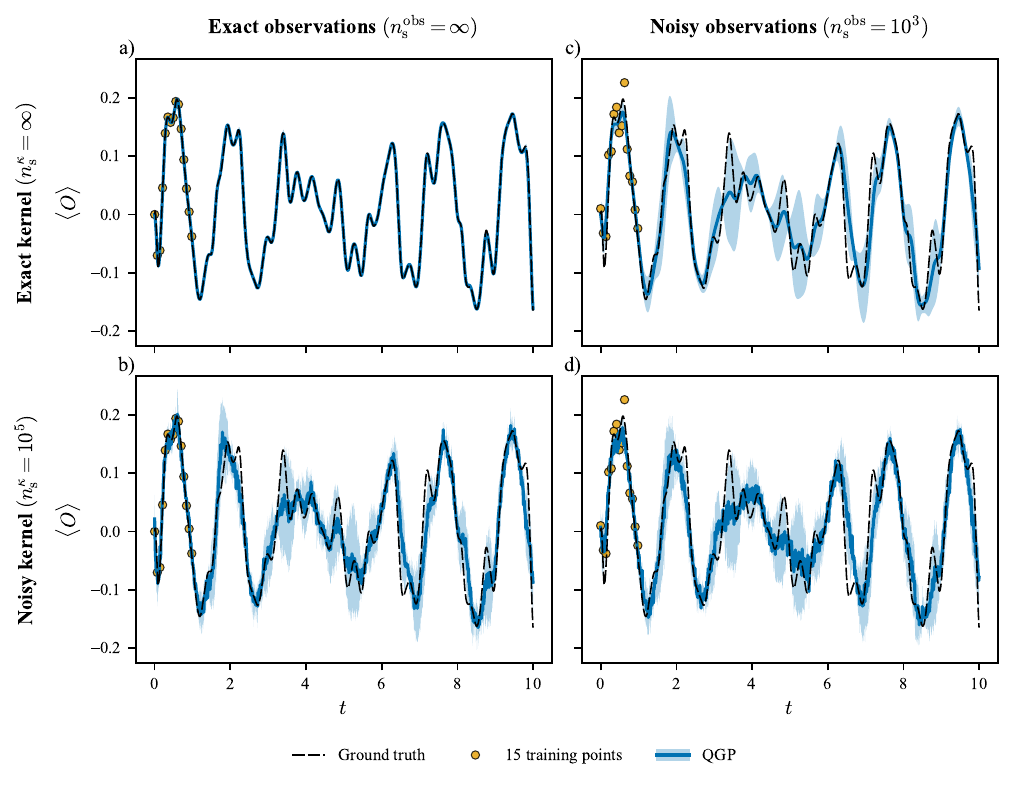}
            \caption{
            \textbf{Ablation of observation and kernel noise in QGP extrapolation.}
            The $2 \times 2$ panel grid isolates the effects of finite-shot noise on the predicted expectation $\braket{O}$ over $t$. The columns contrast exact (left, $\shotsobs = \infty$) with noisy (right, $\shotsobs = 10^3$) observations used for training. The rows contrast exact (left, $\shotsker = \infty$) with noisy (right, $\shotsker = 10^5$) kernel entries for the $\BC_1$ matchgate QGP. (a) and (d) are reproduced from the corresponding main text figure.
            In all panels, the QGP predictive mean and uncertainty (solid blue line and shaded region) are trained on 15 points (orange circles) and compared against the exact ground truth (dashed black line).
            }
            \label{fig:B1_extrapolation:noise_impact}
        \end{figure}

    \subsection{Kernel identification via marginal likelihood}\label{app:marginal_likelihood}

        Maximizing the marginal likelihood to optimize (unknown parameters in) the kernel is known as an empirical Bayes method, which, in contrast to standard Bayesian methods, no longer assumes a fixed prior distribution before any data is observed. Instead, empirical Bayes estimates the prior probability distribution from the data. For instance, while the form or family of the prior is still informed by task knowledge, its concrete instantiation is data-driven. This approach can be highly beneficial in situations where inductive bias is limited. Moreover, the closed-form and efficient access to the marginal likelihood offers a distinct advantage of Gaussian processes over kernel regression. Further details are provided in Sec.~\ref{sec:methods:marginal_likelihood}. In contrast, in kernel regression, the marginal likelihood is not directly available, and other approximate hyperparameter optimization techniques, such as cross-validation, are typically used. 
        
        We demonstrate kernel identification by means of a regression task. Here, we assume that it is known that the unknown transformation $U$ acts unitarily on a single qubit only, i.e., $U \in \mathbb{U}(2)$, followed by an observable $O$ on that qubit, but it is not known which of the $n = 50$ qubits is affected. Therefore, we can consider a non-negative linear combination of 50 local unitary kernels for subsystems of size $m=1$ as
        \begin{equation}
        \kappa(t, t') = \frac{1}{ n} \sum_{i=1}^n \vartheta_i \kappa_i(t, t') \qquad\text{with}\qquad 
        \bm{\vartheta} \in [0, 1]^n
        \quad \text{and} \quad
        \kappa_i(t,t') = \frac{1}{4} \lVert O_i\rVert^2
        \Tr\left[ \rho_i(t) \rho_i(t') \right]\,,
        \end{equation}
        which is guaranteed to form a valid kernel, i.e., it is PSD\footnote{While we formulate a constrained optimization problem ($[0,1]$-box constraints), kernel factors can be made non-negative by construction via $\vartheta_i^2$ (and a renormalization factor $1/\lVert \bm{\vartheta} \rVert^2$ may be introduced to limit their magnitude) to solve the marginal likelihood maximization as an unconstrained optimization problem.} for any non-negative choice of coefficients $\bm{\vartheta}$. 
        The subscript notation $\rho_i$ denotes the reduced density operator in the subsystem of the $i$-th qubit, with the corresponding reduced norm of the single-qubit Pauli (Pauli-$Z$ here) observable is $ \lVert O_i\rVert^2 = 2$.
        The quantum data states here are again defined as in the synthetic dataset used for the extrapolation in the matchgate circuit and $\BC_1$ observable setting (see the Methods section of the main text) but with a single preparation layer only.
        
        Supplemental Fig.~\ref{fig:marginal_likelihood} shows the effectiveness of optimizing the kernel combination parameters $\bm{\vartheta}$. Starting with a uniform initialization $\vartheta_i = 1$ for $i=1, \ldots,n$ reflects no a priori preference for any particular qubits. As evident from Supp.~Fig.~\ref{fig:marginal_likelihood}a, this inductive bias in the initial kernel leads to a weak QGP prior, insufficient to solve the extrapolation task from 15 training data points within $t\in [0, 1]$ to the larger region of $t\in [0, 10]$, as evidenced by Supp.~Fig.~\ref{fig:marginal_likelihood}a. The marginal likelihood maximization, however, identifies the correct local kernel by sending all other parameters in $\bm{\vartheta}$ to zero. As Supp.~Fig.~\ref{fig:marginal_likelihood}b clearly presents, the optimized kernel now suffices for the QGP to achieve perfect extrapolation.
        Note that throughout the marginal likelihood optimization of the combination parameters, no additional executions on the quantum device are required, as all calculations can be performed classically using the pre-evaluated local kernel matrices. We used the L-BFGS-B optimizer \cite{broyden1970convergence,fletcher1970new,goldfarb1970family,shanno1970conditioning,byrd1995limited} to enforce the box constraints, utilizing exact gradients and exact expectation values. Nevertheless, this framework can, in principle, be directly extended to incorporate finite-shot noise during marginal likelihood maximization, in which stochastic optimizers (such as Adam \cite{kingma2015adam} in the unconstrained formulation) are highly suitable.
        
        \begin{figure}
            \centering
            \includegraphics[width=\linewidth]{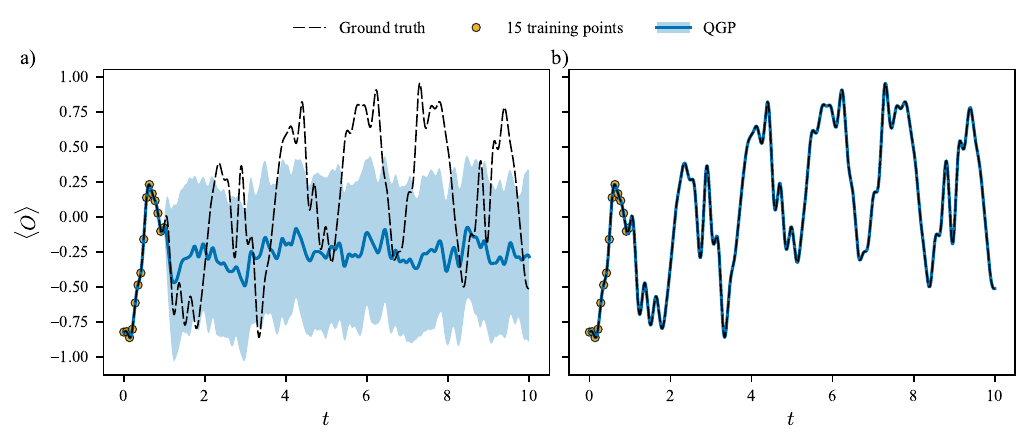}
            \caption{\textbf{Marginal likelihood optimization of combination coefficients of local unitary kernels.} An unknown (single qubit) local unitary transformation with Pauli measurement on corresponding qubits is assumed. 
            (a) Since there is no preference towards any of the $n=50$ qubits, the initial kernel is the average (uniform convex combination) over all local unitary kernels. As this initial kernel resembles only a weak prior, the QGP extrapolation capability is limited. 
            (b) After maximizing the marginal likelihood, which sends all kernel combination coefficients to zero except for the qubit the unknown unitary transformation (as well as the measurement) acts on, perfect extrapolation from the 15 training points to large timescales is achieved.
            The results presented are in the absence of shot noise.}
            \label{fig:marginal_likelihood}
        \end{figure}

    \clearpage
    \section{Extended quantum Gaussian process methodology}

        Due to their probabilistic nature, QGPs have the following key methodological advantages over, e.g., (kernel) ridge regression, which is equivalent to the QGP predictive mean but lacks uncertainty quantification and, hence, does not describe a distribution over possible solutions.

        \subsection{Active learning}\label{sec:methods:active_learning}

            In a scenario where we are not provided with a fixed quantum dataset but have access to acquire one, it is advantageous to actively select the $N$ training points that maximize model performance rather than relying on a predetermined dataset. Active learning achieves this by iteratively expanding the training set, sequentially querying the system at the input $t_{\rm next}$ that exhibits the highest predictive variance under the current QGP model:
            \begin{equation}
                t_{\rm next} = \underset{t}{\arg\max} \; \sigma_{\rm P}^2(t).
            \end{equation}
            While other selection strategies exist \cite{mackay1992information}, targeting the maximum variance is robust and straightforward: by continually querying the system where the model is most uncertain, we maximize information gain and rapidly refine the predictions.

        \subsection{Marginal likelihood in closed-form}\label{sec:methods:marginal_likelihood}

        The proposed QGP framework translates physical system assumptions into quantum kernels, naturally encoding inductive biases. While the table in the main text presents specific kernel definitions tailored to various groups of the evolution and observable types, prior knowledge of a given task may be insufficient to uniquely determine the optimal kernel. Instead, one might only be able to narrow the choice down to a set or combination of such candidate kernels. While non-probabilistic models typically require computationally expensive exhaustive search and cross-validation techniques, kernel and, more broadly, model selection within the QGP framework can be performed efficiently by maximizing the marginal likelihood. Concretely, we can optimize the non-negative coefficients of a linear combination of possible quantum kernels. More sophisticated compositional searches involving the addition and multiplication of base kernels can also be employed \cite{duvenaud2013structure}.

        The marginal likelihood quantifies the probability of the training observations $\bm{y}_{\TC}$ at inputs $\TC$ under a specific kernel choice $\kappa$, marginalized over the unobserved, noiseless function values $\bm{f}_\TC$:
        \begin{equation}\label{eq:marginal_likelihood_general}
            p(\bm{y}_{\TC} \mid \TC, \kappa) = \int p(\bm{y}_{\TC} \mid \bm{f}_{\TC}) \, p(\bm{f}_{\TC} \mid \TC, \kappa) \, d\bm{f}_{\TC}.
        \end{equation}
        The two distributions in the integrand correspond to the observation noise model and the QGP prior, established in Eqs.~\eqref{eq:SI_GP2} and \eqref{eq:SI_GP1}, respectively. Because we assume (i.i.d.) Gaussian observation noise, all distributions in Eq.~\eqref{eq:marginal_likelihood_general} are  Gaussian, resolving into the following closed-form expression for the log\footnote{
        Maximizing the logarithm of the marginal likelihood is standard practice to prevent numerical underflow when multiplying many potentially small probability values. Strict monotonicity of this transformation exactly preserves the locations of the optima.} marginal likelihood:
        \begin{equation}
            \log p(\bm{y}_\TC \mid \TC, \kappa) = -\frac{1}{2} \bm{y}_\TC^\top (K_{\TC\TC} + R)^{-1} \bm{y}_\TC - \frac{1}{2} \log |K_{\TC\TC} + R| - \frac{N}{2} \log 2\pi,
        \end{equation}
        For the sake of brevity, we omit $\bm{\mu}_\TC$, , i.e., typical zero-mean prior $\bm{\mu}_\TC = \bm{0}$.
        Provided the kernel $\kappa$ is differentiable with respect to its continuous (hyper-) parameters, such as the linear combination coefficients, gradient-based optimizers can quickly find accurate solutions that (at least locally) maximize the marginal likelihood. For a comprehensive derivation, we refer the reader to Refs.~\cite{rasmussen_GaussianProcessesMachine_2006,murphy_ProbabilisticMachineLearning_2023}.

\end{document}